\documentclass[onefignum,onetabnum]{siamonline190516}

\usepackage{lipsum}
\usepackage{amsfonts}
\usepackage{graphicx}
\usepackage{epstopdf}
\usepackage{algorithmic}
\usepackage[caption=false,font=footnotesize]{subfig}
\ifpdf
  \DeclareGraphicsExtensions{.eps,.pdf,.png,.jpg}
\else
  \DeclareGraphicsExtensions{.eps}
\fi
\usepackage{tikz}
\usetikzlibrary{positioning}
\usepackage{enumitem}
\setlist[enumerate]{leftmargin=.5in}
\setlist[itemize]{leftmargin=.5in}

\usepackage[square,numbers]{natbib}

\newsiamremark{remark}{Remark}
\newsiamremark{hypothesis}{Hypothesis}
\crefname{hypothesis}{Hypothesis}{Hypotheses}
\newsiamthm{claim}{Claim}
\newsiamthm{application}{Application}
\newsiamthm{example}{Example}
\newsiamremark{assumption}{Assumption}

\usepackage{multirow}

\headers{Separation of linear-quadratic mixtures}{C. Kervazo, N. Gillis, and N. Dobigeon}

\title{Provably robust blind source separation of \\ linear-quadratic near-separable mixtures\thanks{\funding{CK and NG acknowledge the support by the European Research Council (ERC starting grant no 679515), and NG by 
	the Fonds de la Recherche Scientifique - FNRS and the Fonds Wetenschappelijk Onderzoek - Vlanderen (FWO) under EOS Project no O005318F-RG47. ND is partly supported by the AI Interdisciplinary
	Institute ANITI funded by the French ``Investing for the Future – PIA3'' program under
	the Grant agreement number ANR-19-PI3A-0004.}}}

\author{Christophe Kervazo\thanks{T\'{e}l\'{e}com Paris, Institut Polytechnique de Paris, France
  (\email{christophe.kervazo@telecom-paris.fr}).}
\and Nicolas Gillis\thanks{University of Mons, Mons, Belgium 
  (\email{nicolas.gillis@umons.ac.be}, \email{https://sites.google.com/site/nicolasgillis/}).}
\and Nicolas Dobigeon\thanks{University of Toulouse, IRIT/INP-ENSEEIHT, Toulouse, France
	(\email{nicolas.dobigeon@enseeiht.fr}, \email{http://dobigeon.perso.enseeiht.fr/fr/}).}}

\usepackage{amsopn}

\ifpdf
\hypersetup{
	pdftitle={A provably robust algorithm for blind source separation of linear-quadratic near-separable mixtures},
	pdfauthor={C. Kervazo, N. Gillis, and N. Dobigeon}
}
\fi

\DeclareMathOperator*{\rowrank}{rowrank}
\DeclareMathOperator*{\argmin}{argmin}
\DeclareMathOperator*{\argmax}{argmax}

\begin{document}
	\newcommand{\norm}[2]{\left\Vert  {#1} \right\Vert _{#2}}
	\newcommand*{\Ob}{\mathcal{O}_{b}}
	\newcommand{\ckc}[1]{{\color{black} #1}}
	\newcommand*{\ps}{\Pi_2}
	\newcommand*{\psnd}{%
		\begin{tikzpicture}[baseline=-0.75ex]
		\draw (0,0) circle [radius=5pt];  
		\fill (0,0) circle [radius=1pt];  
		\end{tikzpicture}%
	}
	\newcommand*{\psndold}{\textcircled{\textbullet}}
	\newcommand*{\psd}{\Pi_{\odot \setminus I}}
	\newcommand*{\pspr}{\Pi_{\odot \times}}
	\newcommand*{\pstr}{\Pi_{4}}
	\newcommand*{\lbr}{[\![ }
	\newcommand*{\rbr}{]\!]}
	\newcommand*{\snpab}{SNPALQ}
	\newcommand*{\virtual}{virtual}
	\newcommand{\resi}[2]{\mathcal{R}^f_{#1}\left( #2 \right)}
	\newcommand*{\epsb}{\check{\epsilon}}
	
	\newcommand{\nd}[1]{{{\color{red} #1}}}	
	\newcommand{\ndr}[1]{{\color{red} (\textbf{ND:} #1)}}	
	\definecolor{brightpink}{rgb}{1.0, 0.0, 0.5}
	\newcommand{\ngc}[1]{{\color{brightpink} (\textbf{NG:} #1)}}
	\newcommand{\ngi}[1]{{{\color{brightpink} #1}}}	
	\newcommand{\ckr}[1]{{\color{blue} (\textbf{CK:} #1)}}
	\newcommand{\nsl}[1]{\tilde{#1}}
	
	\newcommand{\myparagraph}[1]{\noindent {\color{header1}{\textbf{#1}}} --}

	\maketitle
	
	\begin{abstract}
	In this work, we consider the problem of blind source separation (BSS) by departing from the usual linear model and focusing on the  linear-quadratic (LQ) model. 
	We propose two provably robust and computationally tractable algorithms to tackle this problem under separability assumptions which require the sources to appear as samples in the data set.  
		The first algorithm generalizes the successive nonnegative projection algorithm (SNPA), designed for linear BSS, 
		and is referred to as SNPALQ. 
		By explicitly modeling the product terms inherent to the LQ model along the iterations of the SNPA scheme, the nonlinear contributions of the mixing are mitigated, thus improving the separation quality. SNPALQ is shown to be able to recover the ground truth factors that generated the data, even in the presence of noise. 
		The second algorithm is a brute-force (BF) algorithm, which is used as a post-processing step for SNPALQ. It enables to discard the spurious (mixed) samples extracted by SNPALQ, thus broadening its applicability. The BF is in turn shown to be robust to noise under easier-to-check and milder conditions than SNPALQ. We show that  SNPALQ with and without the BF postprocessing  is relevant in realistic numerical experiments. 
	\end{abstract}
	
	
	
	\section{Introduction}\label{sec:intro}
	Blind source separation (BSS) \cite{comon2010handbook,bobin2015sparsity,kervazo2018PALM} is a powerful paradigm with a wide range of applications such as remote sensing \cite{Schaepman2009}, biomedical and pharmaceutical imaging \cite{Akbari2011,Rodionova2005}, and astronomy \cite{Themelis2012}.  
	BSS aims at decomposing a given data set into a set of unknown elementary signals to be recovered, generally referred to as the  \emph{sources}. 
	Because it is simple and easily interpretable, many works \cite{comon2010handbook} have focused on the \emph{linear} mixing model (LMM) which assumes that the $i$th data set sample $\bar{\mathbf{x}}_i \in \mathbb{R}^m$ for $i \in \lbr n \rbr$ can be written as 
	\begin{equation*}
		\bar{\mathbf{x}}_i = \sum_{k=1}^{r}  h_{ki} \mathbf{w}_k + \mathbf{n}_i,
	\end{equation*}
	where $\mathbf{w}_k$  
	is the $k$th source for $k \in \lbr r \rbr = \{1,2,\dots,r\}$, 
	and $h_{ki}$ its the associated mixing coefficient in the $i$th (mixed) observation. The vector $\mathbf{n}_i$ accounts for any additive noise and/or slight mismodelings in the $i$th pixel. 
	Using a standard matrix formulation, the LMM can thus be rewritten as 
	\[
	\bar{\mathbf{X}} = \mathbf{W}\mathbf{H} + \mathbf{N}, 
	\] 
	where 
	 $\bar{\mathbf{X}} = [\bar{\mathbf{x}}_1,\bar{\mathbf{x}}_2,...,\bar{\mathbf{x}}_m]\in \mathbb{R}^{m\times n}$ is the data set,  
	$\mathbf{W} = [\mathbf{w}_1,\mathbf{w}_2,...,\mathbf{w}_r] \in \mathbb{R}^{m\times r}$ are the sources, 
	 $\mathbf{H} \in \mathbb{R}^{r\times n}$ is the mixing matrix containing the coefficients $h_{ki}$'s,  
	and $\mathbf{N} = [\mathbf{n}_1,\mathbf{n}_2,...,\mathbf{n}_m] \in \mathbb{R}^{m\times n}$ is the noise.  
	We denote by $\mathbf{X = WH}$ the noiseless version of $\bar{\mathbf{X}}$. 
	
	The goal of BSS is to recover $\mathbf{W}$ and $\mathbf{H}$ from the sole knowledge of $\bar{\mathbf{X}}$. This is in general an ill-posed problem~\cite{comon2010handbook}. Hence, in most works, additional constraints are imposed on the unknown matrices $\mathbf{W}$ and $\mathbf{H}$ to make the problem better posed: for instance, orthogonality in principal component analysis (PCA -- \cite{jolliffe1986principal}), independence in independent component analysis (ICA -- \cite{comon2010handbook}), and sparsity in sparse component analysis (SCA -- \cite{zibulevsky2001blind,bobin2015sparsity,kervazo2018PALM}). 
	We will here focus on nonnegativity constraints, akin to nonnegative matrix factorization (NMF)~\cite{lee1999learning}. Although NMF is NP-hard in general~\cite{vavasis2010complexity}, and its solution non-unique~\cite{xiao2019uniq}, Arora et al.~\cite{AGKM11,Arora2016} have introduced the subclass of near-separable non-negative matrices for which NMF can be solved in a polynomial time with weak indeterminacies. This subclass corresponds to data sets in which each source appears purely in at least one data sample. Building on near-separable NMF, several provably robust algorithms have been proposed \cite{AGKM11,Esser2012convex,Recht2012,gillis2014robust}. Among them, one can cite the successive projection algorithm (SPA)  \cite{Araujo01}, which is a fast greedy algorithm provably robust to noise~\cite{Gillis_12_FastandRobust}, or an enhanced version, the successive nonnegative projection algorithm (SNPA) \cite{Gillis2014}, which is more efficient when $\mathbf{W}$ is ill-conditioned and is applicable when  $\mathbf{W}$ is rank-deficient.

	\subsection{LQ mixing model} 
	In various applications, 
	the LMM may however suffer from some limitations and can only be considered as a first-order approximation of non-linear mixing models \cite{Bioucas-Dias2012,dobigeon2014nonlinear,Dobigeon2016}. 
	In such situations, linear-quadratic (LQ) \cite{Deville2019} models can for instance better account for the physical mixing processes by including termwise products of the sources \cite{Dobigeon2014,Heylen2014}. This model can be written as
	\begin{equation}
	\bar{\mathbf{x}}_i = \sum_{k=1}^{r} h_{ki} \mathbf{w}_k + \sum_{p=1}^{r}\sum_{l=p}^{r} \beta_{ipl} (\mathbf{w}_p\odot \mathbf{w}_l) + \mathbf{n}_i.
	\label{eq:LQ}
	\end{equation}
	In \eqref{eq:LQ}, the linear contribution associated to LMM is complemented by a set of second-order interactions $\mathbf{w}_p\odot \mathbf{w}_l$ between the sources, where $\odot$ denotes the Hadamard product and $\beta_{ipl}$ is the amount of the interaction $\mathbf{w}_p\odot \mathbf{w}_l$ within the $i$th observation. 
	It is worth mentioning the closely-related so-called bilinear mixing model \cite{dobigeon2014nonlinear,Deville2019}, which is a particular instance of the LQ mixing model, from which the \emph{squared} terms $\mathbf{w}_p \odot \mathbf{w}_p$ for $p \in \lbr r\rbr$ in \eqref{eq:LQ} are removed; see  Application~\ref{application:hs} below for a discussion in the context of blind hyperspectral unmixing where the LQ and bilinear models are widely used. 
	
	The LQ mixing model~\eqref{eq:LQ} can also be rewritten in a matrix form 
	\begin{equation}
	\bar{\mathbf{X}} = \ps(\mathbf{W}){\mathbf{H}} + \mathbf{N}
	\label{eq:LQMat}
	\end{equation} 
	where $\ps(\mathbf{W})\in \mathbb{R}^{m \times \tilde{r}}$ is the {extended} source matrix containing the sources and their second-order products as its columns, 
	with $\tilde{r}={r(r+3)}/{2}$, 
	and ${\mathbf{H}} \in \mathbb{R}^{\tilde{r} \times n}$ is the matrix gathering all the mixing coefficients associated with the linear ($h_{ki}$'s) and nonlinear ($\beta_{ipl}$'s) contributions. 
	Written in such a matrix form, the similarity between the LQ and linear models is easily visible: the LQ mixings can be written in a linear form by considering the quadratic terms $\mathbf{w}_p \odot \mathbf{w}_l$ as new sources, additional to the usual ones $\mathbf{w}_k$. Following this line of thought, the $\mathbf{w}_p \odot \mathbf{w}_l$ terms are often called \emph{\virtual\ sources}. In the sequel of this paper, this terminology will be adopted and the non-virtual sources $\mathbf{w}_i$ will be referred to as \emph{primary}.

	\begin{application}[Hyperspectral imaging] \label{application:hs}
		To illustrate the BSS of LQ-mixtures (LQ-BSS), we consider throughout this paper the example of hyperspectral (HS) imaging. Despite having a finer spectral resolution than conventional natural images, HS images generally suffer from a limited spatial resolution. Therefore, several materials are generally present in each pixel, and thus the acquired spectra correspond to mixtures of the different pure material spectra, called endmembers. This mandates the use of BSS methods -- more specifically of NMF -- to perform {\it spectral unmixing}. To be more precise, using the terminology of HS unmixing \cite{dobigeon2014nonlinear}, $\mathbf{w}_k$ in (\ref{eq:LQ}) corresponds to the spectral signature of the $k$th endmember and $h_{ki}$ 
		to the abundance of the $k$th endmember 
		in the $i$th pixel. The spectral signature of a source is the fraction of light reflected by that source depending on the wavelength, and hence $0 \leq \mathbf{w}_k \leq 1$ for $k \in \lbr r \rbr$. 
		Concerning the model choice, the linear BSS model is often a too rough approximation in HS: in particular, when the light arriving on the sensor interacts with several materials, nonlinear mixing effects may occur \cite{Bioucas-Dias2012,dobigeon2014nonlinear,Dobigeon2016}.
		Specifically, this is often the case when the scene is not flat, for instance in the presence of large geometric structures, such as in urban \cite{meganem2014linear} or forest \cite{Dobigeon2014} scenes. 
		In such a context, it has been shown \cite{Dobigeon2014,Heylen2014} that LQ models enable to better account for multiple scatterings. While it is  further possible to include higher-order terms, most of the works neglect the interactions of order larger than two since they are expected to be of significantly lower magnitudes \cite{Altmann2012,Meganem2013} as $0 \leq \mathbf{W} \leq 1$.\\
	\end{application}

	\subsection{Identifiability issue in LQ-BSS}

	Despite source identifiability issues in the general context of non-linear BSS problems  \cite{comon2010handbook,deville2015overview,kervazo2019nonlin}, it was recently showed \cite{Deville2019} that the non-linearity inherent to bilinear mixtures leads to an \emph{essentially unique} solution in the noiseless case. More precisely, it was shown that for a data matrix $\mathbf{X}$ following the bilinear model in the absence of noise (and under some appropriate assumptions, see below),   
	any $\hat{\mathbf{W}}$ and $\hat{\mathbf{H}}$ such that $\mathbf{X} = \ps(\hat{\mathbf{W}})\hat{\mathbf{H}}$  satisfy $\hat{\mathbf{W}} = \mathbf{W}$ and $\hat{\mathbf{H}} = \mathbf{H}$ up to a scaling and permutation of the columns of $\hat{\mathbf{W}}$ and the rows of $\hat{\mathbf{H}}$.   
	However, this identifiability result suffers from some limitations: 
		\begin{itemize}
			\item It relies on two strong assumptions: 
			
			\begin{enumerate}
			    \item $\rowrank(\mathbf{X}) = \frac{r(r+1)}{2}$, requiring that $\hat{\mathbf{H}}$ has full row rank and hence that every extended source is present in the data set. In other words,  all possible interactions of two primary sources must be present in some observation. This is unlikely to happen in practice. 
			    
			    \item the products of the sources up to order four must be linearly independent. It requires the family 
			\begin{equation}
			\small
			\left( 
			\mathbf{W}, (\mathbf{w}_i\odot \mathbf{w}_j)_{\substack{i,j\in \lbr r\rbr \\ j < i}}, (\mathbf{w}_i \odot \mathbf{w}_j \odot \mathbf{w}_k)_{\substack{i,j,k\in \lbr r\rbr \\ k < j < i}},(\mathbf{w}_i \odot \mathbf{w}_j \odot \mathbf{w}_k \odot \mathbf{w}_l)_{\substack{i,j,k,l \in \lbr r\rbr \\ l < k < j < i}}\right),
			\label{eq:linInd}
			\end{equation}
			to be linearly independent. As its size is $\frac{r(r+1)}{24}\big((r-1)(r-2) + 12\big)$, such a linear independence assumption might not be  satisfied in real-world scenarios, since the number of observations $m$ must be of order $\Theta(r^4)$. 
			\end{enumerate}

			\item It does not apply to mixings with squared terms \cite[section 7]{Deville2019}, that is, LQ mixings instead of bilinear ones. 

			\item No guarantee is given in the presence of noise. Moreover, 
			finding an exact factorization $\ps(\hat{\mathbf{W}})\hat{\mathbf{H}}$ of $\mathbf{X}$ is a difficult problem. The algorithm used in \cite{Deville2019} is a heuristic and does not find an exact solution (see \cite[Fig.~4]{Deville2019}), leading to errors on the recovered sources. 
		\end{itemize}
	
	\begin{application}[Hyperspectral imaging (cont'd)]
		In HS imaging, the assumption that  $\mathbf{H}$ has full row rank  is unlikely to be satisfied as many endmembers do not interact, because they are located far apart in the image.  
		
		For the second assumption, 
		even with $r=10$ endmembers, which is a relatively small number, \emph{at least} $m \geq 385$ spectral bands would be required to ensure the linear independence of the family (\ref{eq:linInd}). 
		This is not satisfied for typical HS sensors dedicated to Earth observation. As an example, the Airborne Visible / Infrared Imaging Spectrometer (AVIRIS) operated by the Jet Propulsion Laboratory (JPL, NASA), acquires HS images composed of $m=224$ spectral bands, among them several dozens are inexploitable due to low signal-to-noise ratios.  
	\end{application}
	
	\subsection{Near-separable LQ mixings}\label{sec:contraintsOnParameters}
	To overcome the above identifiability issues, we propose in this work to tackle BSS problems of the form~\eqref{eq:LQ} under a near-separable NMF-like paradigm. In particular, the rationale is to convert the linear independence condition on the family (\ref{eq:linInd}) into a non-negative independence condition, which is significantly less restrictive. 
	Consider for instance the family of points located on a circle within the unit simplex in three dimensions, that is, distinct points within the set 
$\{\mathbf{x} \in \mathbb{R}^3_+ \ | \  \| \mathbf{x} \|_1 = 1 , \| \mathbf{x} \|_2 = q \}$ for some $q < 1$. Although the rank of this family is 3, no point is within the convex cone of other points, and hence this family is non-negatively independent. 
	
	More specifically, denoting $\Delta = \{\mathbf{x} \in \mathbb{R}^{\tilde{r}} | x \geq 0, \sum_{i = 1}^{\tilde{r}} x_i \leq 1\}$ and $\ps(\mathbf{W})_{\setminus \{j\}}$ the submatrix of $\ps(\mathbf{W})$ excluding $\mathbf{w}_j$, we assume the following constraints:
	\begin{align} 
		&h_{ki} \geq 0 \text{ for all } i \in \lbr n \rbr \text{ and } k \in \lbr \tilde{r} \rbr \text{ (nonnegativity condition)},\nonumber \\
		&\sum_{k=1}^{\tilde{r}} h_{ki} \leq 1 \text{ for all } i \in \lbr n \rbr \text{ (sum to at most one condition)}, \label{eq:constraints} \\
		&\alpha_{2}(\mathbf{W}) = \min_{j\in \lbr r \rbr} \min_{\mathbf{x}\in \Delta} \norm{\mathbf{w}_j - \ps(\mathbf{W})_{\setminus \{j\}}\mathbf{x}}{2} > 0 \text{ (order-2 }\alpha\text{-robust simplicial}). \nonumber 
	\end{align}
	The two first constraints ensure the mixing coefficients for each pixel to be nonnegative and to sum to at most one, and can be equivalently written as $\mathbf{h}_i \in \Delta$ for all $i \in \lbr \tilde{r} \rbr$. The last one ensures that no source lies within the convex hull formed by the other ones, their second order product and the origin. It is thus an extension of the $\alpha$-robust simplicial\footnote{The denomination \textquotedblleft $\alpha$-robust simplicial\textquotedblright\ is slightly abusive here, as the coefficients of $\mathbf{x}$ sum to \emph{at most} one, in contrast to \cite{Arora2016} in which they sum to exactly one.} definition of~\cite{Arora2016} which requires that $\alpha_{1}(\mathbf{W}) = \min_{j\in \lbr r \rbr} \min_{\mathbf{x}\in \Delta} \norm{\mathbf{w}_j - \mathbf{W}_{\setminus \{j\}}\mathbf{x}}{2} > 0$.

	In addition, extending the subclass of $r$ near-separable mixings of \cite{Gillis_12_SparseandUnique} to the LQ model, we will assume the mixing to be $r$-LQ near-separable, as defined below. 
	
	\begin{definition} \label{ass:rLSsep}
		The matrix $\bar{\mathbf{X}}$ is said to be $r$-LQ near-separable if it can be written as:
		\begin{equation*}  
			\bar{\mathbf{X}} = \ps(\mathbf{W})
			\underbrace{
				\left[\begin{array}{cc}
					\begin{array}{c}
						\mathbf{I}_r\\
						\mathbf{0}_{\frac{r(r-1)}{2} \times r}
					\end{array}
					& \mathbf{H'}
				\end{array}\right]
				\mathbf{P}
			}_{\mathbf{H}}
			+ \mathbf{N},
		\end{equation*}
		where $\mathbf{W} \in \mathbb{R}^{m\times r}$ is order-2 $\alpha$-robust simplicial, $\mathbf{I_r}$ is the $r$-by-$r$ identity matrix,  $\mathbf{0}_{p \times q}$ is the $p$-by-$q$ matrix of zeros, $\mathbf{P}$ is a permutation matrix, and $\mathbf{H}' \in \mathbb{R}^{\tilde{r} \times m - r}$ is a matrix satisfying the sum to at most one and nonnegativity conditions. It is important to note that contrary to the sources $\mathbf{W}$, the virtual sources $(\mathbf{w}_p\odot\mathbf{w}_l)_{p,l \in \lbr r \rbr, l\leq p}$ are not required to appear in some samples. 
	\end{definition}

	\begin{application}[Hyperspectral imaging (cont'd)] 
	    It  has  been  shown \cite{dobigeon2014nonlinear} that  bilinear and LQ  models  enable  to better account  for  multiple  scatterings. Examples of such models include the Fan model~\cite{Fan2009}, the generalized bilinear model \cite{Halimi2011}, 
    	the polynomial post-nonlinear model \cite{Altmann2012}; see~\cite{dobigeon2014nonlinear} and the references therein for more details. In this work, we will focus on the so-called Nascimento model~\cite{Nascimento2009,Somers2009}, which is a bilinear-based model that naturally extends the classical linear model and the sum-to-at-most one constraint on the abundances.
    	

The near-separable assumption in HS is referred to as the pure-pixel assumption, as it requires each endmember to appear at least once purely within a pixel. This hypothesis is common and realistic \cite{Gillis_12_FastandRobust,Ma2013}, provided that the spatial resolution is not too low. 
	\end{application}

	\subsection{Contributions}
	In this paper, we introduce two algorithms which, given a $r$-LQ near separable mixture  (Definition~\ref{ass:rLSsep}), approximately recovers the factors $\mathbf{W}$ and $\mathbf{H}$. As such, our results are  
	(i)~theoretical: we show the identifiability of this problem even in the presence of noise,  and (ii)~practical: in contrast to \cite{Deville2019}, the two algorithms run in polynomial time. More specifically, the contributions -- graphically summarized in Figure~\ref{fig:recap_contributions} -- are the following: 
	\begin{itemize}
	
		\item We introduce the successive nonnegative projection algorithm for linear-quadratic 
		mixtures (\snpab), which generalizes SNPA \cite{Gillis2014} to linear-quadratic (LQ) mixings by explicitly modeling the presence of quadratic products within its greedy search process.
		
		\item The conditions under which \snpab\ is provably robust to noise are detailed in Section~\ref{sec:robustness_SNPAB}. In particular, such conditions encompass the linear case (see Section~\ref{sec:robustness_linear}), which is important as the LQ model we consider generalizes the linear one.
		
		\item To further mitigate the robustness conditions of \snpab\ and broaden its applicability, we introduce a second algorithm dubbed brute force (BF), that we use as a post-processing step to enhance \snpab\ results (which we denote SNPALQ+BF). 
		In Section~\ref{sec:robustness_postProc}, we prove that BF lead to robustness guarantees under weaker conditions than \snpab. 
		
		\item In Section~\ref{sec:experiments}, the effectiveness of the proposed algorithms is attested through extensive numerical experiments, in which among others \snpab\ is shown to obtain better results than SNPA on LQ mixings, and the SNPALQ+BF to obtain a very high rate of perfect recovery of the ground truth factors.
		
	\end{itemize}
	
	\begin{figure*}[!t]
		\centering
		\includegraphics[width=12cm]{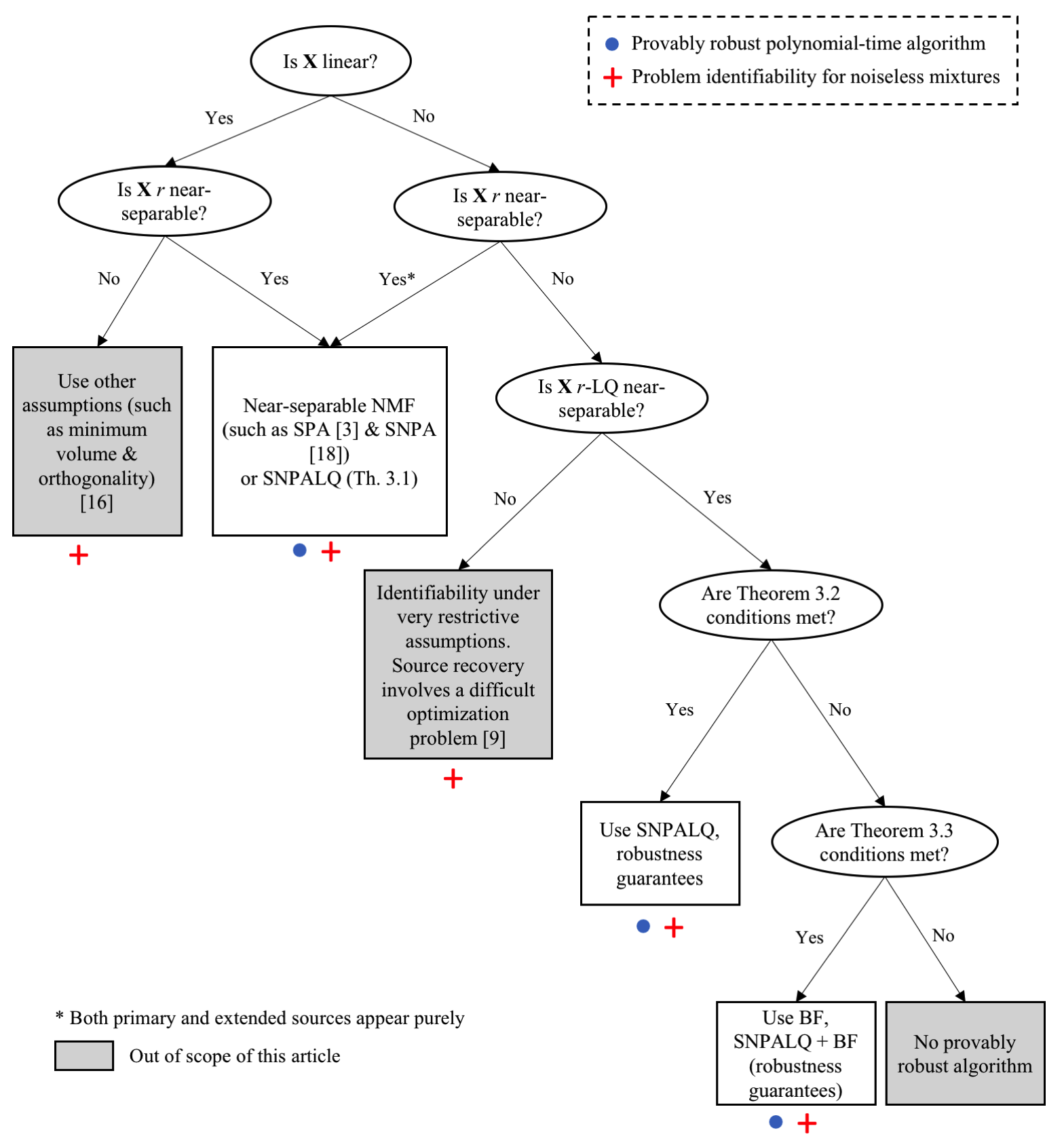}
		\caption{Graphical summary of the contributions, explaining which algorithm to use in which setting. We call a provably robust algorithm an algorithm which is proved to recover the sources even in the presence of noise.  \label{fig:recap_contributions} }
	\end{figure*}
	
		\remark 
		Near-separable algorithms have often been  used to  initialize NMF algorithms that do not rely on the separability assumption~\cite{Gillis2014}. 
		In particular, the initializations of many LQ-BSS algorithms are often (and paradoxically) performed with the output of near-separable algorithms assuming linear mixtures; see for example~\cite{Altmann2012, meganem2014linear}. 
	Therefore, beyond their intrinsic interest, the two algorithms proposed in the next section are fast and theoretically well-grounded   initialization strategies for LQ-BSS algorithms in the absence of the separability assumption.

	
	\subsection{Notation}
	\label{sec:Notations}
	In the following, we denote $\lbr r \rbr =\{1,2,..,r\}$, $|\mathcal{K}|$ the number of elements in the set $\mathcal{K}$ whose $i$th element is denoted $\mathcal{K}(i)$. 
	The $i$th column of a matrix $\mathbf{A} \in \mathbb{R}^{m\times r}$ is denoted $\mathbf{a}_i$. The submatrix formed by the columns indexed by $\mathcal{K}$ is denoted $\mathbf{A}_{\mathcal{K}}$, and the submatrix formed by all the columns of $\mathbf{A}$ except the ones indexed by $\mathcal{K}$ as $\mathbf{A}_{\setminus \mathcal{K}}$. 
	The set $\Delta^r$, for which the superscript is omitted when clear from the context, is $\Delta^r = \{\mathbf{x} \in \mathbb{R}^r | \mathbf{x} \geq 0, \sum_{i = 1}^{r} x_i \leq 1\}$. 
	In addition, we denote by $\Pi_q(\mathbf{W})$ the matrix containing all the columns of $\mathbf{W}$ and their products up to order $q \in \mathbb{N}^*$. 
	We will use $\Pi_2(\mathbf{W})$ which denotes the matrix containing the products up to order 2, that is,  
	\begin{equation*}
		\begin{split}
			\Pi_2(\mathbf{W}) &= [\mathbf{w}_1,\mathbf{w}_2,\dots,\mathbf{w}_r,\mathbf{w}_1 \odot \mathbf{w}_1,\mathbf{w}_2 \odot \mathbf{w}_1,\mathbf{w}_3 \odot \mathbf{w}_1,\mathbf{w}_2 \odot \mathbf{w}_2,\mathbf{w}_3 \odot \mathbf{w}_2,\dots,\mathbf{w}_{r} \odot \mathbf{w}_{r}]\\
			&= \left[(\mathbf{w}_i)_{i\in \lbr r \rbr},(\mathbf{w}_i\odot\mathbf{w}_j)_{\substack{i,j \in \lbr r \rbr \\ i \leq j}}\right] , 
		\end{split}
	\end{equation*} 
	and $\pstr(\mathbf{W})$ which contains the products up to order $4$. 
	Additional notations, specific to the theoretical and proof sections, will be introduced later for the sake 
	of readability. 

	\section{Two algorithms for LQ-BSS: \snpab\ and BF}\label{sec:proposed_algo} 
	
	To perform near-separable BSS of LQ mixtures, a first (naive) approach is to use 
	an LMM-based near-separable NMF algorithm to identify the $\tilde{r}$ extended sources. Since the quadratic terms $(\mathbf{w}_i\odot \mathbf{w}_j)_{\substack{i,j\in \lbr r \rbr \\ j \leq i}}$ can be considered as \emph{\virtual}\ sources (see Eq.~\eqref{eq:LQMat}), they could be retrieved along with the  columns of $\mathbf{W}$, provided that they appear purely in the data set.  
	One could for instance resort to SNPA~\cite{Gillis2014}, an LMM-based algorithm which has shown to yield very good separation performances compared to state-of-the-art LMM-based algorithms such as
	 VCA \cite{Nascimento2005a} and SPA \cite{Araujo01}, 
	and admits robustness guarantees. SNPA is a greedy algorithm: it iteratively constructs the near-separable NMF solution $\mathcal{K}$ by sequentially adding a new source to the current set of sources already identified. More precisely, after initializing the index set $\mathcal{K} = \emptyset$ and a residual matrix $\mathbf{R} = \bar{\mathbf{X}}$, each iteration of SNPA consists of the following two steps:  
		\begin{itemize}
			\item \emph{selection}: the index of the column of $\mathbf{R}$ maximizing a score function $f$ is added to $\mathcal{K}$.
			\item \emph{projection}: the residual is updated by projecting the columns of $\bar{\mathbf{X}}$ onto the convex hull formed by the columns of $\bar{\mathbf{X}}_\mathcal{K}$ and the origin. 
		\end{itemize}
	During the \emph{selection} step, the function $f$ aims at selecting the most relevant column of $\mathbf{R}$ to be identified as a source. This function, which can for example be the $\ell_2$-norm, needs to fulfill the following assumption:
				\begin{assumption}
				The function $f \colon \mathbb{R}^m \mapsto \mathbb{R}_+$ is $\mu$-strongly convex, its gradient is $L$-Lipschitz and its global minimizer is the all zero vector $\mathbf{0}_m$,  that is, $f(\mathbf{0}_m) = 0$. 
				\label{hyp:f}
			\end{assumption}
	The \emph{projection} step is a convex optimization problem and can be solved for example using a fast gradient method~\cite{Nesterov2013}. We refer the reader to \cite[Appendix A]{Gillis2014} for more details. 
   
	Nevertheless, the bottleneck of the above naive approach consisting in using SNPA for LQ mixtures is that the presence of all the  \emph{\virtual}\ sources as pure data samples is too strong. 
	Indeed  all \emph{\virtual}\ sources are not likely to be observed purely in the data set. As such, the recovery of the extended sources by SNPA is not guaranteed, calling for algorithms specifically designed for LQ mixtures.

	To overcome this limitation, we propose two new algorithms\footnote{\label{note:url}The algorithms will be made available online at https://sites.google.com/site/nicolasgillis/code} enabling to tackle LQ mixtures. 
	The first algorithm, referred to as \snpab, is a variant of SNPA specifically designed to handle LQ mixings; see Section~\ref{sec:SNPALQ}. The second one is a brute-force (BF) algorithm, extending the work of \cite{Arora2016} to LQ mixtures and exhibiting robustness guarantees under milder conditions than \snpab; see Section~\ref{sec:post_proc}. 
	As BF is however  computationally more expensive than \snpab, we propose to use it  as a post-processing of the output provided by \snpab. Combining both algorithms in a single method, which we refer to as \snpab+BF, allows us to benefit from the best of each of these algorithms. 
	
	\subsection{\snpab}\label{sec:SNPALQ}
	The rationale behind \snpab\ is that we are  interested by recovering the primary sources only, $\mathbf{w}_i$ for $i \in \lbr r \rbr$. The  virtual sources $\mathbf{w}_i \odot \mathbf{w}_j$ ($i,j \in \lbr r\rbr$) can be  considered as nuisance. We propose to take them into account in the separation process only to improve the extraction of the primary sources. 
At each iteration of \snpab, we perform the following two  steps (see Algorithm \ref{alg:algo_SNPALQ}): 
	\begin{itemize}
		\item \emph{Selection step} (unchanged compared to SNPA): the column of the residual matrix $\mathbf{R}$ maximizing a function $f$ fulfilling Assumption~\ref{hyp:f} is selected. 
		
		
		\item \emph{Projection step}  (different from SNPA):  \snpab\ performs the projection onto the convex hull formed by the origin, the sources extracted so far \emph{and 
		their second-order products}. 
		Therefore, if two sources $\mathbf{w}_i$ and $\mathbf{w}_j$ ($i \neq j$) are extracted during the iterative process of \snpab, the contribution of the \virtual\ sources $\mathbf{w}_i \odot \mathbf{w}_j$, $\mathbf{w}_i\odot \mathbf{w}_i$ and $\mathbf{w}_j\odot \mathbf{w}_j$ are removed. Beyond the advantage that these virtual sources 
		will not be extracted in the subsequent steps, their non-linear contribution is reduced, giving more weight to the linear part. 
	\end{itemize}
	
	Recall that SNPA projects each column of $\bar{\mathbf{X}}$  onto the convex hull formed by the origin and all the sources extracted so far to compute the residual $\mathbf{R}$, and does not take into account the \virtual\ sources. 
	Thus, the primary sources defining $\mathbf{W}$ are  more likely to be extracted by \snpab\ in the early steps of the iterative process; see Figure~\ref{fig_SNPA_dontworks_schema} for an illustration.  
	\begin{figure}[ht!]
	\begin{center}
			\begin{tabular}{cc}
			\includegraphics[width=2.8in]{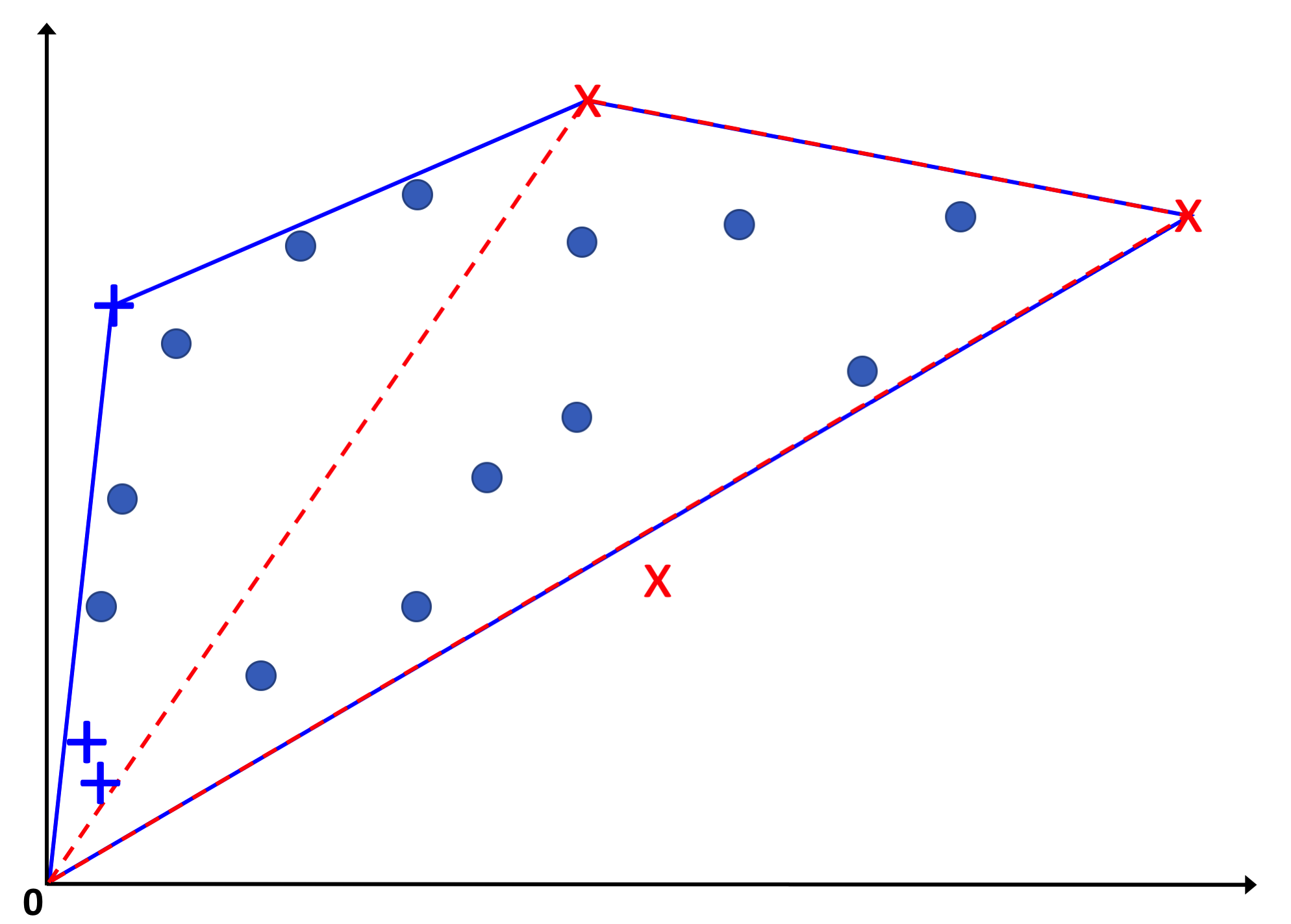}
            & 
			\includegraphics[width=2in]{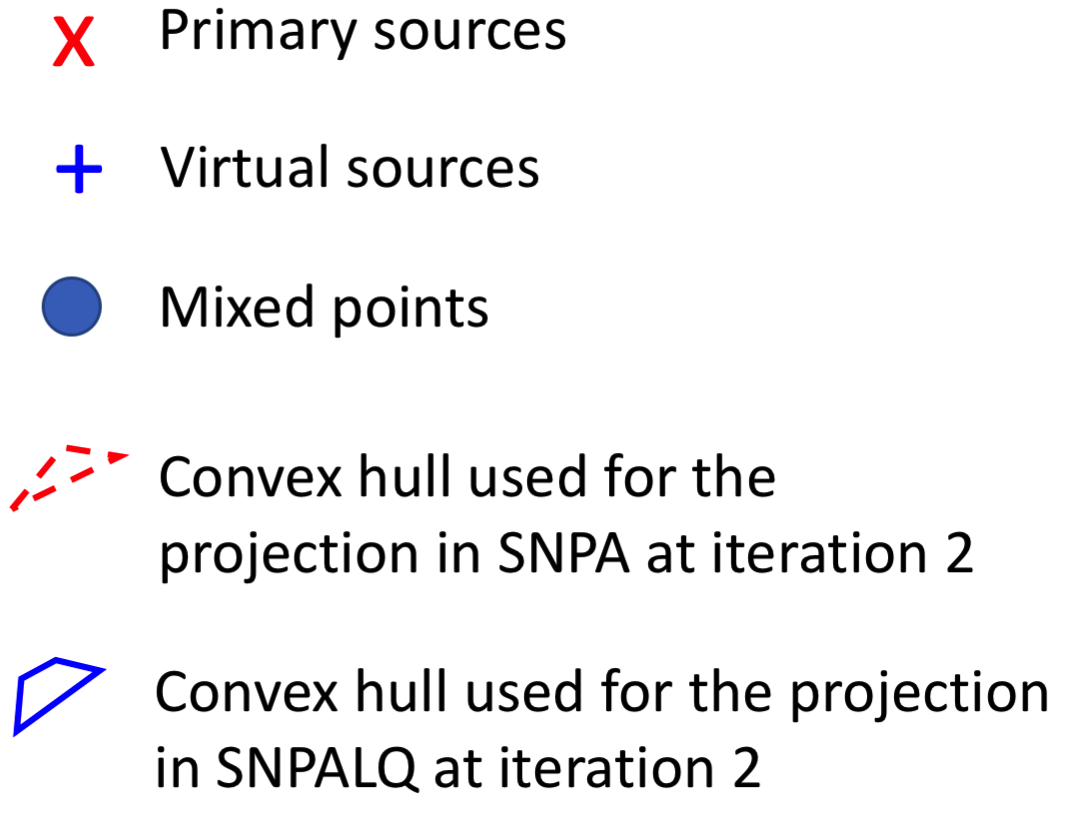} 
			\end{tabular}
			\caption{Example of a bilinear mixing for which \snpab\ is successful at recovering $\mathbf{W}$ but SNPA is not (the principle is the same for LQ, except that there are more virtual sources). 
			There are three primary sources, represented with the red \emph{X} markers, and three virtual sources, namely $\mathbf{w}_i \odot \mathbf{w}_j$ for $i \neq j$ and $1 \leq i,j \leq 3$, represented with the blue $+$ markers. 
			The columns of $\bar{\mathbf{X}}$ are made of the primary sources and the mixed points represented with the blue circles. 
			The red dashed line is the convex hull of the origin and the sources extracted after two iterations of SNPA. The plain blue line is the convex hull of the origin and the sources extracted by SNPALQ after two iterations, 
			as well as the corresponding virtual source. 
			Only the last primary source lies outside of the blue convex hull. Therefore, \snpab\ extracts it in its third iteration and then stops, returning the primary sources only. On the other hand, at the third iteration, SNPA fails to extract the last primary source because some of the (mixed) columns of $\bar{\mathbf{X}}$ lie further  away from the red dashed convex hull. 
			Moreover, it will need in total 8 iterations to terminate because the convex hull of the columns of $\bar{\mathbf{X}}$ has 8 vertices (we assume the virtual sources do not appear purely in the data set). \label{fig_SNPA_dontworks_schema}} 
	\end{center}
		\end{figure}
	 	\snpab\ will be proved in Section~\ref{sec:robustness_SNPAB} to extract the primary sources in the first $r$ steps, under specific conditions. 
	\snpab\ alternates the two above steps until one of the following two criteria is met: 
	\begin{itemize}
	
		\item A maximum of $r_\text{max}$  columns have been extracted. 
		If an upper bound is not available, one can take $r_\text{max} = n$ so that \snpab\ relies on the second stopping criterion only. Our theoretical results will rely on this criterion assuming $r$ is know. 
		
		\item $\norm{\mathbf{R}}{F} \leq t \norm{\bar{\mathbf{X}}}{F}$: 
		the algorithm stops when the relative reconstruction error 
		is sufficiently small. The choice of a good value for the tolerance parameter $t$ is important: if $t$ is too large, the \snpab\ could stop before the extraction of all the sources. If $t$ is too low, the \snpab\ could extract too many source candidates in the presence of noise, 
		making the whole algorithm computationally expensive. Theoretical results concerning the choice of $t$ are left for future work. 
	\end{itemize}
	
	
	
	\begin{algorithm}[!h]
		\caption{Successive Nonnegative Projection Algorithm for LQ mixtures (\snpab)}
		\label{alg:algo_SNPALQ}
		\begin{algorithmic}[1]
			\STATE {\bf Input}: $\bar{\mathbf{X}}  \in \mathbb{R}^{m\times \tilde{r}}$: a $r$-LQ $r$-near-separable matrix  following Definition~\ref{ass:rLSsep} and Constraints \eqref{eq:constraints}, $f$: a strongly convex function satisfying Assumption \ref{hyp:f}, $r_\text{max}$: number of sources,  
			$t \geq 0$: stopping criterion on the norm of the residual. 
			\item[]
			\STATE {\bf Initialization}: $\mathbf{R} = \bar{\mathbf{X}}$, $\mathcal{K} = \{ \}$, $k = 1$
			\item[]
			\WHILE{$\frac{\norm{\mathbf{R}}{F}}{\norm{\bar{\mathbf{X}}}{F}} > t$ and $k \leq r_\text{max}$}
			\STATE $p = \argmax_{j\in \lbr n \rbr}f(\mathbf{r}_j)$;
			\STATE $\mathcal{K} = \mathcal{K} \cup \{p\}$;
			\FOR{$j \in  \lbr n \rbr$}
			\STATE $\mathbf{h}_{j} = \argmin_{\mathbf{h} \in \Delta^{\frac{|\mathcal{K}|(|\mathcal{K}|+3)}{2}}} f(\bar{\mathbf{x}}_j - \ps(\bar{\mathbf{X}}_\mathcal{K})\mathbf{h})$
			\STATE $\mathbf{r}_j = \bar{\mathbf{x}}_j - \ps(\bar{\mathbf{X}}_\mathcal{K})\mathbf{h}_{j}$
			\ENDFOR
			\STATE $k = k + 1$
			\ENDWHILE
			\item[]
			\STATE {\bf Output}: A set $\mathcal{K}$ of indices such that $\bar{\mathbf{X}}_{\mathcal{K}} \simeq \mathbf{W}$ up to a permutation. 
		\end{algorithmic}
	\end{algorithm}

	\subsection{Brute force algorithm}
	\label{sec:post_proc}
	The conditions ensuring \snpab\ to recover the sources might not be satisfied in practice (see Sections~\ref{sec:interp_SNPALQ} and \ref{sec:exp_cond}). Therefore, we propose here a second algorithm, BF, 
	inspired by the algorithm of Arora et al.~\cite{Arora2016} for linear mixtures. As we will see in Section~\ref{sec:robustness_postProc}, 
	it  requires milder assumptions for the source recovery.

	\myparagraph{Noise-free mixtures} For the sake of simplicity, the rationale underlying BF is first exposed in the absence of noise. 
	Let us assume w.l.o.g.\ that there are no duplicated columns in the data set $\mathbf{X}$. Due to the separable assumption, $\mathbf{X}$ can be written as:
	\begin{equation}
	\mathbf{X} = \left[\mathbf{W}, \tilde{\mathbf{X}}\right] \mathbf{P} \;  \in \;  \mathbb{R}^{m \times n},
	\label{eq:SNPAB_output}
	\end{equation}
	where $\mathbf{P}$ is a permutation and  $\tilde{\mathbf{X}}$ contains the LQ mixings of $\mathbf{W}$.  
	Let us consider a column of $\mathbf{X}$, $\mathbf{x}_k$ for $k \in \lbr n \rbr$. 
	We can check whether it is contained in the convex hull of the other columns of $\mathbf{X}$, their LQ mixtures and the origin by solving 
	\begin{equation*}
		s_k \; = \; \min_{\mathbf{h} \in \Delta^{\frac{n(n+3)}{2} - 1}} \norm{\mathbf{x}_k - \ps(\mathbf{X})_{\setminus \{k\}} \mathbf{h}}{2} . 
	\end{equation*} 
	If $\mathbf{x}_k$ is not a column of $\mathbf{W}$, we have $s_k = 0$ under the $r$-LQ separable mixing model (Definition~\ref{ass:rLSsep} with $\mathbf{N} = 0$).   
	Moreover, under the assumption  that $\mathbf{W}$ is order-2 $\alpha$-robust simplicial, that is, $\alpha_2 (\mathbf{W}) > 0$, 	$\mathbf{x}_k$ is a source, that is, a column of $\mathbf{W}$,  if and only if $s_k > 0$. 
	
	For sake of consistency with \snpab, this condition can be generalized to 
	any function~$f$ fulfilling Assumption~\ref{hyp:f}. Adopting this generalization, $\mathbf{x}_k$ is as primary source if and only if 
		\begin{equation}
		\min_{\mathbf{h} \in \Delta^{\frac{n(n+3)}{2} - 1}} f\left(\mathbf{x}_k - \ps(\mathbf{X})_{\setminus \{k\}} \mathbf{h}\right) > 0.
		\label{eq:post_proc}
		\end{equation}
	
	\myparagraph{Noisy mixtures} We here extend the above principles to make the BF algorithm able to recover an approximation of $\mathbf{W}$ from noisy mixtures $\bar{\mathbf{X}}={\mathbf{X}} + {\mathbf{N}}$ for a bounded noise fulfilling $\max_{i \in \lbr t \rbr}\norm{\mathbf{n}_i}{2} \leq \epsilon$ for some $\epsilon \geq 0$; see Algorithm \ref{alg:algo_PP}. 
	To do so, we need to modify~\eqref{eq:post_proc} in two ways. 
	\begin{itemize}
		\item In the noise-free case, we assumed that no duplicated columns are present within $\mathbf{X}$, and it is easy to discard such duplicates.  
		In the noisy setting, 
		when evaluating the residual \eqref{eq:post_proc}, not only the column $\bar{\mathbf{x}}_k$ should be removed from  $\ps(\bar{\mathbf{X}})$ but also all columns 
		close to $\bar{\mathbf{x}}_k$ (see Figure~\ref{fig_PP_robust_loners} for an illustration). 
		
	
		\begin{figure}[!t]
			\centering
			\subfloat[]{
				\includegraphics[width=2.2in]{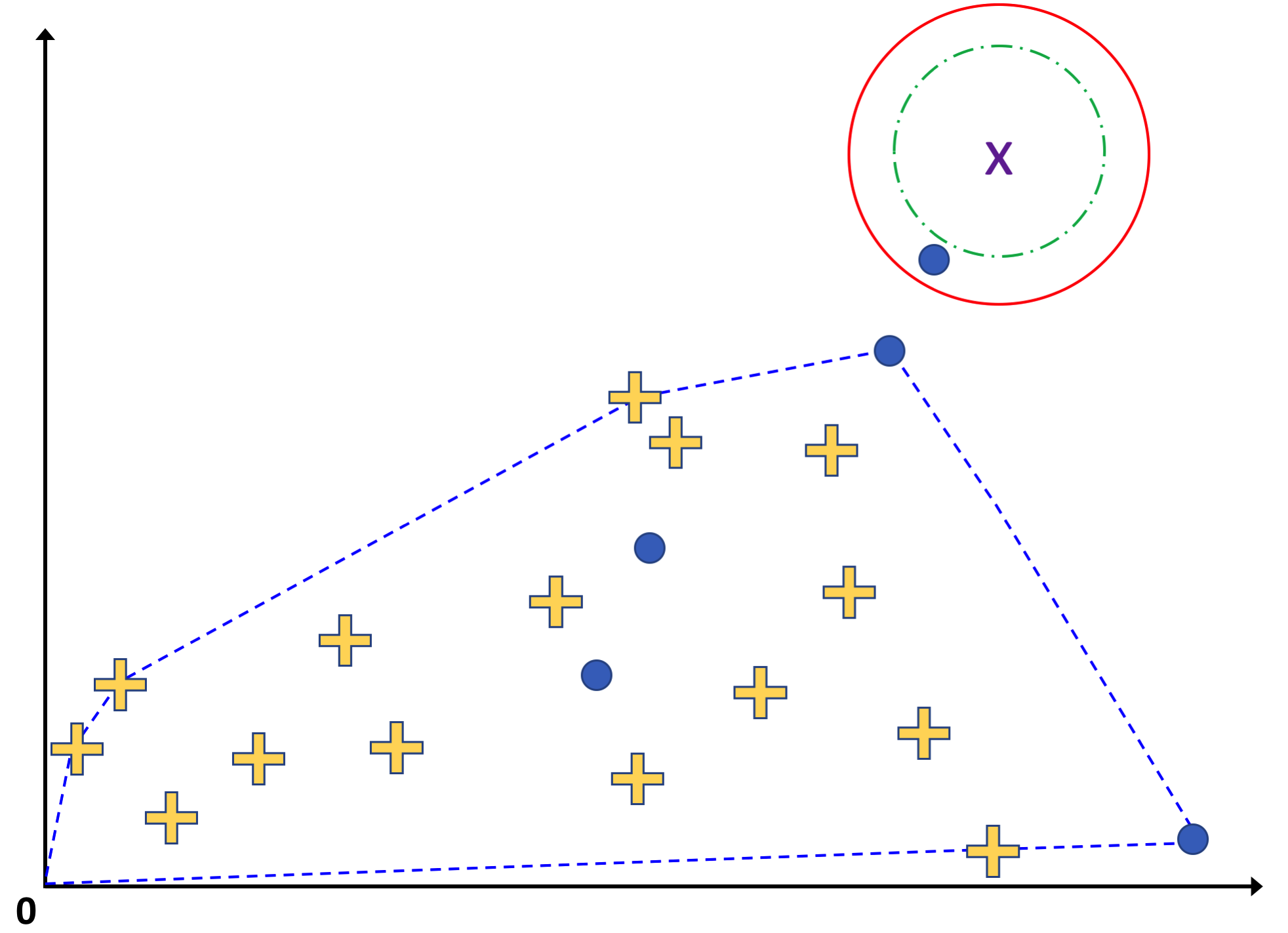}
			}
			\subfloat[]{
				\includegraphics[width=2.2in]{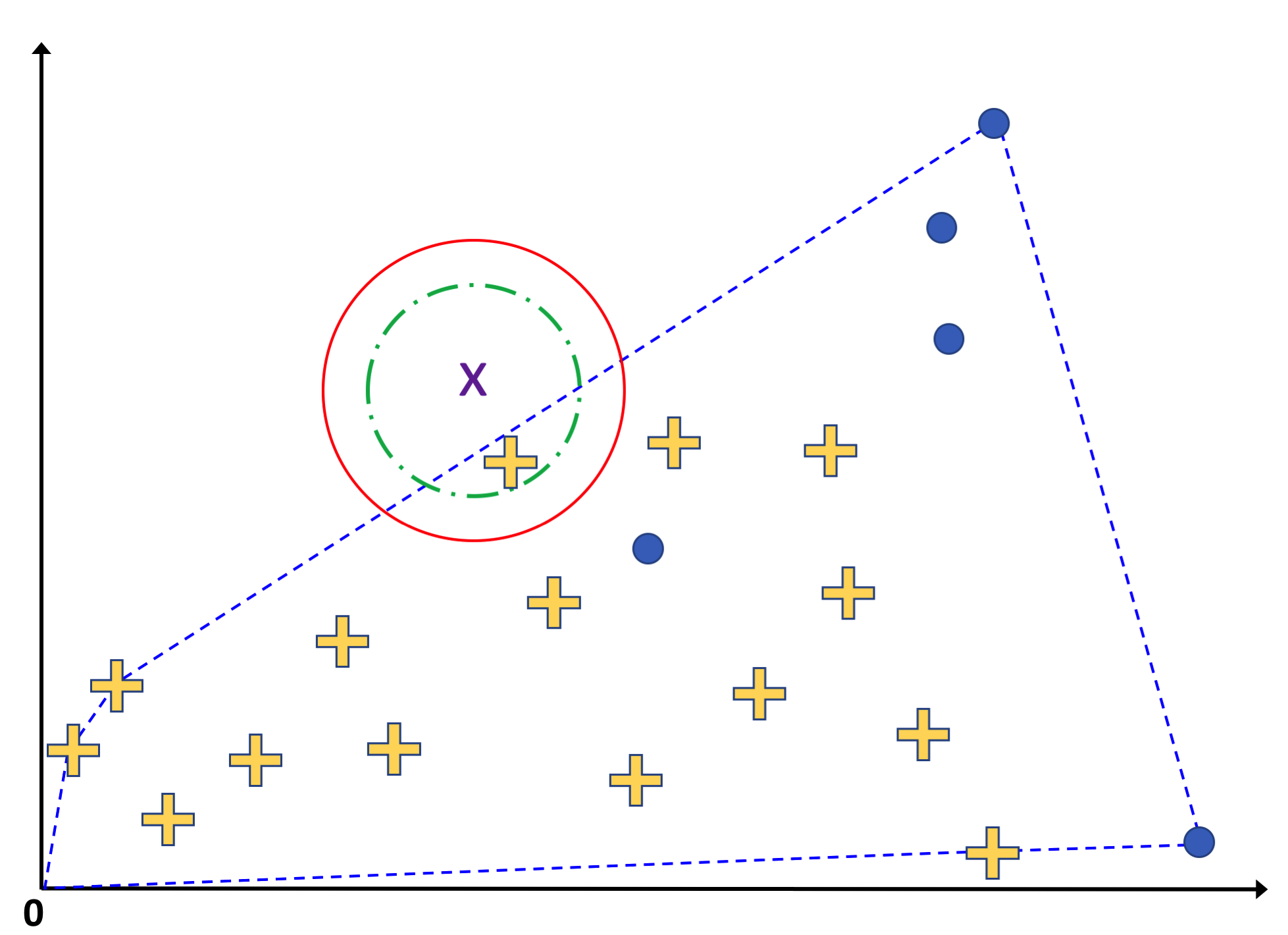}
			}
				\includegraphics[width=1.5in]{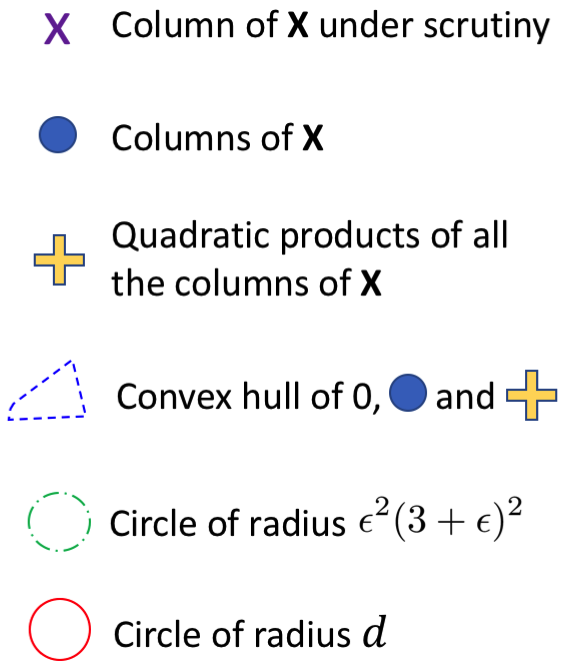}
			\caption{Illustration of condition~\eqref{eq:post_pro_noisy} with $f(\cdot) = \norm{\cdot}{2}$. 
			The point under scrutiny $\bar{\mathbf{x}}_k$  is represented in violet ('X' marker). 
			The dots are the columns of $\mathbf{X}{\setminus \{k\}}$, and the yellow cross ('+' marker) correspond to the quadratic products of the columns of $\mathbf{X}$. 
			The plain line ball of radius $d$ and center $\bar{\mathbf{x}}_k$ contains the columns of $\bar{\mathbf{X}}$ which are discarded in~\eqref{eq:post_pro_noisy}. 
			 The dotted polygon is the convex hull of the origin and the columns of  $\ps(\mathbf{X})_{\setminus \{k\}}$ that are not contained in the ball of radius $d$ around $\bar{\mathbf{x}}_k$. 
			The dashed circle of radius $\epsilon^2(3+\epsilon)^2$ indicates the distance at which the point must be located from the dotted convex hull to be considered an  LQ-robust loner. 
			On the figure (a), the dashed circle does not intersect the convex hull, and hence the cross is  an  LQ-robust loner. 
			On figure (b), the dashed circle overlaps the convex hull, making that its center point is not a robust loner. \label{fig_PP_robust_loners}} 
		\end{figure}
		
		\item Moreover, as the noise might shift mixed data points outside the convex hull formed by $\ps(\mathbf{W})$ and the origin, $s_k$ might be nonzero for a mixed column $\bar{\mathbf{x}}$ (that is, $\bar{\mathbf{x}}_{k} = \tilde{\mathbf{x}}_j$ for some $j\in \lbr n-r \rbr$); see Figure~\ref{fig_PP_robust_loners} for an illustration. 
		
		\end{itemize}{}
	Therefore, the condition \eqref{eq:post_proc} in the noiseless case should be modified to
		\begin{equation}
		\min_{\mathbf{h} \in \Delta^{\frac{n(n+3)}{2} - 1}}f\left(\bar{\mathbf{x}}_k - \ps(\bar{\mathbf{X}})_{\setminus \{i\in \lbr n\rbr \ | \ f\left(\bar{\mathbf{x}}_i - \bar{\mathbf{x}}_k\right) > d\}}\mathbf{h}^*\right) > \frac{L}{2}\epsilon^2(3+\epsilon)^2, 
		\label{eq:post_pro_noisy}
		\end{equation}
		with $L$ the Lipschitz constant of $f$ and $d$ a threshold parameter discussed in Appendix~\ref{sec:proofs}; see~\eqref{eq:d} for an explicit value. 
		The right-hand side stems from the fact that the noise is corrupting both the data columns (with a maximum energy of $\epsilon$) and their quadratic products (with a maximum energy of $2\epsilon+\epsilon^2$ if the columns of $\mathbf{X}$ have a unit norm); see Definition~\ref{def:robust_loner_full}.

	    Following \cite{Arora2016}, the columns of $\bar{\mathbf{X}}$ satisfying the condition \eqref{eq:post_pro_noisy} are called the \emph{LQ-robust loners}. Section~\ref{sec:robustness_postProc} will show that these columns exactly correspond to good approximations of the sources. To approximately recover the sources, the BF algorithm then amounts to check which columns of $\bar{\mathbf{X}}$  are LQ-robust loners. 
	    However, due to the noise, different LQ-robust loners may be candidates for estimating the \emph{same} source. 
	    Therefore, at the end of BF, the LQ-robust loners need to be clustered to obtain a single estimate of each source. Fortunately, such a clustering -- described in Algorithm~\ref{alg:algo_PP} -- is easy and does not lead to any indeterminacy as the LQ-robust loners are located close to the sources, which are comparatively further from each others. 
	    

	\begin{algorithm}[!h]
		\caption{Brute force (BF)}
		\label{alg:algo_PP}
		\begin{algorithmic}[1]
			\STATE {\bf Input}: A $r$-LQ $r$-near-separable matrix $\bar{\mathbf{X}}  \in \mathbb{R}^{m\times \tilde{r}}$ following Definition \ref{ass:rLSsep} and constraints (Eq.~\ref{eq:constraints}), $r$: number of sources, and $f$ a strongly convex function satisfying Assumption~\ref{hyp:f}, $\epsilon = \max_{i \in \lbr t \rbr}\norm{\mathbf{n}_i}{2}$, $d$ given by Equation~\eqref{eq:d}.  
			\item[]
			\STATE {\bf Initialization}: $\mathcal{K} = \{\}$
			\item[]
			\FOR{$k \in \lbr t \rbr$}
			\STATE $\mathbf{h}_k = \argmin_{\mathbf{h} \in \Delta} f\left(\bar{\mathbf{x}}_k - \ps(\bar{\mathbf{X}})_{\setminus \{i \in \lbr t \rbr \ | \ f(\bar{\mathbf{x}}_i - \bar{\mathbf{x}}_k) > d\}}\mathbf{h}\right)$
			\IF{$f\left(\bar{\mathbf{x}}_k - \ps(\bar{\mathbf{X}}_\mathcal{N})_{\setminus \{i \in \lbr t \rbr \ | \ f(\bar{\mathbf{x}}_i - \bar{\mathbf{x}}_k) > d\}}\mathbf{h}_k \right) > \frac{L}{2}\epsilon^2(3+\epsilon)^2$}
			\STATE $\mathcal{K} = \mathcal{K} \cup \{k\}$
			\ENDIF
			\ENDFOR
			\item[]
			\STATE {\bf Clustering on $\bar{\mathbf{X}}_\mathcal{K}$}:  assign two columns $\bar{\mathbf{x}}_i$ and $\bar{\mathbf{x}}_j$ of $\bar{\mathbf{X}}_\mathcal{K}$ to the same cluster if and only if 
			$\norm{\bar{\mathbf{x}}_j - \bar{\mathbf{x}}_k}{2} \leq 2\sqrt{\frac{2}{\mu}(d + \epsilon L (2K(\mathbf{X})+\epsilon))}$. Update $\mathcal{K}$ by keeping only one column for each cluster. 
			\item[]
			\STATE {\bf Output}: A set $\mathcal{K}$ of indices such that $\bar{\mathbf{X}}_{\mathcal{K}} \simeq \mathbf{W}$ up to a permutation.
			
		\end{algorithmic}
	\end{algorithm}

	\myparagraph{BF algorithm as a post-processing} Even if the BF algorithm can be used \emph{per se} to perform separation from LQ near-separable mixtures, it can also serve as a post-processing to refine the results provided by \snpab. This strategy is particularly appealing when \snpab\ robustness conditions are not met, in which case \snpab\ may extract mixed data columns or \virtual\ sources in addition to the sought-after primary sources. 
Given an \snpab\ solution $\bar{\mathbf{X}}_\mathcal{K}$, assume $r$ columns correspond to the primary sources $\mathbf{W}$, and the $|\mathcal{K}|-r$ remaining ones to (spurious) columns in which the primary sources are mixed along with their quadratic products. Up to a permutation, the \snpab\ solution can be written  as 
    \begin{equation}
	\bar{\mathbf{X}}_\mathcal{K} \simeq \left[\mathbf{W}, \tilde{\mathbf{X}}\right] \in \mathbb{R}^{m \times |\mathcal{K}|},
	\end{equation}
	where $P$ is a permutation, and $\tilde{\mathbf{X}}\in \mathbb{R}^{m\times (|\mathcal{K}| - r)}$ are data points. This matches the form of \eqref{eq:SNPAB_output}. Therefore, instead of using the BF algorithm directly on the data set $\bar{\mathbf{X}}$, it can be applied on the \snpab\ solution $\bar{\mathbf{X}}_\mathcal{K}$, which has in practice a significantly smaller number of columns, that is, $|\mathcal{K}| \ll n$. 
	Using BF as a post-processing step significantly reduces the computational cost;  
	see Section~\ref{sec:complexity}. Furthermore, it is worth noting that  \snpab\ already identifies as sources columns of $\bar{\mathbf{X}}$ lying far from each other. Thus, in our experiments, the clustering step in BF, whenever used as a post-processing, was never necessary since each cluster contained exactly one point. 	
	
	\remark While we advocate BF as a post-processing enhancing \snpab\ results, the reciprocal point of view can be also adopted: \snpab\ can be seen as a \emph{screening} (or pruning) method, enabling to select only a few number of potential candidates and lightening the computational burden of BF. 

	\subsection{Computational cost}
	\label{sec:complexity}
	The  computational costs of the two proposed algorithms are as follows:
	\begin{itemize}
	
		\item \emph{\snpab}: The complexity of the $k$th iteration is dominated by computing the projection step, which requires the projection of a $m$-by-$n$ matrix onto a convex hull with $k(k+3)/2 + 1$ vertices, requiring $\mathcal{O}\left(mnk^2\right)$ operations with a first-order method~\cite[Appendix A]{Gillis2014}.  
		
			\item \emph{BF}: Solving~\eqref{eq:post_pro_noisy} for the $n$ data points with a first-order method (as for \snpab)  requires $\mathcal{O}\left( mn^2 \right)$ operations. This is computationally rather heavy. For example, for HS images, $n$ is  the number of pixels and typically of the order of millions. 
		
		\item \emph{\snpab +BF}: Assuming \snpab\ extracts  $|\mathcal{K}|$ indices,  it requires $\mathcal{O}\left(m n  |\mathcal{K}|^2\right)$ operations for \snpab, and 
		$\mathcal{O}\left(m |\mathcal{K}|^2\right)$ operations for the post-processing with BF.  Hence BF used as a post-processing  has a smaller computational cost than \snpab\ which further justifies its use. 
		
	\end{itemize}
	
	
	\remark[Handling simpler models] \label{rem:LQ_bilin} As the LQ mixing model encompasses in particular the linear and bilinear ones, both \snpab\ and BF can be employed to separate these (simpler) mixtures. However, in practice, \snpab+BF should be specifically tailored in agreement with the target mixing model. For instance, bilinear mixtures can be handled by \snpab+BF by removing the projections on the squared sources in the projection steps, reducing the computational burden while improving the separation performance, avoiding the projections on the non-existing quadratic terms. 
	
	\section{Theoretical results}
	This section reports the theoretical results associated with the recovery of the sources by  \snpab\ and BF, even in the presence of noise. 
	More specifically, in Section~\ref{sec:robustness_linear}, we first derive robustness guarantees for \snpab\ when applied to \emph{linear} mixings. These guarantees are then extended to LQ mixings in Section~\ref{sec:robustness_SNPAB}. The required conditions for these recovery results are discussed in Section~\ref{sec:interp_SNPALQ}. 
	In Section \ref{sec:robustness_postProc}, we derive and discuss the recovery guarantees for BF. For the sake of simplicity, the results derived in this section are stated for the particular choice $f(\cdot) = \norm{\cdot}{2}$. Our results are stated in a more general setting for any  function $f(\cdot)$ satisfying Assumption \ref{hyp:f} in Appendix~\ref{sec:proofs}, where the proofs are given.  
	
	
	\subsection{Robustness of \snpab} 
	As the LQ model is a generalization of the linear one (see Section~\ref{sec:intro}), 
	we first prove  robustness of \snpab\  with respect to (w.r.t.) noise for \emph{linear} mixings in Section~\ref{sec:robustness_linear}. 
	However, as expected, we will see that the derived bounds on the admissible noise levels and the corresponding error on the source estimates are slightly worse than those associated with SNPA because of the additional projections on the (non-existing) virtual sources. In Section~\ref{sec:robustness_SNPAB}, robustness of \snpab\ is proved in  the case of LQ mixings.

	\subsubsection{Linear mixtures} \label{sec:robustness_linear}

	Before stating the main result of this section in Theorem~\ref{thm:robustness_lin_simple}, let us introduce additional notations. 
	For a matrix $\mathbf{A} \in \mathbb{R}^{m\times r_A}$, we define\footnote{Note that in the signal processing literature, such a norm is sometimes denoted as $\norm{\mathbf{A}}{\infty,2}$, see for instance \cite{Kowalski2009}. We prefer to keep the original notation of \cite{Gillis2014}.}
	\begin{equation*}
		K(\mathbf{A}) = \norm{\mathbf{A}}{1,2} = \max_{i \in \lbr r_A\rbr}\norm{\mathbf{a}_i}{2},
	\end{equation*}
	which is the maximum of the $\ell_2$ norm of the columns of a matrix $\mathbf{A}$. We denote $\mathcal{P}_\mathbf{A}^f(\mathbf{x})$ the projection of $\mathbf{x}$ onto the convex hull formed by the columns of $\mathbf{A}$ and the origin w.r.t.\ the semimetric induced by the function $f$ (see Assumption~\ref{hyp:f}):
	\begin{equation*}
		\mathcal{P}_\mathbf{A}^f(\mathbf{x}) = \mathbf{A}\mathbf{y}^* 
		\; \text{ with } \mathbf{y}^* = \argmin_{\mathbf{y} \in \Delta}f(\mathbf{x}-\mathbf{A}\mathbf{y}).
	\end{equation*}
	The residual of the projection is denoted $\mathcal{R}^f_\mathbf{A}$,  that is, 
	\begin{equation*}
		\mathcal{R}^f_\mathbf{A}(\mathbf{x} ) = \mathbf{x} - \mathcal{P}^f_\mathbf{A}(\mathbf{x}).
	\end{equation*}
	When used on matrices, both the projection and residual operators are applied column-wise (for instance, for all $i \in \lbr t \rbr$, $\mathcal{R}^f_\mathbf{A}(\mathbf{X})_i = \mathcal{R}^f_\mathbf{A}(\mathbf{x}_i)$). Furthermore, we define the following quantities associated with the minimal norm of the residuals
	\begin{itemize}
		\item[] $\nu_{f,\ps(\mathbf{A})}(\mathbf{A}) = \min_{j\in \lbr r_A\rbr}\norm{\mathcal{R}^f_{\ps(\mathbf{A})_{\setminus \{j\}}}(\mathbf{a}_j)}{2}$,
		\item[] $\gamma_{f,\ps(\mathbf{A})}(\mathbf{A}) = \min_{\substack{i,j \in \lbr r_A\rbr \\ i \neq j}}\norm{\mathcal{R}^f_{\ps(\mathbf{A})_{\setminus \{i,j\}}}(\mathbf{a}_j) - \mathcal{R}^f_{\ps(\mathbf{A})_{\setminus \{i,j\}}}(\mathbf{a}_i)}{2}$,
		\item[] $\beta^{\text{Lin}}_{\ps(\mathbf{A})}(\mathbf{A}) = \min\left(\nu_{f,\ps(\mathbf{A})}(\mathbf{A}),\frac{\sqrt{2}}{2}\gamma_{f,\ps(\mathbf{A})}(\mathbf{A})\right)$.\\
	\end{itemize}
	As such, $\beta^{\text{Lin}}_{\ps(\mathbf{A})}(\mathbf{A})$ is the minimum between the smallest residual of the column of $\mathbf{A}$ and the smallest difference between the residuals of the columns of $\mathbf{A}$ after the projection onto $\ps(\mathbf{A})$. 
	
	The following theorem states the robustness of SNPALQ in the case of linear mixtures. As mentioned earlier, it is here stated in a simplified formulation by assuming that $f(\cdot) = \norm{\cdot}{2}$. Its generalized counterpart for any $f(\cdot)$ satisfying Assumption \ref{hyp:f}, as well as the corresponding detailed proof, are reported in Appendix~\ref{sec:proofs} (see Theorem \ref{thm:robustness_lin}). 
	\begin{theorem}[Robustness of \snpab\ when applied on linear mixings -- Simplified version]\label{thm:robustness_lin_simple}
		Let 
		\begin{equation*}
			\bar{\mathbf{X}} = \mathbf{WH + N} \in \mathbb{R}^{m\times n}
		\end{equation*}{}
		be a near-separable \cite{Gillis_12_SparseandUnique} linear mixing 
		with $\alpha_{\ps(\mathbf{W})}(\mathbf{W}) > 0$ and $\beta^{\text{Lin}}_{\ps(\mathbf{W})}(\mathbf{W}) > 0$. Let $\norm{\mathbf{n}_i}{2}\leq \epsilon$ for all $i \in \lbr t \rbr$ with $\epsilon < \mathcal{O}\left(\frac{\beta^{\text{Lin}}_{\ps(\mathbf{W})}(\mathbf{W})^4}{K(\mathbf{W})^2}\right)$. Then \snpab\ (Algorithm~\ref{alg:algo_SNPALQ}) with $f = \norm{\cdot}{2}$ identifies in $r$ steps all the columns of $\mathbf{W}$ up to error $\mathcal{O}\left(\epsilon\frac{K(\mathbf{W})^2}{\beta^{\text{Lin}}_{\ps(\mathbf{W})}(\mathbf{W})^2}\right)$. 
	\end{theorem}

	As in \cite{Gillis2014}, Theorem~\ref{thm:robustness_lin_simple} can be proved by induction: we show that \snpab\ extracts a new column of $\mathbf{W}$ at each iteration.

	\subsubsection{LQ mixings}\label{sec:robustness_SNPAB} 
	We now extend the above result to the case of LQ mixings. 
	Similarly to the linear case, we define
	\begin{equation*}
		\beta^{\text{LQ}}_{\ps(\mathbf{W})}(\mathbf{A}) = \min\left(\frac{\nu_{f,\ps(\mathbf{W})}(\mathbf{A})}{2}\sqrt{\frac{\mu}{L}}\left[1-\frac{1}{G}\right],\gamma_{f,\ps(\mathbf{W})}(\mathbf{A})\right),
	\end{equation*}
	for some constant $G > 1$ upper-bounded by a quantity depending on the mixtures; see \eqref{eq:hypLQ_simple} below. 
	The robustness of \snpab\ when analyzing LQ mixings is stated below for $f(\cdot) = \norm{\cdot}{2}$. In Appendix~\ref{sec:proofs}, Theorem~\ref{thm:robustnessLQ} generalizes this statement to any $f(\cdot)$ satisfying Assumption~\ref{hyp:f}.
	
	\begin{theorem}[Robustness of \snpab\ when applied on LQ mixings -- Simplified version] \label{thm:robustnessLQ_simple}
		Let 
		\begin{equation*}
			\bar{\mathbf{X}} = \ps(\mathbf{W})\mathbf{H} + \mathbf{N} \in \mathbb{R}^{m\times n}
		\end{equation*}
		be an LQ mixing satisfying Definition~\ref{ass:rLSsep} with $\alpha_{\ps(\mathbf{W})}(\mathbf{W}) > 0$ and ${\beta^{\text{LQ}}_{\ps(\mathbf{W})}(\mathbf{W})} > 0$. Let $\norm{\mathbf{n}_i}{2} \leq \epsilon$ with $\epsilon < \mathcal{O}\left(\frac{{\beta^{\text{LQ}}_{\ps(\mathbf{W})}(\mathbf{W})}^4}{K(\ps(\mathbf{W}))^2}\right)$.
		Furthermore, let us assume that at each iteration of \snpab\ the following condition is fulfilled:
		\begin{equation}
			\small
			K\left(\mathcal{R}^{\norm{\cdot}{2}}_{\ps{(\bar{\mathbf{B}})}}(\mathbf{A})\right) 
			\geq 2G K\left(\mathcal{R}^{\norm{\cdot}{2}}_{\ps{(\bar{\mathbf{B}})}}\left((\mathbf{b}_i)_{i\in \lbr s\rbr},(\mathbf{a}_i\odot \mathbf{a}_j)_{\substack{i\leq j \\ i\in \lbr k\rbr \\ j\in \lbr k\rbr}},(\mathbf{b}_i\odot \mathbf{b}_j)_{\substack{i\leq j \\ i\in \lbr s\rbr \\ j\in \lbr s\rbr}},(\mathbf{a}_i\odot \mathbf{b}_j)_{\substack{i\in \lbr k\rbr \\j \in \lbr s\rbr}}\right)\right)
			\label{eq:hypLQ_simple}
		\end{equation}
		where $\mathbf{B}$ contains the columns of $\mathbf{W}$ already extracted by \snpab\ and $\bar{\mathbf{B}}$ the corresponding columns with noise, 
		$\mathbf{A}$ contains the remaining columns of $\mathbf{W}$ still-to-be extracted, and $G > 1$ is a constant. 
		Then, \snpab\ (Algorithm~\ref{alg:algo_SNPALQ}) with $f = \norm{\cdot}{2}$ identifies in $r$ steps the columns of $\mathbf{W}$ up to an error $\mathcal{O}\left(\epsilon \frac{K(\ps(\mathbf{W}))^2}{{\beta^{\text{LQ}}_{\ps{(\mathbf{W})}}}(\mathbf{W})^2}\right)$.\\
	\end{theorem}
	
	Similarly to the robustness result for linear mixtures, the above theorem is shown by induction. The main difference is that, in the LQ case, the \emph{virtual} sources (and the mixed data columns for which their contribution is nonzero) 
	might have a large residual and hence be extracted, whereas we would like to extract only the primary sources. Therefore, we must introduce the additional condition (\ref{eq:hypLQ_simple}). 
	Roughly speaking, it requires the energy of the residual of a non-already extracted source to be higher than twice the maximum of (i)~the largest energy of the virtual sources, which  prevents \snpab\ to  extract a virtual source, and (ii)~the largest energy of the already-extracted sources, which precludes extracting two columns of $\bar{\mathbf{X}}$ corresponding to the same source.

	\subsubsection{Interpretation of \snpab\ recovery conditions}
	\label{sec:interp_SNPALQ}
	In addition to the mixing constraints described in Section~\ref{sec:contraintsOnParameters}, among which near-separability, we here give more insights concerning some of the conditions for \snpab\ robustness when applied on LQ mixtures. 
	
\noindent \textbf{Condition on $\alpha_{\ps(\mathbf{W})}(\mathbf{W})$ --}  The condition $\alpha_{\ps(\mathbf{W})}(\mathbf{W}) > 0$ is of uttermost importance. It ensures that no column of $\mathbf{W}$ lies within the convex hull of the other columns of $\mathbf{W}$, the origin, and the second order products of the columns of $\mathbf{W}$. On the contrary, $\alpha_{\ps(\mathbf{W})}(\mathbf{W}) = 0$ would mean that at least one columns of $\mathbf{W}$ would be indistinguishable from the mixed data columns. 
		Compared to SNPA, this condition is more restrictive for \emph{linear} mixings. For example, let us consider the noiseless mixtures $\mathbf{X} = \mathbf{WH}$ with
		\begin{equation*}
			\mathbf{W} = 
			\begin{bmatrix}
				1 & 1 & 1 \\
				1 & 0 & 0 \\
				0 & 1 & 0 
			\end{bmatrix}
		\end{equation*}
		for which $\alpha_{\ps(\mathbf{W})}(\mathbf{W}) = 0$. %
		During its two first iterations, \snpab\ extracts the two first columns of $\mathbf{W}$. But as $\mathbf{w}_1 \odot \mathbf{w}_2 = \mathbf{w}_3$, all data columns in $\mathbf{X}$ can be written as a nonnegative combination of $[\mathbf{w}_1,\mathbf{w}_2,\mathbf{w}_1 \odot \mathbf{w}_2]$, and hence \snpab\ stops after the second iteration (the residual being zero) without extracting $\mathbf{w}_3$. On the contrary, SNPA is able to extract the thre columns of $\mathbf{W}$ since $\mathrm{rank}(\mathbf{W}) = 3$. 
		
		
		On the other hand, even if the virtual sources appear purely in the mixture, trying to solve the LQ problem using the naive approach explained at the beginning of Section~\ref{sec:proposed_algo}, namely applying SNPA on a LQ-mixing with the hope to extract both sources and virtual sources and then rejecting the virtual ones, would require $\alpha_{\ps(\mathbf{W})}(\ps(\mathbf{W})) > 0$, which is a stronger condition than the one of \snpab. 
		Indeed, this would require all the virtual sources not to lie within the convex hull of the other columns of $\ps(\mathbf{W})$ and the origin, which should not be required as we do not need to estimate them. 
		
\noindent \textbf{Condition on $\beta^{\text{LQ}}_{\ps(\mathbf{W})}(\mathbf{W})$ --} 
The condition $\beta^{\text{LQ}}_{\ps(\mathbf{W})}(\mathbf{W}) > 0$ 
is stronger than the corresponding condition of SNPA which requires  
$\beta^{\text{Lin}}_{\mathbf{W})}(\mathbf{W}) > 0$.  As discussed for SNPA in~\cite{Gillis2014}, this condition is most often satisfied as long as $\alpha_{\ps(\mathbf{W})}(\mathbf{W}) > 0$. 


		\noindent \textbf{Condition on the noise level $\epsilon$ --} 
		When applied to linear mixings, the admissible noise levels are lower with \snpab\ than SNPA, which requires $\epsilon < \mathcal{O}\left(\frac{{\beta^{\text{Lin}}_{\mathbf{W}}(\mathbf{W})}^2}{K(\mathbf{W})^2}\right)$. This is expected, and will be confirmed in the numerical experiments of Section~\ref{sec:experiments}, since \snpab\ then performs useless additional projections on non-existing virtual sources. 
		On the other hand, when applied to LQ mixings, the admissible noise levels are larger with \snpab\ than SNPA, since the recovery conditions of SNPA involve $\beta_{\ps(\mathbf{W})}(\ps(\mathbf{W}))$. Moreover, SNPALQ does not need the virtual sources to be present in the data set, while SNPA would require each virtual source to appear as a column of $\bar{\mathbf{X}}$.
		
		
\noindent \textbf{Condition on $K\left(\mathcal{R}^{\norm{\cdot}{2}}_{\ps{(\bar{\mathbf{B}})}}(\mathbf{A})\right)$ --} At each iteration of \snpab, the following condition is required:
		\begingroup\makeatletter\def\f@size{9}\check@mathfonts
		\begin{equation*}
			K\left(\mathcal{R}^{\norm{\cdot}{2}}_{\ps{(\bar{\mathbf{B}})}}(\mathbf{A})\right) 
			\geq 2G K\left(\mathcal{R}^{\norm{\cdot}{2}}_{\ps{(\bar{\mathbf{B}})}}\left((\mathbf{b}_i)_{i\in \lbr s\rbr},(\mathbf{a}_i\odot \mathbf{a}_j)_{\substack{i\leq j \\ i\in \lbr k\rbr \\ j\in \lbr k\rbr}},(\mathbf{b}_i\odot \mathbf{b}_j)_{\substack{i\leq j \\ i\in \lbr s\rbr \\ j\in \lbr s\rbr}},(\mathbf{a}_i\odot \mathbf{b}_j)_{\substack{i\in \lbr k\rbr \\j \in \lbr s\rbr}}\right)\right)
		\end{equation*}\endgroup 
		with $\mathbf{B}$ the columns of $\mathbf{W}$ already extracted by \snpab\ ($\bar{\mathbf{B}}$ their noisy approximation) and $\mathbf{A}$ the other columns of $\mathbf{W}$. This means that at each iteration, a new column of $\mathbf{W}$ must have a larger residual than the already extracted sources and the virtual sources. This condition is the most difficult one to fulfil. In particular the difficulties might arise for a large number of sources, as more terms are present in the right-hand side (see Section~\ref{sec:exp_cond}), or when $\mathbf{W}$ has large entries. However, 
		\begin{itemize}
		
			\item The condition is sufficient but not necessary (see Section~\ref{sec:exp_cond}), making that \snpab\ can work even if it is not fulfiled.
			
			\item Some terms in the right-hand side are or might be negligible, as
			 \[
			 \small{K\left(\mathcal{R}^{\norm{\cdot}{2}}_{\ps{(\bar{\mathbf{B}})}}\left((\mathbf{b}_i)_{i \in \lbr s \rbr}\right)\right) \leq K\left(\mathcal{R}^{\norm{\cdot}{2}}_{{\bar{\mathbf{B}}}}\left((\mathbf{b}_i)_{i \in \lbr s \rbr}\right)\right)}
			 \]
			 and  
			 \[
			 \small{K\left(\mathcal{R}^{\norm{\cdot}{2}}_{\ps{(\bar{\mathbf{B}})}}\left((\mathbf{b}_i\odot \mathbf{b}_j)_{\substack{i\leq j \\ i\in \lbr s\rbr \\ j\in \lbr s\rbr}}\right)\right) \leq K\left(\mathcal{R}^{\norm{\cdot}{2}}_{{(\mathbf{b}_i\odot \mathbf{b}_j)_{i\leq j}}}\left((\mathbf{b}_i\odot \mathbf{b}_j)_{\substack{i\leq j \\ i\in \lbr s\rbr \\ j\in \lbr s\rbr}}\right)\right)}, 
			 \] 
			 and the norm of both right-hand side terms is of the order of the noise level $\epsilon$.
			
			\item The two remaining terms are driven by the correlation of the columns of $\mathbf{W}$. If such a correlation is limited, the condition is expected to be more likely fulfilled.
			
			\item Even if \snpab\ extracts spurious columns of $\bar{\mathbf{X}}$, the post-processing with BF will discard them as it does not need this condition to be satisfied. 
		\end{itemize}{}

	\subsection{Robustness of BF on LQ mixings}\label{sec:robustness_postProc}
	We now study the robustness of the BF step. 
	First, Theorem~\ref{thm:318_simple} below states that BF identifies the columns of $\mathbf{W}$, provided some bounds on the admissible noise levels. The maximum corresponding source estimation error is also given. 
	Then the recovery conditions are discussed.
	
	\subsubsection{Main result}
	The following theorem characterizes the robustness of BF. It is stated in a simplified form by considering $f(\cdot)=\norm{\cdot}{2}$. Its generalized counterpart handling any $f(\cdot)$ satisfying Assumption~\ref{hyp:f} is reported in Appendix \ref{sec:proofs} (see Theorem~\ref{thm:318}).
	\begin{theorem}[Robustness of BF when applied on LQ mixings -- Simplified version]\label{thm:318_simple}
		Let $\bar{\mathbf{X}} = \ps(\mathbf{W}) + \mathbf{N}$, satisfying Definition~\ref{ass:rLSsep} with $\norm{\mathbf{n}_i}{1} \leq \epsilon$ for $i \in \lbr n\rbr$. 
		Let further assume that $\epsilon$ satisfies
		\begin{equation*}
			4\sqrt{d + 2\epsilon (2K(\mathbf{X})+\epsilon)} < \alpha_{\mathbf{W}}(\mathbf{W}), 
		\end{equation*}{}
		with $d = \mathcal{O}\left(\frac{\epsilon}{\alpha_{\pstr(\mathbf{W})}(\mathbf{W})^2}\right)$ (see Equation~\eqref{eq:d} in Appendix \ref{sec:proofs} for the full expression). Then, BF (Algorithm~\ref{alg:algo_PP}) applied on $\bar{\mathbf{X}}$ with $f = \norm{\cdot}{2}$ identifies the columns of $\mathbf{W}$ up to a $\ell_2$ error of $\sqrt{d + 2\epsilon (2K(\mathbf{X}) + \epsilon)}$. 
	\end{theorem}
	
		The proof of the above theorem closely follows the proof of \cite{Arora2016} in the linear case. 
		In particular it extends two definitions of \cite{Arora2016} to LQ mixtures: \emph{i)} the \emph{LQ-robust loners} (see condition~\eqref{eq:post_pro_noisy}) and \emph{ii)} the \emph{canonical columns} which are, roughly speaking, the columns corresponding to the sources (up to the noise) in the data set. 
		It then amounts to show that the LQ-robust loners are approximately the canonical columns, which is done by Lemma~\ref{clm:59} (showing that all the robust loners are close to a canonical column) and \ref{clm:510} (showing that all canonical columns are robust loners). Extracting the robust loners thus enables to approximately recover the sources, as shown by Theorem~\ref{thm:318}.

	\subsubsection{Discussion on the BF recovery conditions}
 Let us discuss	the conditions to ensure the source recovery by BF.  
	
	\noindent \textbf{Condition on $\alpha_{\pstr(\mathbf{W})}(\mathbf{W})$ --} Assuming $\alpha_{\pstr(\mathbf{W})}(\mathbf{W}) > 0$ is the counterpart of Deville's result in  \cite{Deville2019}, which required the family \eqref{eq:linInd}, containing the products up to order four of the sources, to be linearly independent. Here, this condition is turned into a nonnegative independence, which is significantly less restrictive in general. In fact, this condition is most likely a necessary condition for LQ unmixing since $\alpha_{\pstr(\mathbf{W})}(\mathbf{W}) = 0$ implies that some columns of $\mathbf{W}$ can be written as mixtures of other observations. 
	
		
		
	\noindent \textbf{Condition on $\epsilon$ --} The condition		
		$4\sqrt{d + 2\epsilon (2+\epsilon)} < \alpha_{\mathbf{W}}(\mathbf{W})$ with $d = \mathcal{O}\left(\frac{\epsilon}{\alpha_{\pstr(\mathbf{W})}(\mathbf{W})^2}\right)$
		is a limit on the admissible noise level.  
		Roughly speaking, the better some sources can be approximated by a non-negative combination of the other terms of family (\ref{eq:linInd}), the smaller the noise power can be.

	\noindent \textbf{Comparison with \snpab\ --} 
	The conditions of recovery of BF are very mild.  
	For example, in the noiseless case, BF only requires $\alpha_{\pstr(\mathbf{W})}(\mathbf{W}) > 0$, while 
	\snpab\ relies on much stronger conditions. This will be confirmed in the numerical experiments in Section~\ref{sec:experiments}. 
	However, BF is computationally much more demanding (see Section~\ref{sec:complexity}), which motivates its use  as a post-processing for  \snpab.

%
%
%

	\section{Numerical results}
	\label{sec:experiments}
	We here study the behaviors of \snpab\ and BF as a post-processing on simulated yet realistic data sets in the specific applicative context of HS unmixing. The observed mixtures are supposed to follow the Nascimento model \cite{Nascimento2009}. The function $f(\cdot)$ used by the algorithms is here chosen as $f(\cdot) = \norm{\cdot}{2}$ (in-depth study of other choice for $f(\cdot)$ is left for future work).  The code is available from \url{https://bit.ly/SNPALQv1}. 
	
	Section \ref{sec:exp_noiseless} dwells on noiseless mixtures. More precisely, we show in Section~\ref{sec:exp_SNPALQ} that \snpab\ yields very good practical results in this setting, which are enhanced by the BF postprocessing in Section~\ref{sec:exp_postProc}. We further show in Section~\ref{sec:exp_cond} that the condition~\eqref{eq:hypLQ} is only sufficient: it does not need to be fulfilled for \snpab\ to provide reliable results. Lastly, Section~\ref{sec:diff_LQ_bilin} confirms that, beyond the usual Nascimento model involving bilinear mixtures, both \snpab\ and the BF generalize well to LQ models. In Section~\ref{sec:exp_noise}, the robustness of \snpab\ in the presence of noise is studied for different non-linearity levels.
	
	SNPA \cite{Gillis2014} and SPA \cite{Gillis_12_FastandRobust}, two well-known algorithms for near-separable NMF, are used to benchmark the results of the proposed algorithm. 
	
	\subsection{Experimental setting and metrics}

	Experiments are conducted on  realistic LQ near-separable nonnegative data sets $\mathbf{X}$ following Definition~\ref{ass:rLSsep}. The parameters of the model are chosen as follows.
	\begin{itemize}
		\item The primary sources (referred to as \emph{endmember} spectra in the HS literature) defining the columns of $\mathbf{W}$ are defined as spectral signatures extracted from the USGS database\footnote{https://www.usgs.gov/}. They correspond to reflectance spectra associated with materials from diverse origins (such as minerals, soils, and plants) and naturally follow $0 \leq \mathbf{W} \leq 1$. 
		
		\item The matrix $\mathbf{H'}$ is generated in the following way: 
		\begin{itemize}
			
			\item The columns of a first matrix $\acute{\mathbf{H}}$ of the same dimension as $\mathbf{H'}$ are generated randomly using a Dirichlet distribution $\mathcal{D}(\alpha,\ldots,\alpha)$ with $\alpha = 0.5$, which is standard in HS imaging \cite{Nascimento2011}.
			
			\item The $r$ first rows (corresponding to the linear contribution) are multiplied by $1-\nu$, while the remaining rows 
			(corresponding to the virtual endmembers) are multiplied by $\nu$ to enable various non-linearity levels:
			\begin{equation}
				\mathbf{H}' = 
				\begin{bmatrix}
					\acute{\mathbf{H}}^{\lbr r \rbr} \times (1-\nu) \\
					\acute{\mathbf{H}}^{[r+1,\tilde{r}]} \times \nu
				\end{bmatrix}.
				\label{eq:generation_Hprime}
			\end{equation} 
			Note that, acccordingly to the Nasciemento model \cite{Nascimento2009} we consider here, we mostly focus on bilinear mixtures in this experimental section. 
			In this case, the lines of $\acute{\mathbf{H}}^{[r+1,\tilde{r}]}$ corresponding to squared sources $(\mathbf{s}_j \odot \mathbf{s}_j)_{j \in \lbr r \rbr}$ are enforced to be all-zero lines. 
			
			\item The previous transformation does not preserve the sums of the entries in each column of $\acute{\mathbf{H}}$ which were equal to one (Dirichlet distribution). Since the columns are assumed to sum to (at most) one, the last step divides each column of $\mathbf{H}'$ by its $\ell_1$ norm. 
			
		\end{itemize}
		\item The elements of the matrix $\mathbf{N}$ are independently and identically drawn from a centered Gaussian distribution with a variance corresponding to a given signal-to-noise ratio (SNR). 
		
		\item The matrix $\bar{\mathbf{X}}$ is finally created ensuring all the entries to be non-negative: $\bar{\mathbf{X}} = \left[\ps(\mathbf{W})\mathbf{H + N}\right]_+$, where $\left[.\right]_+$ is the elementwise projection on the non-negative orthant.  
	\end{itemize}

	
	The quality of an algorithm is assessed using the minimum spectral angle distance (SAD) between the true and the estimated endmembers:
	\begin{equation*}
		\theta_{\textrm{min}} = \min_{i \in \lbr r \rbr} \textrm{SAD}(\mathbf{w}_i, \mathbf{x}_{\mathcal{K}(i)}), 
	\end{equation*}
	where 
	$\textrm{SAD}(\mathbf{u},\mathbf{v}) 
	= \frac{\mathbf{u}^T\mathbf{v}}{\left\|\mathbf{u}\right\|_2\left\|\mathbf{v}\right\|_2}$ and where the set of indices $\mathcal{K}$ is permuted to  maximize $\theta_{\textrm{min}}$. We consider a perfect separation is achieved if $\theta_{\textrm{min}} > 0.999$.

	\subsection{Numerical results on noiseless mixtures}\label{sec:exp_noiseless}
	
	\subsubsection{Study of \snpab}\label{sec:exp_SNPALQ}
	
	We first explore the behavior of \snpab\ as a function of the number of endmembers $r$ in a noiseless setting. We consider $n = 1000$ mixed pixels with $m = 20$ and the non-linearity parameter is chosen as $\nu = 0.5$. We conducted $100$ Monte-Carlo experiments, each time generating a new dataset.
	
	\begin{figure}[!h]
		\centering
		\begin{tikzpicture}
		\node (img1) {\includegraphics[width=3.5in]{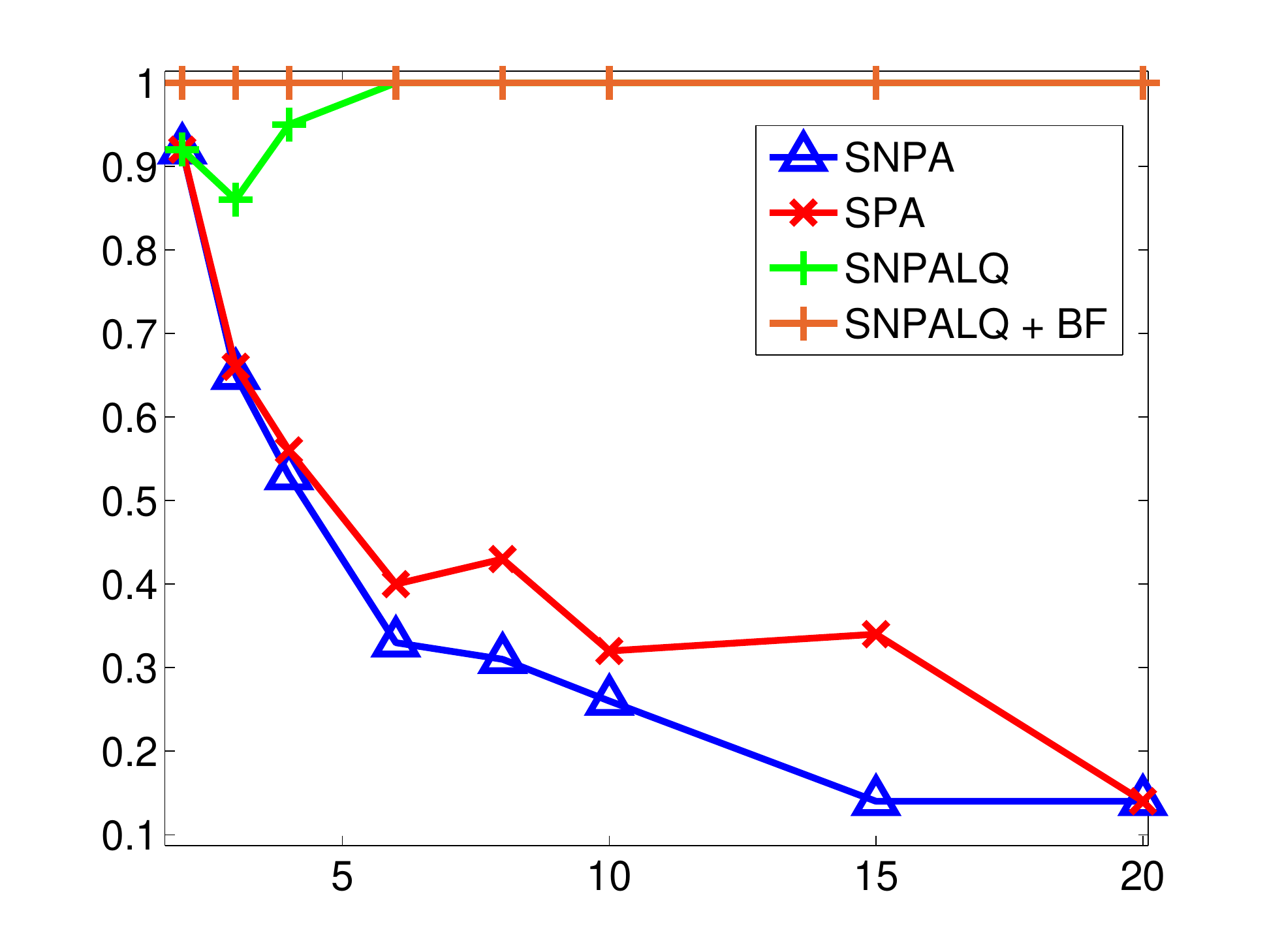}};
		\node[below=of img1, node distance=0cm, yshift=1cm] {\small Number of endmembers $r$};
		\node[left=of img1, node distance=0cm, rotate=90, anchor=center,yshift=-1cm] {\small Percentage of perfect separation};
		\end{tikzpicture}
		\caption{Percentage of experiments in which a perfect separation is achieved, among 100 Monte-Carlo of noiseless bilinear synthetic data sets. The parameters are: $m = 20$ observations, $n = 1000$ pixels and $\nu = 0.5$.}
		\label{fig_res_noiselessLQ}
	\end{figure}
	
	 Fig.~\ref{fig_res_noiselessLQ} reports the percentage of full recovery by the different algorithms. In this experiment, \snpab\ obtains much better results than SNPA or SPA, and achieves in more than $90 \%$ of the experiments a perfect separation. While an initial improvement of the results when $r$ increases might look surprising, it is probably not to be linked directly with the $r$ value itself, but rather with the generated $\mathbf{H}'$. Indeed, when $r$ is small, the data columns are more spread within the convex hull formed by the endmembers and these are therefore more difficult to extract; 
	see Section~\ref{sec:exp_cond}. 
	
	On the other hand, the results of SNPA and SPA are rather bad on this non-linear data set, and deteriorate quickly when the number of endmembers increases. While both algorithms obtain close results, it is interesting to note that SPA becomes worse than SNPA when $r$ becomes closer to $m$, which is expected as SNPA has an interest over SPA mainly when the matrices $\mathbf{W}$ are either not full-rank or ill conditioned \cite{Gillis2014}.
	
	\subsubsection{Study of BF as a post-processing step}
	\label{sec:exp_postProc}
	In this section, we analyze the relevance of the introduction of BF as a post-processing conducted after \snpab. Figure~\ref{fig_res_noiselessLQ} displays in orange the separation quality when applying the BF to \snpab. This result show that BF enables to achieve perfect results for all experiments by improving \snpab\ results, especially for low $r$ values.\\
	A natural question is however the \emph{cost} of such a post-processing; see Section~\ref{sec:post_proc}. Table~\ref{fig:tableau_rEmp} thus displays the number of columns extracted by \snpab\ ($2$nd line) as a function of the actual number $r$ of sources. These columns are the input of the BF, and therefore they determine its computational time. Interestingly enough, on average, \snpab\ does not need to extract more than $r+1$ components to extract all the columns 
	of $\mathbf{W}$. 
	As such, the post-processing step is applied on a small number of columns of $\bar{\mathbf{X}}$ and is cheap.\\
	

	\begin{table}
		\centering
				\caption{\emph{$1$st line}: actual $r$ value. \emph{$2$nd and $3$rd lines}: average number $|\mathcal{K}|$, over 100 Monte-Carlo experiments, of endmembers extracted by \snpab\  and \snpab+BF.}
		\begin{tabular}{ |c|c|cccccccc| } 
			\hline
			\multicolumn{2}{|c|}{$r$} & 2 & 3 & 4 & 6 & 8 & 10 & 15 & 20\\ 
			\hline
			\multirow{2}{*}{$|\mathcal{K}|$} & \snpab & 2.08 & 3.16 & 4.05 & 6 & 8 & 10 & 15 & 20 \\
			\cline{2-10}
		       &   \snpab+BF & 2 & 3 & 4 & 6 & 8 & 10 & 15 & 20\\
			\hline
		\end{tabular}
		\label{fig:tableau_rEmp}
	\end{table}

	\subsubsection{Discussion about condition~(\ref{eq:hypLQ_simple})}\label{sec:exp_cond}
	
	The introduction of condition \eqref{eq:hypLQ_simple} is one of the major difference compared to the linear case, for which it does not appear explicitly\footnote{More exactly, in the linear case this condition is replaced by one on the admissible noise levels; see \cite{Gillis2014}.}. As such, we here aim at discussing its validity on real noiseless HS data. It is important to notice that this context might be favorable, since $\mathbf{W}$ naturally fulfills $0 \leq \mathbf{W} \leq 1$. 	To do so, we propose the following complementary experiment: for each of the 100 Monte-Carlo experiments, we draw a new $\mathbf{W}$ matrix from the USGS database and split the columns of $\mathbf{W}$ into two disjoints matrices: $\mathbf{A} = [(\mathbf{w}_i)_{i \in \mathcal{J}, \mathcal{J} \subseteq \lbr r \rbr}]$ and $\mathbf{B} = [(\mathbf{w}_i)_{i \in \lbr r \rbr \setminus\mathcal{J}}]$. We then check whether these matrices $\mathbf{A}$ and $\mathbf{B}$ fulfill condition \eqref{eq:hypLQ_simple}.
	By repeating the process with all the possible $\mathbf{A}$ and $\mathbf{B}$, we can thus obtain a percentage of subsets $\mathbf{A}$ and $\mathbf{B}$ for which condition~(\ref{eq:hypLQ_simple}) is fulfilled in the USGS database. In this experiment, we consider $n=1\ 000$ samples with $m=50$ observations. 
	
	\begin{figure}[!h]
		\centering
		\begin{tikzpicture}
		\node (img1) {\includegraphics[width=3.5in]{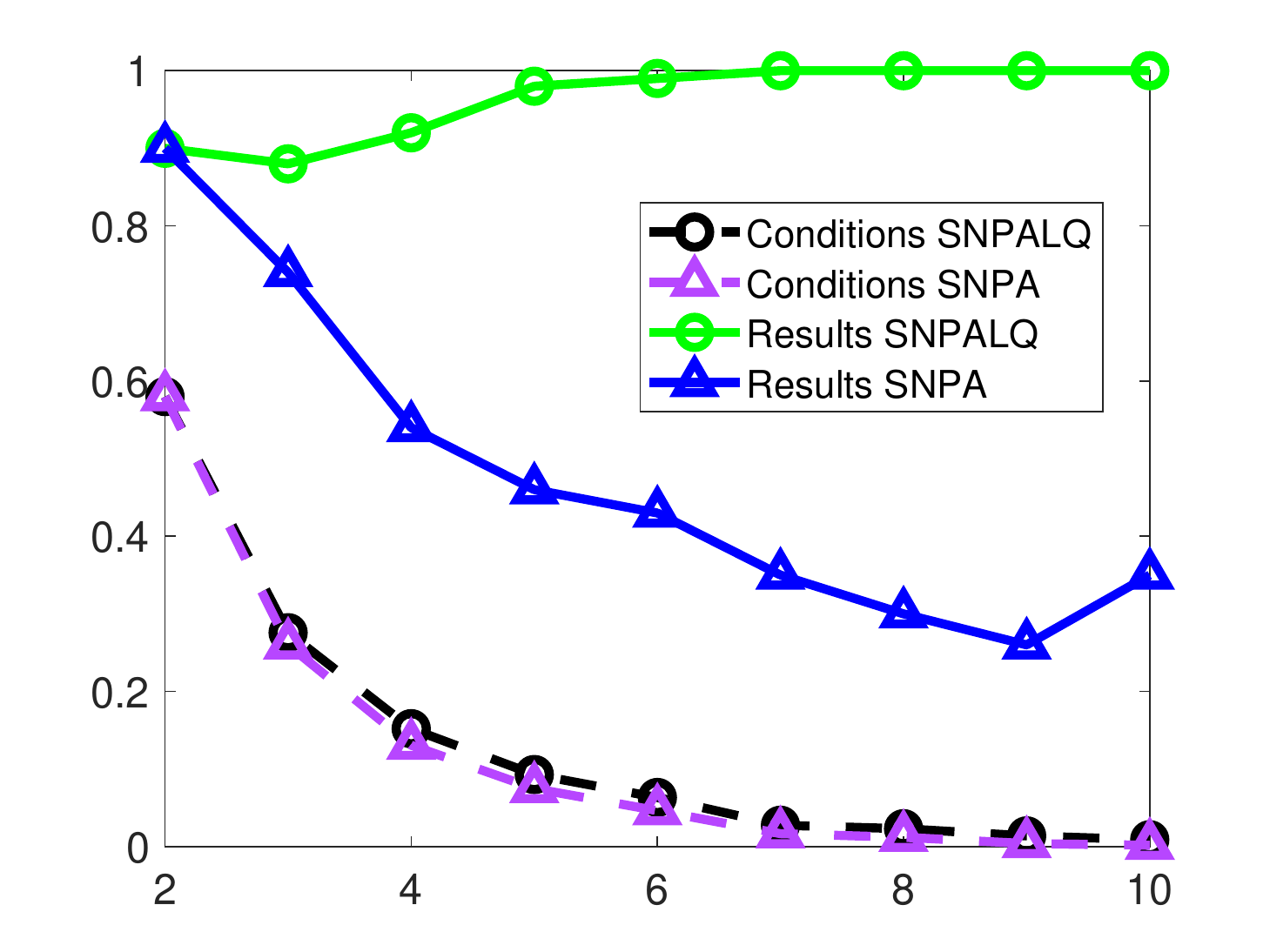}};
		\node[below=of img1, node distance=0cm, yshift=1cm] {\small Number of endmembers $r$};
		\node[left=of img1, node distance=0cm, rotate=90, anchor=center,yshift=-1cm] {\small \small Proportion of $\mathbf{W}$ for which \ref{eq:hypLQ_simple} is fulfilled};
		\end{tikzpicture}
		\caption{Comparison between theoretical conditions ensuring endmember recovery by \snpab\ (resp. SNPA) and actual results. The dashed line correspond to the percentage of submatrices $\mathbf{A}$ and $\mathbf{B}$ for which condition \ref{eq:hypLQ_simple} is fulfilled and the plain lines correspond to the actual proportion of perfect recovery by \snpab\ (resp. SNPA).}
		\label{fig_testCondProj}
	\end{figure}
	Figure~\ref{fig_testCondProj} (dashed lines) depicts, as a function of the number of endmembers $r$, the proportion of the different realizations of $\mathbf{W}$ for which condition \ref{eq:hypLQ_simple} is fulfilled. This proportion of sub-matrices $\mathbf{A}$ and $\mathbf{B}$ fulfilling condition \eqref{eq:hypLQ_simple} decreases with $r$, which was expected as the number of elements in the right-hand side increases. Then, as exemplified in Figure~\ref{fig_SNPA_dontworks_schema}, the condition is observed to be slightly less restrictive in general for \snpab\ than for SNPA. 
	
	Most importantly, the results become quite bad for relatively small $r$ values: for $r = 10$, the condition is fulfilled for only slightly more than $5\%$ of the tested subsets $\mathbf{A}$ and $\mathbf{B}$. Thus it might be surprising that \snpab\ algorithm achieves perfect results in almost all experiments. Such a discrepancy appears because Condition~\ref{eq:hypLQ_simple} is only a sufficient condition. The reason is twofold: 
	\begin{itemize}
	
		\item Condition~\ref{eq:hypLQ_simple} considers all the possible ways to split the matrix $\mathbf{W}$ into two submatrices $\mathbf{A}$ and $\mathbf{B}$. This allows to prove the recovery of \snpab\, regardless of the order in which the columns of $\mathbf{W}$ are extracted. However, in practice, \snpab\, only needs this condition to be satisfied for the order in which it extracts the indices, and hence it is in general much milder. 
		
		\item  The virtual endmembers typically do not not appear purely, which  makes the condition too conservative (recall that this condition is not necessary in the linear case; see Theorem~\ref{thm:robustness_lin_simple}). In other words, Condition~\ref{eq:hypLQ_simple} considers the worst case scenario for any possible mixing matrix $\mathbf{H}$ while, in practice, the non-linearity can be mild. 
		
	\end{itemize}{}
	
	In summary, while Condition~\ref{eq:hypLQ_simple} might seem restrictive,  \snpab\ can yield excellent results in settings in which it is not fulfilled. In particular, it could be of interest to include the non-linearity level $\nu$ in a study of necessary conditions for \snpab, which is left for future work.

	\subsubsection{Differences between LQ and bilinear mixtures}\label{sec:diff_LQ_bilin}
	
	To conclude this section, we now study the slight differences of behavior of \snpab+BF when analyzing LQ or bilinear mixtures. The experiment settings are similar to the one associated with Figure~\ref{fig_res_noiselessLQ}. We consider $100$ Monte-Carlo runs of $n = 1000$ pixels with $m = 20$ and the non-linearity parameter is chosen as $\nu = 0.5$. The difference is that the data sets are now LQ, instead of bilinear: squared sources are included in the mixtures.
	
	Figure~\ref{fig_res_LQ_bilin} displays the results obtained by two variants of SNPALQ + BF: 
	\begin{itemize}
	
		\item The orange curve ($\triangleright$ markers) displays the results of the algorithm when no squared sources are included in the projection steps  of both \snpab\ and the BF;
		
		\item The yellow curve ($+$ markers) displays the results of the algorithm when squared sources are included in the projection steps of both \snpab\ and the BF.
		
	\end{itemize}
	As can be seen with the orange curve, \snpab+BF (LQ version) almost perfectly handles LQ mixtures, similarly to what was shown above for bilinear ones. The slightly deteriorated results (which are still much better than the ones obtained by the linear algorithms) shown with the yellow curve ($+$ markers) were expected: by not incorporating the presence of squared sources during the unxming process, the algorithm introduces errors. As such, the user of \snpab+BF should use as much as possible prior knowledge to determine beforehand whether the data set results from bilinear or LQ mixings.
	\begin{figure}[!h]
		\centering
		\begin{tikzpicture}
		\node (img1) {\includegraphics[width=3.5in]{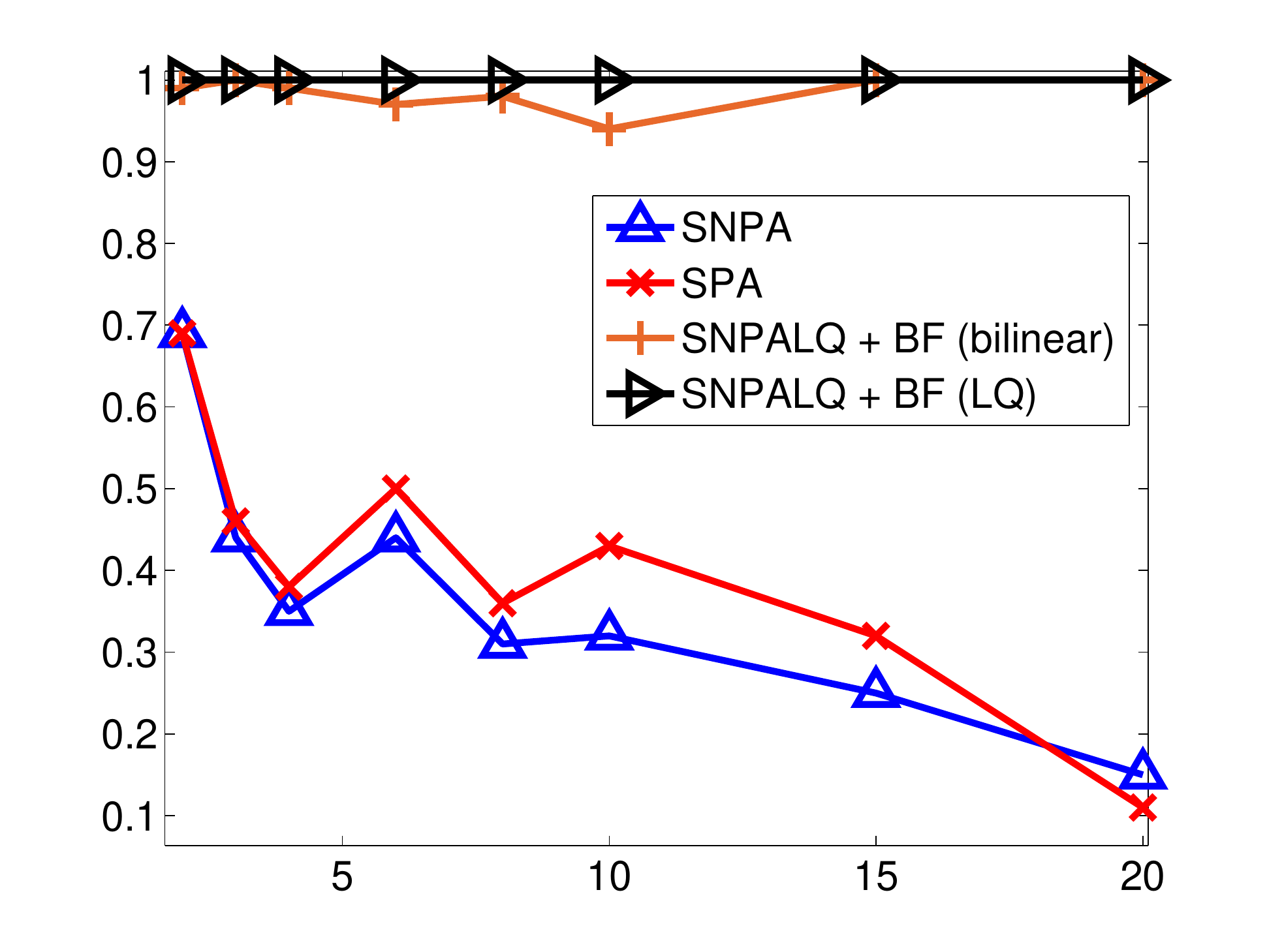}};
		\node[below=of img1, node distance=0cm, yshift=1cm] {\small Number of endmembers $r$};
		\node[left=of img1, node distance=0cm, rotate=90, anchor=center,yshift=-1cm] {\small Percentage of perfect separation};
		\end{tikzpicture}
		\caption{Percentage of experiments in which a perfect separation is achieved. There are 100 Monte-Carlo experiments, with $m = 20$ observations and $n = 1000$ pixels. The non-linearity parameter is $\nu = 0.5$, and the mixtures are linear quadratic: they include squared sources. In addition to the results of SNPALQ bilinear (in which the projection step does not include the source auto-products) and SNPALQ (LQ), the results of SNPA and SPA are included. 
		}
		\label{fig_res_LQ_bilin}
	\end{figure}

	\subsection{Robustness study: noisy mixtures}\label{sec:exp_noise}
	\label{sec:exp_NL_noisy}
	
	The impact of the noise and non-linearity levels is now studied. We generated bilinear data sets $\bar{\mathbf{X}}$ with $7$ different SNR levels and $12$ values for the non-linearity parameter $\nu$. For each pair of SNR and $\nu$ values, 24 Monte-Carlo experiments are conducted on nonlinear mixtures characterized by $m = 50$ spectral bands, $r = 10$ endmembers and $n = 1\ 000$ pixels. 

	\begin{figure*}[t]
		\centering
		\subfloat[]{
			\begin{tikzpicture}
			\node (img1) {\includegraphics[width=2.5in]{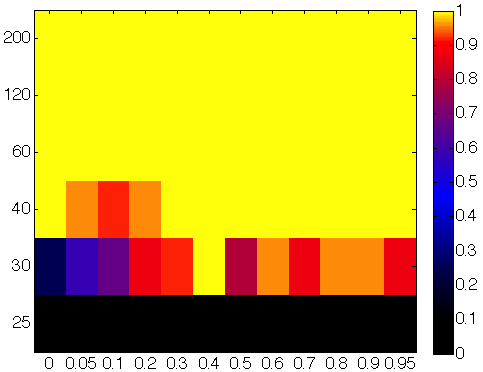}};
			\node[below=of img1, node distance=0cm, yshift=1cm] {\small Non-linearity parameter $\nu$};
			\node[left=of img1, node distance=0cm, rotate=90, anchor=center,yshift=-0.7cm] {\small SNR (dB)};
			\end{tikzpicture}
			\label{fig:SNR_delta_SNPAB}}
		\hfil
		\subfloat[]{
			\begin{tikzpicture}
			\node (img1) {\includegraphics[width=2.5in]{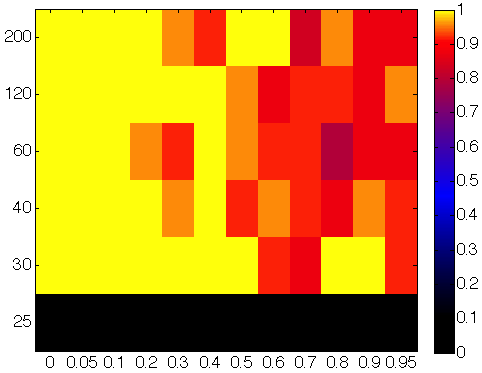}};
			\node[below=of img1, node distance=0cm, yshift=1cm] {\small Non-linearity parameter $\nu$};
			\node[left=of img1, node distance=0cm, rotate=90, anchor=center,yshift=-0.7cm] {\small SNR (dB)};
			\end{tikzpicture}
			\label{fig:SNR_delta_SNPA}}
		\caption{As a function of SNR and non-linearity level $\nu$, percentage of perfect separation using \emph{a)} \snpab+BF, \emph{b)} SNPA. } 
		\label{fig:exp_noisyLQ}
	\end{figure*}
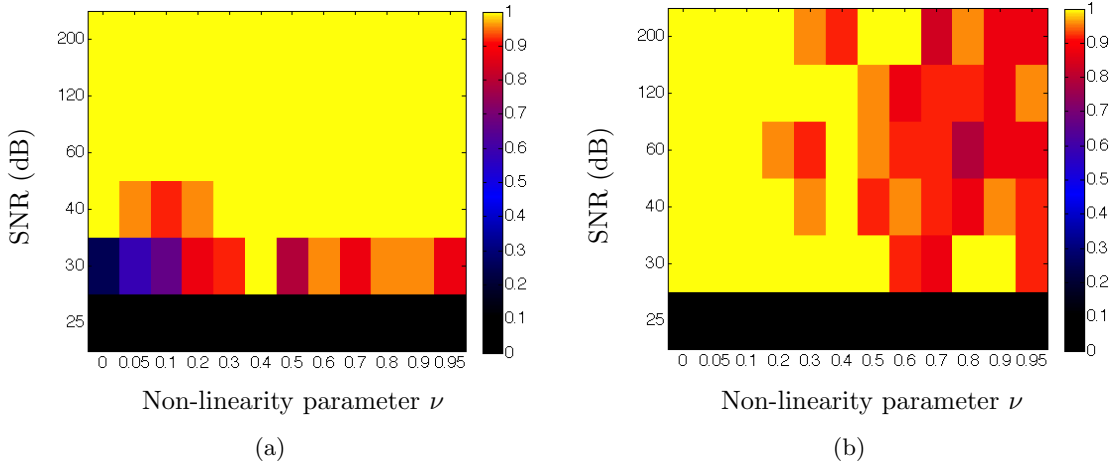


	Figure~\ref{fig:exp_noisyLQ} depicts the recovery performances of \snpab\ and SNPA. For low non-linearity levels, the mixtures approximately follow the LMM: in agreement with their robustness guarantees, the results of \snpab+BF are then perfect when the SNR is high ($\textrm{SNR} \geq 40$dB); see upper-left corner of Figure \ref{fig:SNR_delta_SNPAB}. However, \snpab+BF performs worse than SNPA in the presence of a stronger noise ($\textrm{SNR} \in [25\textrm{dB},30\textrm{dB}]$), which is expected as it projects the residual onto non-existing virtual endmembers, leading to a loss of information (the norm of the residual decreases faster).
	

	\snpab+BF shows its benefit over SNPA when the non-linearity level increases and  the noise level is not too large (upper-right corner of the figures). More precisely, when $\nu \geq 0.3$ and $\textrm{SNR} \geq 40$ dB, \snpab+BF always obtains a perfect recovery, which represents a significant improvement over SNPA, up to $20 \%$. In the lower-right corner of the figure, when the SNR decreases, the results of both algorithms deteriorate as the problem is highly difficult. 

	\section{Conclusion}
	
	In this paper, we have considered the problem of linear-quadratic blind source separation, under the near-separable assumption which requires the primary sources to appear purely in the data set.  
	We  first introduced \snpab, an extension of SNPA~\cite{Gillis2014}, which takes into account the presence of quadratic terms in the projection step. 
	\snpab\ is guaranteed to recover the sources for linear-quadratic under appropriate conditions. 
	We then introduced a second algorithm, namely brute-force (BF), and extension of the algorithm of Arora et al~\cite{Arora2016}, which provably recovers the sources under milder conditions than \snpab. 
	It is recommended to use BF as a post-processing of \snpab\ (denoted by \snpab+BF)  due to its high computational cost. 
    Finally, we illustrated the performance of \snpab\ and \snpab+BF in various settings, and showed that they obtained good separation results on realistic hyperspectral data sets, and for  various experimental settings, including linear, bilinear and linear quadratic mixtures. Improving \snpab+BF results for low SNR while still alleviating recovery conditions of both algorithms is left for future work.

	\appendix

			
			
			
	
	\section{A few useful results of \cite{Gillis_12_FastandRobust,Gillis2014}}
	\begin{lemma}[Lemma 3.3 in \cite{Gillis2014}]\label{lem:gillis_33}
		For any $\mathbf{B} \in \mathbb{R}^{m\times s}$, $\mathbf{x} \in \mathbb{R}^{m}$, and $f$ satisfying Assumption~\ref{hyp:f}, we have
			$\norm{\mathcal{R}_\mathbf{B}^f(\mathbf{x})}{2} \leq \sqrt{\frac{L}{\mu}} \norm{\mathbf{x}}{2}$. 
	\end{lemma}
	
	\begin{lemma}[Lemma 3.4 in \cite{Gillis2014}]\label{lem:gillis_34}
		Let $\mathbf{B} \in \mathbb{R}^{m\times s}$ and $\mathbf{B} = \bar{\mathbf{B}} + \mathbf{N}$ with $\norm{\mathbf{N}}{1,2} \leq \epsb$, and $f$ satisfy Assumption~\ref{hyp:f}. Then, 
			$\max_j \norm{\mathcal{R}^f_{\bar{\mathbf{B}}}(b_j)}{2} \leq \sqrt{\frac{L}{\mu}}\epsb$. 
	\end{lemma}

	\begin{lemma}[Lemma 3.7 in \cite{Gillis2014}]\label{lem:gillis_37}
		Let $\mathbf{A} \in \mathbb{R}^{m\times k}$, $\mathbf{B}$, and $\bar{\mathbf{B}}\in \mathbb{R}^{m\times s}$ satisfy $\norm{\mathbf{B} - \bar{\mathbf{B}}}{1,2} \leq \epsb$, and let $f$ satisfy Assumption~\ref{hyp:f}. Then, 
		\begin{equation*}
		    \nu \left( \mathcal{R}^f_{\bar{\mathbf{B}}}(\mathbf{A}) \right) \geq \alpha_{[\mathbf{A},\mathbf{B}]}({[\mathbf{A},\mathbf{B}]}) - \min(s,2)\epsb.
		\end{equation*}
	\end{lemma}

	\begin{lemma}[Lemma 3.13 in \cite{Gillis2014}]\label{lem:gillis_313}
		Let $\mathbf{B} \in \mathbb{R}^{m\times s}$, $\mathbf{A} \in \mathbb{R}^{m\times k}$, $\mathbf{n} \in \mathbb{R}^m$, and $z \in \Delta^k$, and let $f$ satisfy Assumption~\ref{hyp:f}. Then,
		\begin{equation*}
			f\left(\mathcal{R}_\mathbf{B}^f(\mathbf{Az + n})\right) \leq f\left(\mathcal{R}_\mathbf{B}^f(\mathbf{Az + n})\right) \quad \text{and} \quad f\left(\mathcal{R}_\mathbf{B}^f(\mathbf{Az + n})\right) \geq f\left(\mathcal{R}_\mathbf{B}^f(\mathbf{A}) \mathbf{z + n}\right). 
		\end{equation*}{}
	\end{lemma}{}
	
	\begin{lemma}[Lemma 3 in \cite{Gillis_12_FastandRobust}]\label{lem:gillis_314}
		Let the function $f$ satisfy Assumption~\ref{hyp:f}. Then, for any $\norm{\mathbf{x}}{2} \leq K$ and $\norm{\mathbf{n}}{2} \leq \epsilon \leq K$, 
		\begin{equation*}
			f(\mathbf{x}) - \epsilon K L \leq f(\mathbf{x + n}) \leq f(\mathbf{x}) + \frac{3}{2}\epsilon K L.
		\end{equation*}{}
	\end{lemma}{}
	
	\begin{lemma}[Lemma 2 in \cite{Gillis_12_FastandRobust}]\label{lem:gillis_315}
		Let $\mathbf{Z = [P,Q]}$, where $\mathbf{P} \in \mathbb{R}^{m\times k}$ and $\mathbf{Q} \in \mathbb{R}^{m\times s}$, and let $f$ satisfy Assumption~\ref{hyp:f}. If $\nu(\mathbf{P}) > 2 \sqrt{\frac{L}{\mu}}K(\mathbf{Q})$, then, for any $0 \leq \delta \leq \frac{1}{2}$,
		\begin{equation*}
			f^* = \max_{x\in \Delta} f(\mathbf{Zx}) \text{ such that } x_i \leq 1-\delta \text{ for } 1 \leq i \leq k
		\end{equation*}{}
		satisfies
			$f^* \leq \max_{i}f(p_i) - \frac{1}{2}\mu (1-\delta)\omega(\mathbf{P})^2$. 
	\end{lemma}{}
	
	\begin{lemma}[Lemma B.1 in \cite{Gillis2014}]\label{lem:gillis_B1}
		Let $\mathbf{x} \in \mathbb{R}^m$, $\mathbf{B}$ and $\bar{\mathbf{B}} \in \mathbb{R}^{m\times s}$ satisfy the inequality $\norm{\mathbf{B} - \bar{\mathbf{B}}}{1,2} \leq \epsb \leq \norm{\mathbf{B}}{1,2}$, and let $f$ satisfy Assumption~\ref{hyp:f}. Then,
		\begin{equation*}
			\norm{\mathcal{R}^f_{\mathbf{B}}(\mathbf{x}) - \mathcal{R}^f_{\bar{\mathbf{B}}}(\mathbf{x})}{2}^2 \leq 12 \frac{L}{\mu}\bar{\epsilon}\norm{\mathbf{B} - \bar{\mathbf{B}}}{1,2}.
		\end{equation*}{}
	\end{lemma}{}
	
	\begin{lemma}[Lemma B.2 in \cite{Gillis2014}]\label{lem:gillis_B2}
		Let $\mathbf{x,y} \in \mathbb{R}^m$, $\mathbf{B}$, and $\bar{\mathbf{B}} \in \mathbb{R}^{m\times s}$ be such that $\norm{\mathbf{B} - \bar{\mathbf{B}}}{1,2} \leq \epsb \leq \norm{\mathbf{B}}{1,2}$, and let $f$ satisfy Assumption~\ref{hyp:f}. Then,
		\begin{equation*}
			\norm{\mathcal{R}^f_{\bar{\mathbf{B}}}(\mathbf{x}) - \mathcal{R}^f_{\bar{\mathbf{B}}}(\mathbf{y})}{2}^2 
			\geq \norm{\mathcal{R}^f_{{\mathbf{B}}}(\mathbf{x}) - \mathcal{R}^f_{{\mathbf{B}}}(\mathbf{y})}{2}^2 
			\geq 4 \sqrt{\frac{3KL}{\mu}\epsb}.
		\end{equation*}{}
	\end{lemma}{}

	\section{Proofs of our main results: \snpab\ and BF are provably robust in the presence of noise}   
	\label{sec:proofs}
	
	In this section, we study the robustness of \snpab\ (Section~\ref{sec:proofSNPALQ} and \ref{sec:proof_SNPALQforLQ}) and of BF (Section~\ref{sec:proofBF}). But before, let us introduce a few additional notations. For two matrices $\mathbf{A} \in \mathbb{R}^{m\times r_A}$ and $\mathbf{B} \in \mathbb{R}^{m\times r_B}$, we define
	\begin{equation*}
		\alpha_{\mathbf{B}}(\mathbf{A}) = \min_{\substack{j\in \lbr r_A \rbr \\ \mathbf{x}\in \Delta}} \norm{\mathbf{a}_j - \mathbf{B}_{\setminus\{j\}}\mathbf{x}}{2}.
	\end{equation*}  
	For instance, in the special case $\mathbf{A} = \mathbf{B}$, $\alpha_{\mathbf{A}}(\mathbf{A})$ is the minimum distance between a column of $\mathbf{A}$ and the convex hull formed by the other columns of $\mathbf{A}$ and the origin. Let us also denote 
		$\nu(\mathbf{A}) = \min_{i \in \lbr r_A\rbr}\norm{\mathbf{a}_i}{2}$,
		$\gamma(\mathbf{A}) = \min_{\substack{i,j \in \lbr r_A\rbr \\ i \neq j}}\norm{\mathbf{a}_i - \mathbf{a}_j}{2}$,
		$\omega(\mathbf{A}) = \min\left(\nu(\mathbf{A}),\frac{\sqrt{2}}{2}\gamma(\mathbf{A})\right)$, and 
		$\Omega(\mathbf{A}) = \min \left( \frac{K(\mathbf{A})}{2}\sqrt{\frac{\mu}{L}}\left[1-\frac{1}{G}\right], \gamma(\mathbf{A})\right)$, 
	where $\mu, L$ and $G$ are some constants that will be specified later.

	\subsection{Proof of \snpab\ robustness for linear mixtures} \label{sec:proofSNPALQ}
	
    The proof is conducted by induction. We first derive a few useful lemmas, which are then used to prove the induction step in Theorem~\ref{th:th316}. The main result is then stated in Theorem~\ref{thm:robustness_lin}.

	\begin{lemma}[Extension of \cite{Gillis2014}-Lemma~3.5]\label{lem:lem35}
		Let $\mathbf{A} \in \mathbb{R}^{m \times k}$, $\mathbf{B} \in \mathbb{R}^{m \times s}$ and $f$ satisfy Assumption~\ref{hyp:f}. Then,
		\begin{equation}
		\nu\left( \resi{\ps(\mathbf{B})}{\mathbf{A}}\right) \geq \alpha_{\ps(\mathbf{[A,B]})}([\mathbf{A,B}]).
		\end{equation}
	\end{lemma}

	\begin{proof} 
		The result follows from the definitions of these quantities: 
		\begin{equation*}
			\begin{split}
				\alpha_{\ps(\mathbf{[A,B]})}([\mathbf{A,B}]) 
				&= \min_{\substack{j\in \lbr k+s \rbr \\ y\in \Delta}} \norm{[\mathbf{A,B}]_j - \ps(\mathbf{[A,B]})_{ \setminus\{j\}} y}{2} \\
				&\leq \min_{\substack{j\in \lbr k \rbr \\ y\in \Delta}} \norm{\mathbf{a}_j - \ps(\mathbf{[A,B]})_{ \setminus\{j\}} y}{2}\\
				&\leq \min_{\substack{j\in \lbr k \rbr \\ y\in \Delta}} \norm{\mathbf{a}_j - \ps(\mathbf{B}) y}{2}\\
				&\leq \nu\left( \resi{\ps(\mathbf{B})}{\mathbf{A}}\right). 
			\end{split}
		\end{equation*}
	\end{proof}

	\begin{lemma}[Extension of \cite{Gillis2014}-Lemma~3.6]\label{lem:lem36}
		Let $\mathbf{Z}$ and $\bar{\mathbf{Z}} \in \mathbb{R}^{m\times r}$ satisfy $\norm{\mathbf{Z} - \bar{\mathbf{Z}}}{1,2} \leq \epsb$. Then,
		\begin{equation*}
			\alpha_{\ps(\bar{\mathbf{Z}})}(\bar{\mathbf{Z}}) 
			\geq \alpha_{\ps(\mathbf{Z})}(\mathbf{Z}) - \epsb(1 + \max(1,2K(\mathbf{Z}) + \epsb)).
		\end{equation*}
	\end{lemma}

	\begin{proof} 
		We have 
		\begin{equation*}
			\begin{split}
				\alpha_{\ps(\bar{\mathbf{Z}})}(\bar{\mathbf{Z}}) 
				&= \min_{\substack{j\in \lbr r \rbr \\ y\in \Delta}} \norm{\bar{\mathbf{z}}_j - \ps(\bar{\mathbf{Z}})_{ \setminus\{j\}}y}{2}\\
				&= \min_{\substack{j\in \lbr r \rbr \\ y\in \Delta}}\norm{\mathbf{z}_j - \mathbf{n}_j - \ps(\mathbf{Z - N})_{ \setminus\{j\}}y}{2}\\
				&\begin{split}
    				=\min_{\substack{j\in \lbr r \rbr \\ y\in \Delta}}
    				\big\| &\mathbf{z}_j - \mathbf{n}_j - \\ &\left(\ps(\mathbf{Z})_{ \setminus\{j\}} - [\mathbf{N},\mathbf{0}]_{\setminus\{j\}} - 2 [\mathbf{0},(\mathbf{z}_k\odot \mathbf{n}_l)_{l\leq k}]_{ \setminus\{j\}} 
    				+ [\mathbf{0},(\mathbf{n}_k\odot \mathbf{n}_l)_{l\leq k}]_{ \setminus\{j\}}\right)y \big\|_2
				\end{split}
				\\
				& \geq \min_{\substack{j\in \lbr r \rbr \\ y\in \Delta}} \norm{\mathbf{z}_j - \ps(\mathbf{Z})_{ \setminus\{j\}}y}{2} - \norm{\mathbf{n}_j}{2} - \norm{\left[\mathbf{N},2 (\mathbf{z}_k\odot \mathbf{n}_l)_{l\leq k} - (\mathbf{n}_k\odot \mathbf{n}_l)_{l\leq k}\right]_{ \setminus\{j\}}y}{2}\\
				& \geq \alpha_{\ps(\mathbf{Z})}(\mathbf{Z}) - \epsb(1 + \max(1,2K(\mathbf{Z}) + \epsb)).
			\end{split}
		\end{equation*}
	\end{proof}
	
	\begin{corollary}[Extension of \cite{Gillis2014}-Corollary 3.7]\label{cor:Cor37}
		Let $\mathbf{A} \in \mathbb{R}^{m\times k}$, $\mathbf{B}$ and $\bar{\mathbf{B}} \in \mathbb{R}^{m\times s}$ satisfy  \mbox{$\norm{\mathbf{B}-\bar{\mathbf{B}}}{1,2} < C\epsilon$}, and let $f$ satisfy Assumption~\ref{hyp:f}. Then, 
		\begin{equation}
		\nu \left( \resi{\ps(\bar{\mathbf{B}})}{\mathbf{A}}     \right) > \alpha_{\ps([\mathbf{A},\mathbf{B}])}([\mathbf{A},\mathbf{B}]) - C\epsilon(1 + \max(1, 2K(\mathbf{[\mathbf{A},\mathbf{B}]}) + C\epsilon)).
		\end{equation}
	\end{corollary}
	
	\begin{proof} 
		The cases $s = 0$ and $s = 1$, with $s$ the number of columns of $\mathbf{B}$, are direct extensions of Lemma~\ref{lem:gillis_37} as $\ps(\bar{\mathbf{B}}) = \bar{\mathbf{B}}$ and $\alpha_{\ps([\mathbf{A},\mathbf{B}])}([\mathbf{A},\mathbf{B}]) \leq \alpha_{[\mathbf{A},\mathbf{B}]}([\mathbf{A},\mathbf{B}])$. 
		
		\emph{For the case $s > 1$}, Lemma~\ref{lem:lem35} and \ref{lem:lem36} imply that:
		\begin{equation*}
			\begin{split}
				\nu\left( \resi{\ps(\bar{\mathbf{B}})}{\mathbf{A}}\right) 
				&\geq \alpha_{\ps([\mathbf{A},\bar{\mathbf{B}}])}([\mathbf{A},\bar{\mathbf{B}}])\\
				&\geq \alpha_{\ps([\mathbf{A},\mathbf{B}]}([\mathbf{A},\mathbf{B}]) - C\epsilon(1 + \max(1, 2K(\mathbf{[\mathbf{A},\mathbf{B}]}) + C\epsilon)). 
			\end{split}
		\end{equation*}
	\end{proof}

	\begin{lemma}[Extension of \cite{Gillis2014}-Lemma B-3]\label{lem:lemB3}
		Let $\mathbf{A}\in \mathbb{R}^{m\times k}$, $\mathbf{B}\in \mathbb{R}^{m\times s}$, $\bar{\mathbf{B}}\in \mathbb{R}^{m\times s}$, $f$ satisfy Assumption~\ref{hyp:f}, and let $\epsb$ be such that $\norm{\mathbf{B} - \bar{\mathbf{B}}}{1,2} \leq \epsb \leq \norm{\mathbf{B}}{1,2}$ and {$\epsb \leq -1 + \sqrt{1+K}$}. Then, 
		\begin{equation*}
			\omega(\resi{\ps(\bar{\mathbf{B}})}{\mathbf{A}}) \geq \beta^{\text{Lin}}_{\ps([\mathbf{A},\mathbf{B}])}([\mathbf{A},\mathbf{B}]) - 2\sqrt{6\frac{L}{\mu}\norm{B}{1,2}\epsb(2+\epsb)}.
		\end{equation*}
	\end{lemma}

	\begin{proof} 
		For all $i$, we have 
		\begin{equation*}
			\begin{split}
				\beta_{\ps([\mathbf{A,B}])}([\mathbf{A,B}]) - \norm{\resi{\ps(\bar{\mathbf{B}})}{\mathbf{a}_i}}{2}
				&\leq \norm{\resi{\ps(\mathbf{B})}{\mathbf{a}_i}}{2} - \norm{\resi{\ps(\bar{\mathbf{B}})}{\mathbf{a}_i}}{2}\\
				& \leq \norm{\resi{\ps(\mathbf{B})}{\mathbf{a}_i}  -  \resi{\ps(\bar{\mathbf{B}})}{\mathbf{a}_i}}{2}\\
				& \leq \sqrt{12 \frac{L}{\mu}\epsb(\epsb+2)\norm{\ps(\mathbf{B})}{1,2}},
			\end{split}
		\end{equation*}
		where the last line is obtained using Lemma~\ref{lem:gillis_B1} by noting that 
		\begin{equation*}
			\norm{\ps(\mathbf{B}) - \ps(\bar{\mathbf{B}})}{2} \leq \epsb(\epsb + 2) \leq \norm{\ps(\mathbf{B})}{1,2}.
		\end{equation*}
		
		Furthermore, for all $i,j$, 
		\begin{equation*}
			\begin{split}
				&\frac{1}{\sqrt{2}}\norm{\resi{\ps(\bar{\mathbf{B}})}{\mathbf{a}_i}  -  \resi{\ps(\bar{\mathbf{B}})}{\mathbf{a}_j}}{2}\\
				& 
				\geq \frac{1}{\sqrt{2}}\norm{\resi{\ps(\mathbf{B})}{\mathbf{a}_i}  -  \resi{\ps(\mathbf{B})}{\mathbf{a}_j}}{2} -\frac{4}{\sqrt{2}}\sqrt{3\norm{\ps(\mathbf{B})}{1,2}\frac{L}{\mu}\epsb(\epsb + 2)}\\
				& \geq \beta_{\ps([\mathbf{A,B}])}([\mathbf{A,B}]) - 2\sqrt{6\norm{\mathbf{B}}{1,2}\frac{L}{\mu}\epsb(\epsb + 2)}, 
			\end{split}
		\end{equation*}
		where the last line is obtained by Lemma~\ref{lem:gillis_B2}, since $\norm{\ps(\mathbf{B})}{1,2} = \norm{\mathbf{B}}{1,2}$ (as $\norm{\mathbf{b}_i}{2} \leq 1$ for  $i \in \lbr s\rbr$). 
	\end{proof}

		\begin{theorem}[Robustness of SNPAB when applied on linear mixings - induction step]\label{th:th316}
		Let the following hold:
		\begin{itemize}
			\item $f$ satisfies Assumption~\ref{hyp:f} with strong convexity parameter $\mu$ and a gradient Lipschitz constant $L$.
			\item $\bar{\mathbf{X}}$ follows a \emph{linear} mixing model. Precisely, $\bar{\mathbf{X}}$ is near-separable \cite{Gillis_12_SparseandUnique} with 
			\begin{equation*}
				\begin{split}
					\bar{\mathbf{X}} = \mathbf{WH + N},\ \quad \mathbf{W} &= \mathbf{[A,B]} \quad \text{ and } \quad \mathbf{A} \in \mathbb{R}^{m\times k}, \ \mathbf{B} \in \mathbb{R}^{m\times s}, \\
					\mathbf{H} &= [\mathbf{I}_r,\mathbf{H}'] \in \mathbb{R}^{r\times n}_+ \text{ where } \forall j \in \lbr n\rbr, \mathbf{h}_j\in \Delta. 
				\end{split}{}
			\end{equation*}{}
			Let further assume that the noise is bounded with   $\norm{\mathbf{n}_i}{2} \leq \epsilon$ for all $i \in \lbr t \rbr$.
			\item $\mathbf{W} = [\mathbf{A,B}]$ is such that $\alpha_{\ps(\mathbf{W})}(\mathbf{W}) > 0$ and $\beta^{\text{Lin}}_{\ps(\mathbf{W})}(\mathbf{W}) > 0$. 
			\item The error on $\bar{\mathbf{B}}\in \mathbb{R}^{m\times s}$ satisfies
			\begin{equation*}
				\norm{\mathbf{B} - \bar{\mathbf{B}}}{1,2} \leq \epsb = C \epsilon \text{ for some $C > 0$.}
			\end{equation*}
			\item $\epsilon$ is sufficiently small and satisfies  
			\begingroup\makeatletter\def\f@size{7}\check@mathfonts
			\begin{equation*}
				\small
				\begin{split}
					C\epsilon < \min 
					&\left(\frac{\alpha_{\ps(\mathbf{W})}(\mathbf{W}) \mu}{2(L + \mu)}\right. ,\frac{2L + \mu}{2\mu} + K(\mathbf{W}) + \sqrt{\left( \frac{2L+\mu}{2\mu} + K(\mathbf{W})\right)^2 + \alpha_{\ps(\mathbf{W})}(\mathbf{W})},\\
	       	    &\left.\frac{\beta^{\text{Lin}}_{\ps(\mathbf{W})}(\mathbf{W})^2\mu^{3/2}C}{144  K(\mathbf{W}) L^{3/2}},-1 + \sqrt{1 + \frac{\beta^{\text{Lin}}_{\ps(\mathbf{W})}(\mathbf{W})^2}{96K(\mathbf{W})}\frac{\mu}{L}},\sqrt{1+K(\mathbf{W})}-1,CK(\mathbf{W})\right).
				\end{split}
			\end{equation*}\endgroup
			
		\end{itemize} 
		Then the index $i$ corresponding to a column $\bar{\mathbf{x}}_i$ of $\bar{\mathbf{X}}$ maximizing the function $f(\resi{\ps(\bar{\mathbf{B}})}{.})$ satisfies
		\begin{equation}
		\mathbf{x}_i = \mathbf{Wh}_i = \mathbf{[A,B]h}_i \text{ where } h_{il} \geq 1 - \delta \text{ for some } l\in \lbr k \rbr,
		\label{eq:eq33}
		\end{equation}
		where 
			$\delta = \frac{72\epsilon K(\mathbf{W}) L^{3/2}}{\beta^{\text{Lin}}_{\ps(\bar{\mathbf{W}})}(\mathbf{W})^2\mu^{3/2}}$. 
		This implies 
		\begin{equation}
		\begin{split}
		\norm{\bar{\mathbf{x}}_i - \mathbf{w}_l}{2} = \norm{\bar{\mathbf{x}}_i - \mathbf{a}_l}{2}
		\leq \epsilon + 2 K(\mathbf{W}) \delta
		= \epsilon\left(1 + 144 \frac{K(\mathbf{W})^2}{\beta^{\text{Lin}}_{\ps(\bar{\mathbf{W}})}(\mathbf{W})^2}\frac{L^{3/2}}{\mu^{3/2}}\right).
		\end{split}
		\label{eq:eq34}
		\end{equation}
	\end{theorem}
	
		\begin{proof} 
		The result (\ref{eq:eq33}) is proved by contradiction. Let us assume that the column of $\bar{\mathbf{X}}$ maximizing $f(\resi{\ps(\mathbf{\bar{B}})}{.})$ satisfies $\bar{\mathbf{x}}_i = \mathbf{W}\mathbf{h}_i + \mathbf{n}_i$ with $h_{il} < 1 - \delta$ for $1 \leq l \leq k$. We have 
		\begin{equation*}
			\begin{split}
				f\left(\resi{\ps(\mathbf{\bar{B}})}{\bar{\mathbf{x}}_i}\right) 
				&\underset{Lemma~\ref{lem:gillis_313}}{\leq} f\left(\resi{\ps(\mathbf{\bar{B}})}{\mathbf{W}}\mathbf{h}_i + \mathbf{n}_i\right) \\
				&\underset{Lemma~\ref{lem:gillis_314}}{\leq} f\left(\resi{\ps(\mathbf{\bar{B}})}{\mathbf{W}}\mathbf{h}_i\right)  + \frac{3}{2}\epsilon K\left(\resi{\ps(\mathbf{\bar{B}})}{\mathbf{W}}\right)L
			\end{split}
		\end{equation*}
		Furthermore, due to Lemma~\ref{lem:gillis_33}, 
		\begin{equation*}
			\begin{split}
				\norm{\resi{\ps(\mathbf{\bar{B}})}{\mathbf{W}}\mathbf{h}_i }{2}
				\leq \max_i \norm{\resi{\ps(\mathbf{\bar{B}})}{\mathbf{w}_i}}{2} \underset{Lemma~\ref{lem:gillis_33}}{\leq} \sqrt{\frac{L}{\mu}}K(\mathbf{W}). 
			\end{split}	
		\end{equation*}
		Therefore, 
		\begin{equation}
		\begin{split}
		f\left(\resi{\ps(\mathbf{\bar{B}})}{\bar{\mathbf{x}}_i}\right) 
		&\leq f\left(\resi{\ps(\mathbf{\bar{B}})}{\mathbf{W}}\mathbf{h}_i\right) + \frac{3}{2}\epsilon K(\mathbf{W}) \frac{L^{3/2}}{\mu^{1/2}}\\
		& \leq \max_{\substack{\mathbf{x}\in \Delta^r \\ \mathbf{x}_l \leq 1 - \delta \\ 1 \leq l \leq k}} f\left(\resi{\ps(\mathbf{\bar{B}})}{\mathbf{W}}\mathbf{x}\right) + \frac{3}{2}\epsilon K(\mathbf{W}) \frac{L^{3/2}}{\mu^{1/2}}. 
		\end{split}
		\end{equation}
		Now, to bound $f\left(\resi{\ps(\mathbf{\bar{B}})}{\mathbf{W}}\mathbf{x}\right)$ using $f\left(\resi{\ps(\mathbf{\bar{B}})}{\mathbf{A}}\right)$, we use 
		Lemma~\ref{lem:gillis_315}. To do that, we must check that $\nu \left( \resi{\ps(\mathbf{\bar{B}})}{\mathbf{A}} \right)> 2\sqrt{\frac{L}{\mu}}K\left(\resi{\ps(\mathbf{\bar{B}})}{\mathbf{B}}\right)$, enabling to use the lemma with $P = \resi{\ps(\mathbf{\bar{B}})}{\mathbf{A}}$ and $Q = \resi{\ps(\mathbf{\bar{B}})}{\mathbf{B}}$:
		\begin{equation*}
			\begin{split}
				\nu \left( \resi{\ps(\mathbf{\bar{B}})}{\mathbf{A}} \right) 
				&\underset{\quad \quad \quad \quad \quad \quad \ Cor.~\ref{cor:Cor37}\quad \quad \quad \quad \quad}{\geq} \alpha_{\ps(\mathbf{W})}(\mathbf{W}) -\epsb(1 + \max(1, 2K(\mathbf{W}) + \epsb))\\
				&
				\underset{{\color{black} \substack{\epsb \leq \frac{\alpha_{\ps(\mathbf{W})}(\mathbf{W}) \mu}{2(L + \mu)} \\ \epsb \leq \frac{2L + \mu}{2\mu} + K(\mathbf{W}) + \sqrt{\left( \frac{2L+\mu}{2\mu} + K(\mathbf{W})\right)^2 + \alpha_{\ps(\mathbf{W})}(\mathbf{W})}}}}
				{\geq} 2 \frac{L}{\mu}\epsb\\ 
				& \underset{\quad\quad \quad \quad \quad \quad \quad \quad Lemma~\ref{lem:gillis_34} \quad \quad \quad \quad \quad}{\geq} 2 \sqrt{\frac{L}{\mu}} K\left(\resi{\mathbf{\bar{B}}}{\mathbf{B}}\right)\\
				&  \quad \quad \quad \quad \quad \quad \quad \quad \ \geq \ \quad \quad \quad \quad \quad \quad \quad \quad 2 \sqrt{\frac{L}{\mu}} K\left(\resi{\ps(\mathbf{\bar{B}})}{\mathbf{B}}\right). 
			\end{split}
		\end{equation*}
		Thus, as $\delta < 1/2$ when {\color{black}$\epsilon < \frac{\beta_{\ps(\mathbf{W})}(\mathbf{W})^2\mu^{3/2}}{144  K L^{3/2}}$}, Lemma~\ref{lem:gillis_315} applies and we obtain 
		\begin{equation*}
			\begin{split}
				\max_{\substack{\mathbf{x}\in \Delta^r \\ \mathbf{x}_l \leq 1 - \delta \\ 1 \leq l \leq k}} f\left(\resi{\ps(\mathbf{\bar{B}})}{\mathbf{W}}\mathbf{x}\right)
				\leq \max_{j} f\left(\resi{\ps(\mathbf{\bar{B}})}{\mathbf{a}_j}\right)
				- \frac{1}{2}\mu\delta(1-\delta)\omega\left(\resi{\ps(\mathbf{\bar{B}})}{\mathbf{A}}\right)^2 . 
			\end{split}
		\end{equation*}
		Therefore, 
		\begin{equation*}
			\small
			\begin{split}
				f\left(\resi{\ps(\mathbf{\bar{B}})}{\bar{\mathbf{x}}_i}\right) 
				& \leq \max_{j} f\left(\resi{\ps(\mathbf{\bar{B}})}{\mathbf{a}_j}\right) - \frac{1}{2}\mu\delta(1-\delta)\omega\left(\resi{\ps(\mathbf{\bar{B}})}{\mathbf{A}}\right)^2 + \frac{3}{2}\epsilon K(\mathbf{W}) \frac{L^{3/2}}{\mu^{1/2}}\\
				&\underset{ }{<} \max_{j} f\left(\resi{\ps(\mathbf{\bar{B}})}{\bar{\mathbf{a}}_j} - \mathbf{n}_j\right) - \frac{1}{8}\mu\delta(1-\delta)\beta_{\ps(\mathbf{W})}(\mathbf{W})^2 + \frac{3}{2}\epsilon K(\mathbf{W}) \frac{L^{3/2}}{\mu^{1/2}}, 
			\end{split}	
		\end{equation*}
		where the last line is obtained by Lemma~\ref{lem:gillis_313} and the fact that (see Lemma~\ref{lem:lemB3}) 
		\begin{equation*}
			\omega\left(\resi{\ps(\mathbf{\bar{B}})}{\mathbf{A}}\right) \geq \beta_{\ps(\mathbf{W})}(\mathbf{W}) - 2 \sqrt{\frac{6K(\mathbf{W})L\epsb(\epsb+2)}{\mu}}, 
		\end{equation*}{} 
		and thus $\omega\left(\resi{\ps(\mathbf{\bar{B}})}{\mathbf{W}}\right) > \beta_{\ps(\mathbf{W})}(\mathbf{W})/2$ when {\color{black} $\epsb < -1 + \sqrt{1 + \frac{\beta_{\ps(\mathbf{W})}(\mathbf{W})^2}{96\norm{\mathbf{B}}{1,2}}\frac{\mu}{L}}$}.\\
		Lastly, using again Lemma~\ref{lem:gillis_314} and the fact that, if {\color{black} $\epsilon < K(\mathbf{W})$}, 
		\begin{equation}
		\norm{\resi{\ps(\mathbf{\bar{B}})}{\bar{\mathbf{a}}_j}}{2} \leq \sqrt{\frac{L}{\mu}}\norm{\bar{\mathbf{a}}_j}{2} \leq \sqrt{\frac{L}{\mu}}(K(\mathbf{W}) + \epsilon) \leq 2\sqrt{\frac{L}{\mu}}K(\mathbf{W}),
		\end{equation}
		we obtain 
		\begin{equation}
		\begin{split}
		f\left(\resi{\ps(\mathbf{\bar{B}})}{\bar{\mathbf{x}}_i}\right) 
		< \max_j f\left(\resi{\ps(\mathbf{\bar{B}})}{\bar{\mathbf{a}}_j}\right) - \frac{1}{8}\mu\delta(1-\delta)\beta_{\ps(\mathbf{W})}(\mathbf{W})^2 + \frac{9}{2}\epsilon K(\mathbf{W}) \frac{L^{3/2}}{\mu^{1/2}}. 
		\end{split}
		\end{equation}
		Since
		\begin{equation}
		\small
		\begin{split}
		\frac{1}{8}\mu\delta(1-\delta)\beta_{\ps(\mathbf{W})}(\mathbf{W})^2 &\geq \frac{1}{16}\mu\beta_{\ps(\mathbf{W})}(\mathbf{W})^2\delta \\
		&=\frac{1}{16}\mu\beta_{\ps(\mathbf{W})}(\mathbf{W})^2\left(\frac{72\epsilon K(\mathbf{W}) L^{3/2}}{\beta_{\ps(\mathbf{W})}(\mathbf{W})^2\mu^{3/2}}\right)\\
		&= \frac{9}{2}\epsilon K(\mathbf{W}) \frac{L^{3/2}}{\mu^{1/2}},
		\end{split}
		\end{equation}
		we obtain  $f\left(\resi{\ps(\mathbf{\bar{B}})}{\bar{\mathbf{x}}_i}\right) < f\left(\resi{\ps(\mathbf{\bar{B}})}{\bar{\mathbf{a}}_j}\right)$, which is a contradiction since $\bar{\mathbf{x}}_i$ should maximize $f\left(\resi{\ps(\mathbf{\bar{B}})}{.}\right)$ among the $\mathbf{X}$ columns and the $\mathbf{a}_i$ are among these columns. 
		
		The proof of (\ref{eq:eq34}) follows the exact same lines as in \cite{Gillis2014}: we use (\ref{eq:eq33}), implying that 
		\begin{equation*}
			\mathbf{x}_i = (1 - \delta')\mathbf{w}_l + \sum_{k\neq l}\gamma_k \mathbf{w}_k \text{ for some $l$ and $1 - \delta' \geq 1 - \delta$},
		\end{equation*} 
		so that $\sum_{k\neq l}\gamma_k \leq \delta' \leq \delta$. 
		Therefore 
		\begin{equation*}
			\begin{split}
				\norm{\mathbf{x}_i - \mathbf{w}_l}{2} = \norm{-\delta'\mathbf{w}_l + \sum_{k\neq l}\gamma_k\mathbf{w}_k}{2} \leq 2\delta' \max_{j}\norm{\mathbf{w}_j}{2} \leq 2\delta'K(\mathbf{W}) \leq 2K(\mathbf{W})\delta,
			\end{split}
		\end{equation*}
		which leads to, when considering the noisy version of $\mathbf{X}$, 
		\begin{equation*}
			\begin{split}
				\norm{\bar{\mathbf{x}}_i - \mathbf{w}_l}{2} 
				\leq \norm{(\bar{\mathbf{x}}_i - \mathbf{x}_i) + (\mathbf{x}_i - \mathbf{w}_l)}{2}
				\leq \epsilon + 2 K(\mathbf{W}) \delta \text{ for some $1 \leq l \leq k$.}
			\end{split}
		\end{equation*}
	\end{proof}
	

	\begin{theorem}[Robustness of \snpab\ when applied on linear mixings]\label{thm:robustness_lin}
		Let 
		\begin{equation*}
			\bar{\mathbf{X}} = \mathbf{WH + N} \in \mathbb{R}^{m\times n}
		\end{equation*} 
		be a near-sepable \cite{Gillis_12_SparseandUnique} linear mixing with $\alpha_{\ps(\mathbf{W})}(\mathbf{W}) > 0$ and $\beta_{\ps(\mathbf{W})}^{\text{Lin}}(\mathbf{W}) > 0$. Let furthermore $f$ satisfy Assumption~\ref{hyp:f} and the noise be bounded: $\norm{\mathbf{n}_i}{2} \leq \epsilon$ for all $i \in \lbr t \rbr$ with 
		\begingroup\makeatletter\def\f@size{9}\check@mathfonts
		\begin{equation*}
		\begin{split}
		    C\epsilon < \min 
			&\left(\frac{C\beta^{\text{Lin}}_{\ps(\mathbf{W})}(\mathbf{W})^2\mu^{3/2}}{144  K(\mathbf{W}) L^{3/2}}\right.,\frac{\alpha_{\ps(\mathbf{W})}(\mathbf{W}) \mu}{2(L + \mu)} ,\frac{2L + \mu}{2\mu} + K(\mathbf{W}) + \sqrt{\left( \frac{2L+\mu}{2\mu} + K(\mathbf{W})\right)^2 + \alpha_{\ps(\mathbf{W})}(\mathbf{W})},
			\\
			&\left.-1 + \sqrt{1 + \frac{\beta^{\text{Lin}}_{\ps(\mathbf{W})}(\mathbf{W})^2}{96K(\mathbf{W})}\frac{\mu}{L}},\sqrt{1+K(\mathbf{W})}-1,K(\mathbf{W})\right), 
		\end{split}
		\end{equation*}\endgroup
		where $C = \left( 1 + 144\frac{K^2}{\beta^{\text{Lin}}_{\ps(\mathbf{W})}(\mathbf{W})^2}\frac{L^{3/2}}{\mu^{3/2}}\right)$ and $L$ and $\mu$ defined in Assumption~\ref{hyp:f}. Then, SNPALQ (Algorithm~\ref{alg:algo_SNPALQ}) identifies in $r$ steps all the columns of $\mathbf{W}$ up to error $C\epsilon$. Precisely, denoting by $\mathcal{K}$ the index set extracted by \snpab\ after $r$ steps, there exists a permutation $\pi$ of $\lbr r \rbr$ such that:
 		\begin{equation*}
 			\max_{1 \leq j \leq r} \norm{\bar{\mathbf{x}}_{\mathcal{K}(j)} - \mathbf{w}_{\pi(j)}}{2} \leq \epsb = C \epsilon.
 		\end{equation*}
	\end{theorem}

	\begin{proof}
		The result follows by applying Theorem~\ref{th:th316} inductively using 
		\begin{equation*}
			\small
			C = \left( 1 + 144\frac{K(\mathbf{W})^2}{\beta^{\text{Lin}}_{\ps(\mathbf{W})}(\mathbf{W})^2}\frac{L^{3/2}}{\mu^{3/2}}\right). 
		\end{equation*}
		The matrix $\mathbf{B}$ of Theorem \ref{th:th316} corresponds to the columns extracted so far by SNPA, while $\mathbf{A}$ corresponds to the columns of $\mathbf{W}$ remaining to be extracted. Note that the initialisation of the induction is done with $\mathbf{B}$ being the empty matrix. 
	\end{proof}

	\subsection{Proof of \snpab\ robustness for LQ mixtures}
	\label{sec:proof_SNPALQforLQ}
	Similarly to the above derivations, after stating a few useful lemmas, the induction step of the proof of \snpab\ robustness for LQ mixtures is given in Theorem~\ref{th:indStep} and the main result is stated in Theorem~\ref{thm:robustnessLQ}.
	
	\begin{lemma}[\cite{Gillis2014}-Lemma~15 extended]\label{lem:15extended}
		Let $\mathbf{Z} = [\mathbf{P},\mathbf{Q}]$, where $\mathbf{P} \in \mathbb{R}^{m\times k}$ and $\mathbf{Q} \in \mathbb{R}^{m\times r-k}$, and let $f$ satisfy Assumption~\ref{hyp:f}. If 
		\begin{equation}
		K(\mathbf{P}) \geq 2G \sqrt{\frac{L}{\mu}}K(\mathbf{Q})\quad \text{ with } \quad G > \sqrt{\frac{L}{\mu}} \geq 1, 
		\label{eq:hypLem15}
		\end{equation}
		then, for any $\delta \in \left[0,\frac{1}{2}\right]$, 
		\begin{equation}
		f^* = \max_{x\in \Delta} f(\mathbf{Zx}) \text{ such that } \mathbf{x}_i \leq 1-\delta \text{ for } 1 \leq i \leq k
		\label{eq:pb_lem15}
		\end{equation}
		satisfies
		\begin{equation}
		f^* \leq \max_{i}f(\mathbf{p}_i) - \frac{1}{2}\mu(1-\delta)\delta\Omega(\mathbf{P})^2
		\end{equation}
		with $\Omega(\mathbf{P}) = \min \left(\gamma(\mathbf{P}), \frac{K(\mathbf{P})}{2}\sqrt{\frac{\mu}{L}}\left[1-\frac{1}{G}\right]\right)$. 
	\end{lemma}
	
	\begin{proof} 
		First, let us provide a lower bound for $f^*$. Remember that due to the strong convexity of $f$ with parameter $\mu$, its gradient Lipschitz continuity and the fact that $f(\mathbf{0}_m) = 0$, we have that for all $\mathbf{y} \in \mathbb{R}^m$
		\begin{equation}
		\frac{\mu}{2}\norm{\mathbf{y}}{2}^2 \leq f(\mathbf{y}) \leq \frac{L}{2}\norm{\mathbf{y}}{2}.
		\label{eq:ineqF}
		\end{equation}
		Consequently, since $(1-\delta)\mathbf{p}_i$ is an admissible solution, we have that 
		\begin{equation}
		f^* \geq f((1-\delta)\mathbf{p}_i) \geq \frac{1}{2}\mu (1-\delta)^2 \norm{\mathbf{p}_i}{2}^2 \geq \frac{\mu}{8}\norm{\mathbf{p}_i}{2}^2, 
		\label{eq:l15_lowerbound}
		\end{equation}
		where the last inequality is due to the assumption $\delta \leq 1/2$.
		
		Let us now discuss upperbounds of $f$. By strong convexity of $f$, the optimal solution $x^*$ of (\ref{eq:pb_lem15}) is attained at a vertex of the feasible domain $\{\mathbf{x} \in \mathbb{R}^r_+ |  \sum_{i=1}^{r} \mathbf{x}_i \leq 1 \text{ and } \mathbf{x}_i \leq 1-\delta \text{ for } 1 \leq i \leq r\}$. Here are the different cases 
		\begin{itemize}
			\item[a)] $\mathbf{x}^* = \mathbf{0}_r$;
			\item[b)] $\mathbf{x}^* = \mathbf{e}_i$ for $k+1 \leq i \leq r$;
			\item[c)] $\mathbf{x}^* = (1-\delta)\mathbf{e}_j$ for $1\leq j \leq k$;
			\item[d)] $\mathbf{x}^* = \delta \mathbf{e}_i + (1-\delta)\mathbf{e}_j$ for $1 \leq i,j \leq k$;
			\item[e)] $\mathbf{x}^* = \delta \mathbf{e}_i + (1-\delta)\mathbf{e}_j$ for $k+1 \leq i \leq r$ and $1 \leq j \leq k$\\
		\end{itemize}
		Let us analyze them separately. 
		\begin{itemize}
			\item[a)] This first case is clearly impossible, as $f(\mathbf{0}_m) = 0$ and $f(\mathbf{y}) > 0$ for all $y \neq 0$; see Eq.~(\ref{eq:ineqF}). 
			
			\item[b)] $\mathbf{Zx}^* = \mathbf{q}_i$ for some $i$. Using Eq.~(\ref{eq:ineqF}), we obtain
			\begin{equation}
			f^* = f(\mathbf{q}_i) \leq \frac{L}{2}K(\mathbf{Q}) \underset{\text{Hyp. }(\ref{eq:hypLem15})}{\leq} \frac{\mu}{8G^2}K(\mathbf{P})^2 \underset{Eq.~\eqref{eq:l15_lowerbound}}{<} f^*,
			\end{equation}
			which is a contraction. 
			
			\item[c)] $\mathbf{Zx}^* = (1-\delta)\mathbf{p}_i$ for some $i$. Let us then distinguish two subcases: 
			\begin{itemize}
				\item[\textbullet] If $\norm{\mathbf{p}_j}{2}^2 \leq \frac{\mu}{4L}K(\mathbf{P})^2$, then:
				\begin{equation*}
					\begin{split}
						f^* = f((1-\delta)\mathbf{p}_j) \underset{\substack{f(\mathbf{0}_m) = 0\\ f\text{ strongly convex}}}{<} (1-\delta)f(\mathbf{p}_j) &\underset{Eq.~(\ref{eq:ineqF})}{\leq} (1-\delta)\frac{L}{2}\norm{\mathbf{p}_j}{2}^2 \\
						&\leq (1-\delta)\frac{\mu}{8}K(\mathbf{P})^2 \\
						&< \frac{\mu}{8}K(\mathbf{P})^2 \leq f^* , 
					\end{split}	
				\end{equation*}
				which is a contradiction. 
				
				\item[\textbullet] If $\norm{\mathbf{p}_j}{2}^2 > \frac{\mu}{4L}K(\mathbf{P})^2$, then, by strong convexity of $f$, 
				\begin{equation*}
					f^* \leq (1-\delta)f(\mathbf{p}_j) - \frac{1}{2}\mu\delta(1-\delta)\norm{\mathbf{p}_j}{2}^2 = f(\mathbf{p}_j) - \delta f(\mathbf{p}_j) - \frac{1}{2}\mu\delta(1-\delta)\norm{\mathbf{p}_j}{2}^2.
				\end{equation*}
				Since $f(\mathbf{p}_j) \underset{Eq.~(\ref{eq:ineqF})}{\geq} \frac{\mu}{2}\norm{\mathbf{p}_j}{2}^2 \geq \frac{1}{2}\mu (1-\delta)\norm{\mathbf{p}_j}{2}^2$, 
				\begin{equation*}
					\begin{split}
						f^* &< f(\mathbf{p}_j) - \mu \delta(1-\delta)\norm{\mathbf{p}_j}{2}^2 \\
						& \leq f(\mathbf{p}_j) - \frac{1}{2}\mu \delta(1-\delta)\left[\sqrt{\frac{\mu}{2L}}K(\mathbf{P})\right]^2\\
						& \leq f(\mathbf{p}_j) - \frac{1}{2}\mu\delta(1-\delta)\left[\frac{K(\mathbf{P})}{2}\sqrt{\frac{\mu}{L}}\left(1-\frac{1}{G}\right)\right]^2,
					\end{split}
				\end{equation*}
			\end{itemize}
			which satisfies the bound of the theorem.
			
			\item[d)] $\mathbf{Zx}^* = \delta \mathbf{p}_i + (1-\delta)\mathbf{p}_j$ for some $i \neq j$. Then, by strong convexity of $f$,  
			\begin{equation*}
				f^* \leq \delta f(\mathbf{p}_i) + (1-\delta)f(\mathbf{p}_j) - \frac{1}{2}\mu \delta (1-\delta) \norm{\mathbf{p}_i - \mathbf{p}_j}{2}^2 \leq K(\mathbf{P}) - \frac{1}{2}\mu \delta (1-\delta) \gamma(\mathbf{P})^2
			\end{equation*}
			
			\item[e)] $\mathbf{Yx}^* = \delta \mathbf{q}_i + (1-\delta)\mathbf{p}_j$ for some $i,j$. First (similarly to case c), let us distinguish two subcases: 
			\begin{itemize}
				\item[\textbullet] Let us assume $\norm{\mathbf{p}_j}{2}^2 \leq \frac{\mu}{4L}K(\mathbf{P})^2$. Then,
				\begin{equation*}
					\begin{split}
						f^* = f(\delta \mathbf{q}_i + (1-\delta)\mathbf{p}_j) &\underset{\text{Strong convexity}}{<} \delta f(\mathbf{q}_i) + (1-\delta)f(\mathbf{p}_j) \\
						&\underset{Eq.~(\ref{eq:ineqF})}{\leq} \delta\frac{L}{2}\norm{\mathbf{q}_i}{2}^2 + (1-\delta)\frac{L}{2}\norm{\mathbf{p}_j}{2}^2\\
						&\underset{Hyp.~(\ref{eq:hypLem15})}{<} \delta \frac{\mu}{8}K(\mathbf{P})^2 + (1-\delta)\frac{\mu}{8}K(\mathbf{P})^2\\
						&= \frac{\mu}{8}K(\mathbf{P})^2 \underset{\ref{eq:l15_lowerbound}}{\leq} f^*, 
					\end{split}
				\end{equation*}
				 a contradiction. 
				
				\item[\textbullet] If $\norm{\mathbf{p}_j}{2}^2 > \frac{\mu}{4L}K(\mathbf{P})^2$, then, by strong convexity of $f$,  
				\begin{equation*}
					f^* \leq \delta f(\mathbf{q}_i) + (1-\delta)f(\mathbf{p}_j) - \frac{1}{2}\mu \delta (1-\delta)\norm{\mathbf{q}_i - \mathbf{p}_j}{2}^2
				\end{equation*}
				Using the triangle inequality, we obtain 
				\begin{equation*}
					\norm{\mathbf{q}_i - \mathbf{p}_j}{2} \geq \norm{\mathbf{p}_j}{2} - \norm{\mathbf{q}_i}{2}, 
				\end{equation*}
				since $\norm{\mathbf{q}_i}{2} \leq K(\mathbf{Q}) 
				\underset{Hyp.~(\ref{eq:hypLem15})}{<} \frac{1}{g}\sqrt{\frac{\mu}{4L}}K(\mathbf{P}) \leq \norm{\mathbf{p}_j}{2}$. Thus,
				\begin{equation*}
					\norm{\mathbf{q}_i - \mathbf{p}_j}{2} \geq \frac{1}{2}\sqrt{\frac{\mu}{L}}K(\mathbf{P}) - \frac{1}{2G}\sqrt{\frac{\mu}{L}}K(\mathbf{P}) = \frac{1}{2}\sqrt{\frac{\mu}{L}}K(\mathbf{P})\left(1 - \frac{1}{G}\right). 
				\end{equation*}
				Furthermore,
				\begin{equation*}
					\begin{split}
						f(\mathbf{q}_i) \leq \frac{L}{2}\norm{\mathbf{q}_i}{2}^2 \leq \frac{L}{2}K(\mathbf{Q})^2 
						&\leq \frac{\mu}{8G^2}K(\mathbf{P})^2 \\
						&< \frac{L}{2G^2}\norm{\mathbf{p}_j}{2}^2\\
						&\leq \frac{L}{\mu G^2}f(\mathbf{p}_j).
					\end{split}
				\end{equation*}
				Putting the above expression altogether, we have 
				\begin{equation*}
					\begin{split}
						f^* &< f(\mathbf{p}_j) + \delta \frac{L}{\mu G^2}f(\mathbf{p}_j) - \delta f(\mathbf{p}_j) - \frac{1}{2}\mu \delta (1-\delta)\left[\frac{1}{2}\sqrt{\frac{\mu}{L}}K(\mathbf{P})\left(1 - \frac{1}{G}\right)\right]^2\\
						&\leq f(\mathbf{p}_j) + \left(\frac{L}{\mu G^2} - 1\right)\delta f(\mathbf{p}_j) - \frac{1}{2}\mu \delta (1-\delta)\left[\frac{1}{2}\sqrt{\frac{\mu}{L}}K(\mathbf{P})\left(1 - \frac{1}{G}\right)\right]^2\\
						& \leq f(\mathbf{p}_j) - \frac{1}{2}\mu \delta (1-\delta)\left[\frac{1}{2}\sqrt{\frac{\mu}{L}}K(\mathbf{P})\left(1 - \frac{1}{G}\right)\right]^2
					\end{split}	
				\end{equation*}
				where the last line, which satisfies the bound of the theorem, requires $L < \mu G^2$.
			\end{itemize}
		\end{itemize}
	\end{proof}
	
	\begin{lemma}\label{lem:313}
		Let $\mathbf{A} \in \mathbb{R}^{m\times k}$ and $\bar{\mathbf{B}} \in \mathbb{R}^{m\times s}$ be such that $\mathbf{B} - \bar{\mathbf{B}} = \mathbf{N}$ and $\norm{\mathbf{N}}{2}< \epsilon$, and let $f$ satisfy Assumption~\ref{hyp:f}.  
		Then,
		\begin{equation*}
			\Omega(\mathcal{R}^f_{\bar{\mathbf{B}}}(\mathbf{A}))^2 \geq \Omega(\mathcal{R}^f_{{\mathbf{B}}}(\mathbf{A}))^2 - 4 \epsilon \ckc{(K(\mathbf{A}) + K(\mathbf{B}))}.
		\end{equation*}
	\end{lemma}
	
	
	\begin{proof} 
		Let us look at the two terms of $\Omega(\mathcal{R}^f_{\bar{\mathbf{B}}}(\mathbf{A}))$: 
		\begin{itemize}
			\item[\textbullet] 
			Denoting $\mathbf{z}_{\mathbf{a}_j} = \argmin_{\mathbf{x} \in \Delta}\norm{\mathbf{a}_j - \bar{\mathbf{B}}\mathbf{x}}{2}$, we have, for $j \in \lbr k\rbr$, 
			\begin{equation*}
				\norm{\mathcal{R}^f_{\bar{\mathbf{B}}}(\mathbf{a}_j)}{2}
				= \norm{\mathbf{a}_j - \bar{\mathbf{B}}\mathbf{z}_{\mathbf{a}_j}}{2} 
				= \norm{\mathbf{a}_j - (\mathbf{B + N})\mathbf{z}_{\mathbf{a}_j}}{2} 
				\geq \left| \norm{\mathbf{a}_j - \mathbf{B}\mathbf{z}_{\mathbf{a}_j}}{2} - \norm{\mathbf{N}\mathbf{z}_{\mathbf{a}_j}}{2}\right|
			\end{equation*}
			Thus, for $j \in \lbr k\rbr$,
			\begin{equation*}
				\begin{split}
					\norm{\mathcal{R}^f_{\bar{\mathbf{B}}}(\mathbf{a}_j)}{2}^2
					&\geq \left( \norm{\mathbf{a}_j - \mathbf{B}\mathbf{z}_{\mathbf{a}_j}}{2} - \norm{\mathbf{N}\mathbf{z}_{\mathbf{a}_j}}{2}\right)^2 \\
					&\geq \norm{\mathbf{a}_j - \mathbf{B}\mathbf{z}_{\mathbf{a}_j}}{2}^2 - 2 \norm{\mathbf{a}_j - \mathbf{B}\mathbf{z}_{\mathbf{a}_j}}{2}\norm{\mathbf{N}\mathbf{z}_{\mathbf{a}_k}}{2} + \norm{\mathbf{N}\mathbf{z}_{\mathbf{a}_j}}{2}^2\\
					&\geq \norm{\mathcal{R}^f_{{\mathbf{B}}}(\mathbf{a}_j)}{2}^2 - 2(K(\mathbf{A}) + K(\mathbf{B}))\epsilon , 
				\end{split}
			\end{equation*}
			where the last line is obtained since $\mathbf{z}_{\mathbf{a}_j} \in \Delta$. 
			This yields 
	           \[ 	K({\mathcal{R}^f_{\bar{\mathbf{B}}}(\mathbf{A})})^2
					\geq K({\mathcal{R}^f_{{\mathbf{B}}}(\mathbf{A})})^2 - (K(\mathbf{A}) + K(\mathbf{B}))\epsilon, 
					\] 
			and, as $\frac{\mu}{L}\left[1-\frac{1}{G}\right]^2 \leq 1$,
			\begin{equation*}
				\begin{split}
					\frac{K({\mathcal{R}^f_{\bar{\mathbf{B}}}(\mathbf{A})})^2}{4}\frac{\mu}{L}\left[1-\frac{1}{G}\right]^2
					\geq \frac{K({\mathcal{R}^f_{{\mathbf{B}}}(\mathbf{A})})^2}{4}\frac{\mu}{L}\left[1-\frac{1}{G}\right]^2 - \frac{\epsilon (K(\mathbf{A}) + K(\mathbf{B}))}{2}.
				\end{split}
			\end{equation*}
			
			\item[\textbullet]  Denoting 
			\begin{equation*}
				\mathbf{z}_{\mathbf{a}_i} = \argmin_{\mathbf{x} \in \Delta}\norm{\mathbf{a}_i - \bar{\mathbf{B}}\mathbf{x}}{2} \quad \text{ and } \quad \mathbf{z}_{\mathbf{a}_j} = \argmin_{\mathbf{x} \in \Delta}\norm{\mathbf{a}_j - \bar{\mathbf{B}}\mathbf{x}}{2},
			\end{equation*}{}
			we have, for $i \neq j, i,j \in \lbr k\rbr$, 
			\begin{equation*}
				\norm{\mathcal{R}^f_{\bar{\mathbf{B}}}(\mathbf{a}_i) - \mathcal{R}^f_{\bar{\mathbf{B}}}(\mathbf{a}_j)}{2}
				\geq \left| \norm{\mathbf{a}_i - \mathbf{B}\mathbf{z}_{\mathbf{a}_i} - (\mathbf{a}_j - \mathbf{B}\mathbf{z}_{\mathbf{a}_j})}{2} - \norm{\mathbf{N}(\mathbf{z}_{\mathbf{a}_i} - \mathbf{z}_{\mathbf{a}_j})}{2}\right|.
			\end{equation*}
			This yields 
			\begin{equation*}
				\begin{split}
					\norm{\mathcal{R}^f_{\bar{\mathbf{B}}}(\mathbf{a}_i) - \mathcal{R}^f_{\bar{\mathbf{B}}}(\mathbf{a}_j)}{2}^2
					&\geq  \gamma(\mathcal{R}^f_{\mathbf{B}}(\mathbf{A}))^2 - 4\epsilon (K(\mathbf{A}) + K(\mathbf{B})), 
				\end{split}
			\end{equation*}
		\end{itemize}
		which gives the result. 
	\end{proof}
	
	\begin{lemma}\label{lem:314}
		Let $\mathbf{A} \in \mathbb{R}^{m\times k}$ such that $K(\mathbf{A}) \leq 1$ and $\bar{\mathbf{B}} \in \mathbb{R}^{m\times s}$ be such that $\mathbf{B} - \bar{\mathbf{B}} = \mathbf{N}$ and $\norm{\mathbf{N}}{2}< \epsilon < 1$, and let $f$ satisfy Assumption~\ref{hyp:f}.  
		Then,
		\begin{equation*}
			\Omega(\mathcal{R}^f_{\ps(\bar{\mathbf{B}})}(\mathbf{A}))^2 \geq \Omega(\mathcal{R}^f_{\ps({\mathbf{B}})}(\mathbf{A}))^2 - 4 \ckc{(K(\mathbf{A}) + K(\ps(\mathbf{B})))\max(\epsilon, 2\epsilon K(\mathbf{B})+ \epsilon^2)}.
		\end{equation*}
	\end{lemma}
	\begin{proof}
		The result follows directly from the previous Lemma~\ref{lem:313}, by noting that 
		\begin{equation*}
			\norm{\ps(\bar{\mathbf{B}}) - \ps({\mathbf{B}})}{2} \leq \max(\epsilon, 2\epsilon K(\mathbf{B})+ \epsilon^2).
		\end{equation*} 
	\end{proof}

	\begin{theorem}[Robustness of SNPAB when applied on linear-quadratic mixings - induction step]
		\label{th:indStep}
		Let the following hold:
		\begin{itemize}
			\item ${\mathbf{X}}$ follows a LQ mixing model. Precisely, $\bar{\mathbf{X}}$ satisfies Definition~\ref{ass:rLSsep} with 
			\begin{equation*}
				\begin{split}
					\bar{\mathbf{X}} = \mathbf{\ps(W)H + N}, \quad \mathbf{W} &= \mathbf{[A,B]} \quad \text{ and } \quad \mathbf{A} \in \mathbb{R}^{m\times k}, \ \mathbf{B} \in \mathbb{R}^{m\times s}, \\
					\mathbf{H} &\in \mathbb{R}^{\frac{r(r+1)}{2}\times n}_+ \text{ with } \forall j \in \lbr n\rbr, \mathbf{h}_j\in \Delta . 
				\end{split}{}
			\end{equation*}{}
			Let us further assume the noise to be bounded as  $\norm{\mathbf{N}}{1,2} \leq \epsilon$, and denote by $\mathbf{X} = \mathbf{\ps(W)H}$ the noiseless version of $\bar{\mathbf{X}}$. 
			Note that, with these notations, 
			\begin{equation*}
				\ps(\mathbf{W}) = \left[(\mathbf{a}_i)_{i\in \lbr k\rbr},(\mathbf{b}_i)_{i\in \lbr s\rbr},(\mathbf{a}_i\odot \mathbf{a}_j)_{\substack{i\leq j \\ i\in \lbr k\rbr \\ j\in \lbr k\rbr}},(\mathbf{b}_i\odot \mathbf{b}_j)_{\substack{i\leq j \\ i\in \lbr s\rbr \\ j\in \lbr s\rbr}},(\mathbf{a}_i\odot \mathbf{b}_j)_{\substack{i\in \lbr k\rbr \\j \in \lbr s\rbr}}\right]. 
			\end{equation*}{}

			\item $\bar{\mathbf{B}}\in \mathbb{R}^{m\times s}$ satisfies
			$\norm{\mathbf{B} - \bar{\mathbf{B}}}{1,2} \leq C\epsilon$ for some $C > 0$. 
			
			\item $\mathbf{W} = [\mathbf{A,B}]$ is such that $\alpha_{\ps(\mathbf{B})}(\mathbf{A}) > 0$, $\gamma(\mathcal{R}^f_{\ps({\mathbf{B}})}(\mathbf{A})) > 0$. We further assume w.l.o.g.\ that $K(\mathbf{W}) \leq 1$. 
			
			\item For some $G > 1$, the matrix $\mathbf{W}$ and the considered $\mathbf{A}$ and $\mathbf{B}$ satisfy:
			\begingroup\makeatletter\def\f@size{9}\check@mathfonts
			\begin{equation}
			K\left(\mathcal{R}^f_{\ps{(\bar{\mathbf{B}})}}(\mathbf{A})\right) \geq 2G K\left(\mathcal{R}^f_{\ps{(\bar{\mathbf{B}})}}\left((\mathbf{b}_i)_{i\in \lbr s\rbr},(\mathbf{a}_i\odot \mathbf{a}_j)_{\substack{i\leq j \\ i\in \lbr k\rbr \\ j\in \lbr k\rbr}},(\mathbf{b}_i\odot \mathbf{b}_j)_{\substack{i\leq j \\ i\in \lbr s\rbr \\ j\in \lbr s\rbr}},(\mathbf{a}_i\odot \mathbf{b}_j)_{\substack{i\in \lbr k\rbr \\j \in \lbr s\rbr}}\right)\right). 
			\label{eq:hypLQ}
			\end{equation}\endgroup
			\item $f$ satisfies Assumption~\ref{hyp:f} with strong convexity parameter $\mu$ and gradient Lipschitz constant $L$ such that $L < \mu G^2$.
			\item $\epsilon$ is sufficiently small so that 
			
			\begingroup\makeatletter\def\f@size{9}\check@mathfonts
			\begin{equation}
			\ckc{
			\begin{split}
            &C\epsilon  <\min \left(\frac{-40L^{3/2} - 16 \mu^{3/2}CK(\mathbf{W}) + \sqrt{\left(40L^{3/2} + 16 \mu^{3/2}CK(\mathbf{W})\right)^2 + \frac{32\mu^3 C^2\Omega(\mathcal{R}^f_{\ps({\mathbf{B}})}(\mathbf{A}))^2}{K(\ps(\mathbf{W}))}}}{16 \mu^{3/2}C},\right.\\
				& \frac{\mu^{3/2}C\Omega(\mathcal{R}^f_{\ps({\mathbf{B}})}(\mathbf{A}))^2}{K(\ps(\mathbf{W}))(40L^{3/2} + 8 C \mu^{3/2})}, \frac{\min(M,\Omega(\mathcal{R}^f_{\ps({\mathbf{B}})}(\mathbf{A})))^2}{8K(\ps(\mathbf{W}))},\left.  \sqrt{1 + \frac{\min(M,\Omega(\mathcal{R}^f_{\ps({\mathbf{B}})}(\mathbf{A})))^2}{8K(\ps(\mathbf{W}))}} -K(\mathbf{W})  			\right) 
			\end{split}
			}
			\end{equation}\endgroup
			
			with $M$ a constant and
			\begin{equation*}
				\Omega(\mathcal{R}^f_{\ps({\mathbf{B}})}(\mathbf{A})) = 
				\min \left(\gamma(\mathcal{R}^f_{\ps({\mathbf{B}})}(\mathbf{A})), \frac{K(\mathcal{R}^f_{\ps({\mathbf{B}})}(\mathbf{A}))}{2}\sqrt{\frac{\mu}{L}}\left[1-\frac{1}{G}\right]\right).
			\end{equation*}
		\end{itemize}
		
		Then the index $i$ corresponding to a column $\bar{\mathbf{x}}_i$ of $\bar{\mathbf{X}}$ maximizing the function $f(\resi{\ps(\bar{\mathbf{B}})}{.})$ satisfies
		\begin{equation}
		\begin{split}
		\mathbf{x}_i &= \ps(\mathbf{W})\mathbf{h}_i \quad \text{ with } \quad h_{il} \geq 1 - \delta \quad \text{ for some } l \in \lbr k \rbr,
		\end{split}{}
		\label{eq:eq39}
		\end{equation} 
		and 
		\begin{equation}
		\small
		\delta = \frac{20 \epsilon
			K(\ps(\mathbf{W}))}{\Omega(\mathcal{R}^f_{\ps({\mathbf{B}})}(\mathbf{A}))^2 - \ckc{8K(\ps(\mathbf{W}))C\epsilon \max(1, 2 K(\mathbf{W})+ C\epsilon)}}\frac{L^{3/2}}{\mu^{3/2}} \leq \frac{1}{2}.
		\label{eq:l15_delta}
		\end{equation}
		This implies 
		\begin{equation}
		\norm{\bar{\mathbf{x}}_i - \mathbf{w}_l}{2} = \norm{\bar{\mathbf{x}}_i - \mathbf{a}_l}{2}
		\leq \epsilon\left[1 + \frac{40K(\ps(\mathbf{W}))^2}{\Omega(\mathcal{R}^f_{\ps({\mathbf{B}})}(\mathbf{A}))^2 - M^2}\frac{L^{3/2}}{\mu^{3/2}}\right].
		\label{eq:eq314}
		\end{equation}
	\end{theorem}
	
	\begin{proof} 
		The robustness is proved by contradiction. Let us assume that the column of $\bar{\mathbf{X}}$ maximizing $f(\resi{\ps(\mathbf{\bar{B}})}{.})$ satisfies $\bar{\mathbf{x}}_i = \ps(\mathbf{W})\mathbf{h}_i + \mathbf{n}_i$ with $h_{il} < 1 - \delta$ for $1 \leq l \leq k$.
		We have 
		\begin{equation*}
			\begin{split}
				f\left(\resi{\ps(\mathbf{\bar{B}})}{\bar{\mathbf{x}}_i}\right) 
				&\underset{Lemma~\ref{lem:gillis_313}}{\leq} f\left(\resi{\ps(\mathbf{\bar{B}})}{\ps(\mathbf{W})}\mathbf{h}_i + \mathbf{n}_i\right) \\
				&\underset{Lemma~\ref{lem:gillis_314}}{\leq} f\left(\resi{\ps(\mathbf{\bar{B}})}{\ps(\mathbf{W})}\mathbf{h}_i\right)  + \frac{3}{2}\epsilon K\left(\resi{\ps(\mathbf{\bar{B}})}{\ps(\mathbf{W})}\mathbf{h}_i\right)L . 
			\end{split}
		\end{equation*} 
		Using Lemma~\ref{lem:gillis_33}, 
		\begin{equation*}
			\norm{\resi{\ps(\mathbf{\bar{B}})}{\ps(\mathbf{W})}\mathbf{h}_i }{2}
			\leq \max_i \norm{\resi{\ps(\mathbf{\bar{B}})}{\ps(\mathbf{W})_i}}{2} \underset{Lemma~\ref{lem:gillis_33}}{\leq} \sqrt{\frac{L}{\mu}}K(\ps(\mathbf{W})),
		\end{equation*}
		we obtain 
		\begin{equation}
		\small
		\begin{split}
		f\left(\resi{\ps(\mathbf{\bar{B}})}{\bar{\mathbf{x}}_i}\right) 
		&\leq f\left(\resi{\ps(\mathbf{\bar{B}})}{\ps(\mathbf{W})}\mathbf{h}_i\right) + \frac{3}{2}\epsilon K(\ps(\mathbf{W})) \frac{L^{3/2}}{\mu^{1/2}}\\
		& \leq \max_{\substack{\mathbf{x}\in \Delta^r \\ \mathbf{x}_l \leq 1 - \delta \\ 1 \leq l \leq k}} f\left(\resi{\ps(\mathbf{\bar{B}})}{\ps(\mathbf{W})}\mathbf{x}\right) + \frac{3}{2}\epsilon K(\ps(\mathbf{W})) \frac{L^{3/2}}{\mu^{1/2}}\\
		&\underset{\substack{Lem~(\ref{lem:15extended})\\ Eq.~(\ref{eq:hypLQ})}}{\leq} \max_{j} f\left(\resi{\ps(\mathbf{\bar{B}})}{\mathbf{a}_j}\right) - \frac{1}{2}\mu\delta(1-\delta)\Omega\left(\resi{\ps(\mathbf{\bar{B}})}{\mathbf{A}}\right)^2 + \frac{3}{2}\epsilon K(\ps(\mathbf{W})) \frac{L^{3/2}}{\mu^{1/2}}\\
		&\begin{split}
		    \underset{Lemma~\ref{lem:gillis_313}}{\leq} \max_{j} &f\left(\resi{\ps(\mathbf{\bar{B}})}{\bar{\mathbf{a}}_j} - \mathbf{n}_j\right) - \frac{1}{2}\mu\delta(1-\delta)\Omega\left(\resi{\ps(\mathbf{\bar{B}})}{\mathbf{A}}\right)^2 \\
		    &+ \frac{3}{2}\epsilon K(\ps(\mathbf{W})) \frac{L^{3/2}}{\mu^{1/2}}
		\end{split}
		\\
		&\underset{Lemma~\ref{lem:gillis_314}}{\leq} \max_{j} f\left(\resi{\ps(\mathbf{\bar{B}})}{\bar{\mathbf{a}}_j}\right) - \frac{1}{2}\mu\delta(1-\delta)\Omega\left(\resi{\ps(\mathbf{\bar{B}})}{\mathbf{A}}\right)^2 + \frac{9}{2}\epsilon K(\ps(\mathbf{W})) \frac{L^{3/2}}{\mu^{1/2}}\\
		&\underset{Lem~(\ref{lem:314})}{\leq} 
		\begin{split}
		\max_{j} &f\left(\resi{\ps(\mathbf{\bar{B}})}{\bar{\mathbf{a}}_j}\right) 
		- \frac{1}{2}\mu\delta(1-\delta)\left[\Omega\left(\resi{\ps(\mathbf{B})}{\mathbf{A}}\right)^2 \right.\\
		&\left. - \ckc{4(K(\mathbf{A}) + K(\ps(\mathbf{B})))}\ckc{C\epsilon\max(1, 2 K(\mathbf{B})+ C\epsilon)}\right] + \frac{9}{2}\epsilon K(\ps(\mathbf{W})) \frac{L^{3/2}}{\mu^{1/2}}. 
		\end{split}\\
		&\underset{Lem~(\ref{lem:314})}{\leq} 
		\begin{split}
		\max_{j} f\left(\resi{\ps(\mathbf{\bar{B}})}{\bar{\mathbf{a}}_j}\right) 
		&- \frac{1}{2}\mu\delta(1-\delta)\left[\Omega\left(\resi{\ps(\mathbf{B})}{\mathbf{A}}\right)^2 - \ckc{8K(\ps(\mathbf{W}))C\epsilon}\right.\\
		&\left. \ckc{\times \max(1, 2 K(\mathbf{W})+ C\epsilon)}\right] + \frac{9}{2}\epsilon K(\ps(\mathbf{W})) \frac{L^{3/2}}{\mu^{1/2}}. 
		\end{split}
		\end{split}
		\end{equation}
		The fifth inequality follows from Lemma~\ref{lem:gillis_33} since
		\begin{equation*}
			\norm{\mathcal{R}^f_{\ps(\bar{\mathbf{B}})}(\bar{\mathbf{a}}_j)}{2} \leq \sqrt{\frac{L}{\mu}}\norm{\bar{\mathbf{a}}_j}{2} \leq \sqrt{\frac{L}{\mu}}(K(\ps(\mathbf{W}))+\epsilon) \leq 2\sqrt{\frac{L}{\mu}}K(\ps(\mathbf{W})),
		\end{equation*}
		\ckc{if $\epsilon \leq K(\ps(\mathbf{W}))$}, which implies that 
		\begin{equation*}
			\small
			\begin{split}
			    f^* &< \max_{j} f\left(\resi{\ps(\mathbf{\bar{B}})}{\bar{\mathbf{a}}_j}\right)\\
			    & - \frac{1}{2}\mu\delta(1-\delta)\left[\Omega\left(\resi{\ps(\mathbf{{B}})}{\mathbf{A}}\right)^2 - \ckc{8K(\ps(\mathbf{W}))C\epsilon \max(1, 2 K(\mathbf{W})+ C\epsilon)}\right] + \frac{10}{2}\epsilon K(\ps(\mathbf{W})) \frac{L^{3/2}}{\mu^{1/2}}. 
			\end{split}
		\end{equation*}
		Then, replacing $\delta$ by its expression (\ref{eq:l15_delta}), we obtain
		\begin{equation*}
			\begin{split}
				&\frac{1}{2} \mu \delta (1-\delta) \left[\Omega(\mathcal{R}^f_{\ps({\mathbf{B}})}(\mathbf{A}))^2 - \ckc{8K(\ps(\mathbf{W}))C\epsilon \max(1, 2 K(\mathbf{W})+ C\epsilon)}\right]\\
				&\geq \frac{1}{4}  \mu \delta \left[\Omega(\mathcal{R}^f_{\ps({\mathbf{B}})}(\mathbf{A}))^2 - \ckc{8K(\ps(\mathbf{W}))C\epsilon \max(1, 2 K(\mathbf{W})+ C\epsilon)}\right]\\
				&
				\begin{split}
				    = &\frac{1}{4} \mu \left(20 \frac{\epsilon K(\ps(\mathbf{W}))}{\Omega(\mathcal{R}^f_{\ps({\mathbf{B}})}(\mathbf{A}))^2 - \ckc{8K(\ps(\mathbf{W}))C\epsilon \max(1, 2 K(\mathbf{W})+ C\epsilon)}}\frac{L^{3/2}}{\mu^{3/2}}\right)\\ &\times \left[\Omega(\mathcal{R}^f_{\ps({\mathbf{B}})}(\mathbf{A}))^2 - \ckc{8K(\ps(\mathbf{W}))C\epsilon \max(1, 2 K(\mathbf{W})+ C\epsilon)}\right]
				\end{split}
				\\
				& = \frac{10}{2}\epsilon K(\ps(\mathbf{W})) \frac{L^{3/2}}{\mu^{3/2}}.
			\end{split}
		\end{equation*}
		Therefore we finally obtain a contradiction since we should have 
		\begin{equation*}
			f\left(\resi{\ps(\mathbf{\bar{B}})}{\bar{\mathbf{x}}_i}\right) < \max_{j} f\left(\resi{\ps(\mathbf{\bar{B}})}{\bar{\mathbf{a}}_j}\right),
		\end{equation*}
		which is impossible since $\bar{\mathbf{x}_i}$ should maximize $f\left( \mathcal{R}^f_{\ps(\bar{\mathbf{B}})}(.)\right)$ among the columns of $\bar{\mathbf{X}}$ and the $\bar{\mathbf{a}_j}$ are among these columns.
		
		Note that in the previous reasoning, we have assumed $\delta$ to be in $\left[0,\frac{1}{2}\right]$, which is satisfied if:   
		\ckc{
		\begin{equation*}
			\small
			\begin{split}
				C\epsilon < &\min \left(\frac{-40L^{3/2} - 16 \mu^{3/2}CK(\mathbf{W}) + \sqrt{\left(40L^{3/2} + 16 \mu^{3/2}CK(\mathbf{W})\right)^2 + \frac{32\mu^3 C^2\Omega(\mathcal{R}^f_{\ps({\mathbf{B}})}(\mathbf{A}))^2}{K(\ps(\mathbf{W}))}}}{16 \mu^{3/2}C},\right.\\
				& \qquad \frac{C\mu^{3/2}\Omega(\mathcal{R}^f_{\ps({\mathbf{B}})}(\mathbf{A}))^2}{K(\ps(\mathbf{W}))(40L^{3/2} + 8 C \mu^{3/2})}, \frac{\Omega(\mathcal{R}^f_{\ps({\mathbf{B}})}(\mathbf{A}))^2}{8K(\ps(\mathbf{W}))},\left. -K(\mathbf{W}) + \sqrt{K(\mathbf{W})^2 + \frac{\Omega(\mathcal{R}^f_{\ps({\mathbf{B}})}(\mathbf{A}))^2}{8K(\ps(\mathbf{W}))}}\right)
			\end{split}
		\end{equation*}
		}
		The proof of (\ref{eq:eq314}) follows from result (\ref{eq:eq39}). We have 
		\begin{equation*}
			\mathbf{x}_i = (1 - \delta')\mathbf{w}_l + \sum_{k\neq l}\gamma_k \mathbf{w}_k + \sum_{i,j}g_{ij}\mathbf{w}_i \odot \mathbf{w}_j\text{ for some $l$ and $1 - \delta' \geq 1 - \delta$}
		\end{equation*}
		so that $\sum_{k\neq l}\gamma_k + \sum_{i,j} g_{ij} \leq \delta' \leq \delta$. Hence 
		\begin{equation*}
			\begin{split}
				\norm{\mathbf{x}_i - \mathbf{w}_l}{2} = \norm{-\delta'\mathbf{w}_l + \sum_{k\neq l}\gamma_k\mathbf{w}_k + \sum_{i,j} g_{ij}\mathbf{w}_i \odot \mathbf{w}_j }{2} &\leq 2\delta' \max_{j}\norm{\ps(\mathbf{W})_j}{2}\\
				&= 2\delta'K(\ps(\mathbf{W}))\\
				&\leq 2\delta K(\ps(\mathbf{W})),
			\end{split}
		\end{equation*}
		which gives, when considering the noisy version of $\mathbf{X}$, 
		\begin{equation*}
			\begin{split}
				\norm{\bar{\mathbf{x}}_i - \mathbf{w}_l}{2} 
				\leq \norm{\bar{(\mathbf{x}}_i - \mathbf{x}_i) + (\mathbf{x}_i - \mathbf{w}_l)}{2}
				\leq \epsilon + 2 K(\ps(\mathbf{W})) \delta \epsilon \text{ for some $1 \leq l \leq k$.}
			\end{split}
		\end{equation*}
		To conclude the proof, we use the fact that 
		\begin{equation*}
			\epsilon + 2 K(\ps(\mathbf{W})) \delta\epsilon \leq \epsilon\left[1 + \frac{40K(\ps(\mathbf{W}))^2}{\Omega(\mathcal{R}^f_{\ps({\mathbf{B}})}(\mathbf{A}))^2 - M^2}\frac{L^{3/2}}{\mu^{3/2}}\right] = \hat{C}\epsilon, 
		\end{equation*}
		where $M^2$ is a constant\footnote{The reader might wonder why such a constant $M$ does not appear in \snpab\ robustness proof for linear mixing. Actually, it was implicitly chosen as $M^2 = \Omega(\mathcal{R}^f_{\ps({\mathbf{B}})}(\mathbf{A}))^2/2$.} chosen such that 
		\ckc{$M^2 = 8K(\ps(\mathbf{W}))C\epsilon \max(1, 2 K(\mathbf{W})+ C\epsilon)$}
		, which requires
        \ckc{
        \begin{equation*}
			C\epsilon < \frac{M^2}{8K(\ps(\mathbf{W}))}
		\end{equation*} 
		and
        \begin{equation*}
			C\epsilon < -K(\mathbf{W}) + \sqrt{K(\mathbf{W})^2 + \frac{M^2}{8K(\ps(\mathbf{W}))}}.
		\end{equation*} 
		}
	\end{proof}

	\begin{theorem}[Robustness of \snpab\ when applied on LQ mixings] \label{thm:robustnessLQ}
		Let 
		\begin{equation*}
			\bar{\mathbf{X}} = \ps(\mathbf{W})\mathbf{H} + \mathbf{N} \in \mathbb{R}^{m\times n}
		\end{equation*}
		be an LQ mixing satisfying Definition~\ref{ass:rLSsep} with $\alpha_{\ps(\mathbf{W})}(\mathbf{W}) > 0$ and $\beta_{\ps(\mathbf{W})}^{\text{LQ}}(\mathbf{W})$. Let $f$ satisfy Assumption~\ref{hyp:f} and $\norm{\mathbf{n}_i}{2} \leq \epsilon$ for all $i \in \lbr t \rbr$, with 
\begingroup\makeatletter\def\f@size{9}\check@mathfonts
			\begin{equation}
			\ckc{
			\begin{split}
            \hat{C}\epsilon < &\min \left(\frac{-40L^{3/2} - 16 \mu^{3/2}\hat{C}K(\mathbf{W}) + \sqrt{\left(40L^{3/2} + 16 \mu^{3/2}\hat{C}K(\mathbf{W})\right)^2 + \frac{32\mu^3 \hat{C}^2\beta^{\text{LQ}}_{\ps(\mathbf{W})}(\mathbf{W})^2}{K(\ps(\mathbf{W}))}}}{16 \mu^{3/2}\hat{C}},\right.\\
			& \frac{\mu^{3/2}\hat{C}\beta^{\text{LQ}}_{\ps(\mathbf{W})}(\mathbf{W})^2}{K(\ps(\mathbf{W}))(40L^{3/2} + 8 \hat{C} \mu^{3/2})}, \frac{\min(M,\beta^{\text{LQ}}_{\ps(\mathbf{W})}(\mathbf{W}))^2}{8K(\ps(\mathbf{W}))},\left. \sqrt{K(\mathbf{W})^2 + \frac{\min(M,\beta^{\text{LQ}}_{\ps(\mathbf{W})}(\mathbf{W}))^2}{8K(\ps(\mathbf{W}))}}-K(\mathbf{W})
			\right)
			\end{split}
			}
			\end{equation}\endgroup
		where
		\begin{equation*}
			\small
			\hat{C} = 1 + \frac{40K(\ps(\mathbf{W}))^2}{\beta^{\text{LQ}}_{\ps(\mathbf{W})}(\mathbf{W})^2 - M^2}\frac{L^{3/2}}{\mu^{3/2}}, 
		\end{equation*}
		with $M$ a constant\footnote{Despite a slight loss of generality, the reader can think of $M^2 = \Omega(\mathcal{R}^f_{\ps({\mathbf{B}})}(\mathbf{A}))^2/2$ to create a link with the linear case.} (the smaller $M$, the more restrictive the condition on the noise, but the better the estimation). 
		Furthermore, let us assume that at each iteration of \snpab\ the following condition is fulfilled:
		\begin{equation*}
			\small
			K\left(\mathcal{R}^f_{\ps{(\bar{\mathbf{B}})}}(\mathbf{A})\right) \geq 2G K\left(\mathcal{R}^f_{\ps{(\bar{\mathbf{B}})}}\left((\mathbf{b}_i)_{i\in \lbr s\rbr},(\mathbf{a}_i\odot \mathbf{a}_j)_{\substack{i\leq j \\ i\in \lbr k\rbr \\ j\in \lbr k\rbr}},(\mathbf{b}_i\odot \mathbf{b}_j)_{\substack{i\leq j \\ i\in \lbr s\rbr \\ j\in \lbr s\rbr}},(\mathbf{a}_i\odot \mathbf{b}_j)_{\substack{i\in \lbr k\rbr \\j \in \lbr s\rbr}}\right)\right), 
		\end{equation*}
		where $\mathbf{B}$ contains the columns of $\mathbf{W}$ already extracted by \snpab\ and $\bar{\mathbf{B}}$ the corresponding columns with noise, $\mathbf{A}$ contains the remaining columns of $\mathbf{W}$ still-to-be extracted, and $L < \mu G^2$ is a constant. Then, \snpab\ identifies in $r$ steps the columns of $\mathbf{W}$ up to an error $\hat{C}\epsilon$. Precisely, denoting $\mathcal{K}$ the index set extracted by \snpab\ after $r$ steps, there exists a permutation $\pi$ of $\lbr r \rbr$ such that: 
		\begin{equation*}
			\max_{1 \leq j \leq r} \norm{\bar{\mathbf{x}}_{\mathcal{K}(j)} - \mathbf{w}_{\pi(j)}}{2} \leq \hat{C}\epsilon . 
		\end{equation*}
	\end{theorem}
	
	\begin{proof}
		The result follows by induction. 
		\begin{itemize}
		
		\item In the initialization step, $\mathbf{B}$ is the empty matrix. 
		
			\item The induction step is given by Theorem~\ref{th:indStep}: the $\mathbf{B}$ matrix corresponds to the columns of $\mathbf{W}$ extracted so far by \snpab, while the columns of $\mathbf{A}$ the ones still-to-be extracted. 
			Letting 
			\begin{equation*}
				\hat{C} = 1 + \frac{40K(\ps(\mathbf{W}))^2}{\beta^{\text{LQ}}_{\ps(\mathbf{W})}(\mathbf{W})^2 - M^2}\frac{L^{3/2}}{\mu^{3/2}} 
			\end{equation*}
			in Theorem~\ref{th:indStep}, 			we obtain that if the already extracted columns are at a distance at most $\hat{C}\epsilon$ of some columns of $\mathbf{W}$ (more exactly, $\norm{\bar{\mathbf{B}} - \mathbf{B}}{2} \leq \hat{C}\epsilon$), then the next extracted column will be at distance at most $\hat{C}\epsilon$ from a new column of $\mathbf{W}$ (that is, a column of $\mathbf{A}$), provided that $\epsilon$ is small enough.

		\end{itemize}
		
	\end{proof}

		\subsection{Proof for the Brute Force algorithm (BF)} \label{sec:proofBF}


	\begin{definition}[LQ-robust loner]\label{def:robust_loner_full}
	Let $\bar{\mathbf{X}}$ be an LQ mixing satisfying Definition~\ref{ass:rLSsep} and $j\in \lbr t \rbr$. 
Let us denote $\mathcal{L}$ the set of indices $k \in \lbr t \rbr$ such that 
		\begin{equation}
		\begin{split}
		\small
		    &f(\bar{\mathbf{x}}_k - \bar{\mathbf{x}}_j) 
		    \leq d = \epsilon V, 
		\end{split}
		\label{eq:d}
		\end{equation}
		where 
		\begin{align*} 
	V & = 	\frac{L^2}{\mu \alpha_{\pstr(\mathbf{W})}(\mathbf{W})^2}K(\mathbf{X})^2 Y^3 \left[ \epsilon (2K(\mathbf{X})\text{$+$}\epsilon)^2 / (2Y) + K(\mathbf{X})  + \max(1,2K(\mathbf{X})\text{$+$}\epsilon)
	(\epsilon\text{$+$}K(\mathbf{X}))\right]\\
		& 
		\qquad + \frac{3}{2}L(4K(\mathbf{X})+\epsilon), 
			\end{align*}  
			with $Y = 1+\max(1,K(\mathbf{X}))$. 
		We call $\bar{\mathbf{x}}_j$ a robust loner if 
		\[
		\min_{\mathbf{h^*}\in \Delta}f(\bar{\mathbf{x}}_j - \ps(\bar{\mathbf{X}}_{\lbr t \rbr \setminus \mathcal{L}})\mathbf{h^*}) > \frac{L}{2}\epsilon^2(1+\max(1,2K(\mathbf{X}) + \epsilon))^2. 
		\] 
	\end{definition}

	\begin{definition}[Canonical columns]    
	    \ckc{Let $\bar{\mathbf{X}}$ be $r$-LQ near-separable; see  Definition~\ref{ass:rLSsep}. 
	    We call canonical columns (associated to $i \in \lbr r \rbr$), the columns $\bar{\mathbf{X}}_{k(i)}$, $k(i) \in \lbr t \rbr$, of $\bar{\mathbf{X}}$ such that all the columns of $\mathbf{H}_{k(i)}$ have a single nonzero entry located in their $i$th row.\\
	    Note that by definition of near-separability, there exists at least a canonical column for all $i \in \lbr r \rbr$. Moreover, all the canonical columns $\bar{\mathbf{x}}_{k(i)}$ associated to $i \in \lbr r \rbr$ satisfy $f(\bar{\mathbf{x}}_{k(i)} - \mathbf{w}_i) < \frac{L}{2}\epsilon^2$.  }
	\end{definition}

	\begin{lemma}\label{clm:58} \ckc{Let $\bar{\mathbf{X}}$ be $r$-LQ near-separable (Definition~\ref{ass:rLSsep}). 
		Considering all the canonical columns, written as $\bar{\mathbf{X}}_{\mathcal{K}}$ (that is, the canonical columns associated to all $i \in \lbr r \rbr$), every column $\bar{\mathbf{x}}_j$ of $\bar{\mathbf{X}}$ is such that}
		\begin{equation*}
			\min_{\mathbf{h^*}\in \Delta}f\left(\bar{\mathbf{x}}_j - \ps\left(\bar{\mathbf{X}}_{\mathcal{K}}\right)\mathbf{h^*}\right) \leq \frac{L}{2}\epsilon^2(1+\max(1,2K(\mathbf{X}) + \epsilon))^2.
		\end{equation*}
	\end{lemma}
	
	\begin{proof}
		For all $\mathbf{h\in \Delta}$, we have:
		\begin{equation*}
			\begin{split}
				f\left(\bar{\mathbf{x}}_j - \ps\left(\bar{\mathbf{X}}_{\mathcal{K}}\right)\mathbf{h}\right) 
				&\leq \frac{L}{2}\norm{\bar{\mathbf{x}}_j - \ps{\left(\bar{\mathbf{X}}_{\mathcal{K}}\right)}\mathbf{h}}{2}^2\\
				&\leq \frac{L}{2} \left(\norm{\bar{\mathbf{x}}_j -{\mathbf{x}}_j}{2} + \norm{{\mathbf{x}}_j - \ps{\left({\mathbf{X}}_{\mathcal{K}}\right)}\mathbf{h}}{2} \right. 
			\\ & \qquad \qquad \left. + \norm{\left(\ps{\left({\mathbf{X}}_{\mathcal{K}}\right)} - \ps{\left({\bar{\mathbf{X}}}_{\mathcal{K}}\right)}\right)\mathbf{h}}{2} \right)^2.
			\end{split}
		\end{equation*}
		Moreover 
		\begin{equation*}
			\begin{split}
				\norm{\bar{\mathbf{x}}_j -{\mathbf{x}}_j}{2} + &\norm{{\mathbf{x}}_j - \ps{\left({\mathbf{X}}_{\mathcal{K}}\right)}\mathbf{h}}{2} + \norm{\left(\ps{\left({\mathbf{X}}_{\mathcal{K}}\right)} - \ps{\left({\bar{\mathbf{X}}}_{\mathcal{K}}\right)}\right)\mathbf{h}}{2} \\
				&\leq \epsilon + \norm{{\mathbf{x}}_j - \ps{\left({\mathbf{X}}_{\mathcal{K}}\right)}\mathbf{h}}{2} + \max(\epsilon, 2K(\mathbf{X})\epsilon + \epsilon^2).
			\end{split}
		\end{equation*}
		Thus,
		\begin{equation*}
			\min_{\mathbf{h^*}\in \Delta}f\left(\bar{\mathbf{x}}_j - \ps\left(\bar{\mathbf{X}}_{\mathcal{K}}\right)\mathbf{h^*}\right) \leq \frac{L}{2}\epsilon^2(1+\max(1,2K(\mathbf{X}) + \epsilon))^2 . 
		\end{equation*}
	\end{proof}
	
	\begin{lemma}\label{lem:lem3} Let $\bar{\mathbf{X}} = \ps(\mathbf{W})\mathbf{H + N}$ be $r$-LQ near-separable (Definition~\ref{ass:rLSsep}). Let us denote $\bar{\mathbf{x}}_{k(i)}$ any robust loner associated to $i \in \lbr r \rbr$.
	If \mbox{$f(\bar{\mathbf{x}}_j - \mathbf{w}_i) > d + \epsilon L (2K(\mathbf{X})+ \epsilon)$} for some $j \in \lbr t\rbr$, then $f(\bar{\mathbf{x}}_j - \bar{\mathbf{x}}_{k(i)}) > d$. 
	\end{lemma}
	\begin{proof} We have 
		\begin{equation*}
			\begin{split}
				f(\bar{\mathbf{x}}_j - \bar{\mathbf{x}}_{k(i)}) 
				&= f(\bar{\mathbf{x}}_j - \mathbf{w}_i - \mathbf{n}_{k(i)})\\
				&\geq f(\bar{\mathbf{x}}_j - \mathbf{w}_i) - \epsilon K(\bar{\mathbf{x}}_j - \mathbf{w}_i)L\\
				&> d + \epsilon L(2K(\mathbf{X})+\epsilon) - \epsilon(2K(\mathbf{X})+\epsilon)L\\
				&= d . 
			\end{split}
		\end{equation*}
	\end{proof}
	
	\begin{lemma}\label{clm:59} \ckc{Let $\bar{\mathbf{X}} = \ps(\mathbf{W})\mathbf{H + N}$ be $r$-LQ near-separable (Definition~\ref{ass:rLSsep}).  
		If a column $\bar{\mathbf{x}}_j$ is a robust loner, then there is an index $i \in \lbr r \rbr$ such that}
		\begin{equation*}
			f(\bar{\mathbf{x}}_j - \mathbf{w}_i) \leq d + \epsilon L (2K(\mathbf{X})+\epsilon).
		\end{equation*}
	\end{lemma}
	\begin{proof}
		The result is proved by contraposition. We want to show that 
		\begin{equation*}
			\text{If } \forall i \in \lbr r \rbr, f(\bar{\mathbf{x}}_j - \mathbf{w}_i) > d + \epsilon(2K(\mathbf{X})+\epsilon) , 
			\; \text{ then } 
			\; 
			\bar{\mathbf{x}}_j \text{ is not a robust loner.} 
		\end{equation*}
		If $\bar{\mathbf{x}}_j$ is such that $\forall i \in \lbr r \rbr , f(\bar{\mathbf{x}}_j - \mathbf{w}_i) > d + \epsilon L(2K(\mathbf{X})+\epsilon)$, then $\forall i \in \lbr r \rbr , f(\bar{\mathbf{x}}_j - \bar{\mathbf{x}}_{k(i)}) > d$, \ckc{with $\bar{\mathbf{x}}_{k(i)}$ the canonical columns associated to $i$}; see lemma~\ref{lem:lem3}. As such,\ckc{ denoting $\bar{\mathbf{X}}_{\mathcal{K}}$ all the canonical columns,} to be a robust loner $\bar{\mathbf{x}}_j$ must  satisfy 
		\begin{equation*}
			\min_{\mathbf{h^*} \in \Delta}f\left(\bar{\mathbf{x}}_j - \ps(\bar{\mathbf{X}}_{\mathcal{K}})\mathbf{h^*}\right) > \frac{L}{2}\epsilon^2(1+\max(1,2K(\mathbf{X}) + \epsilon))^2.
		\end{equation*}
		This is however not the case according to Lemma \ref{clm:58}. Thus, by definition, $\bar{\mathbf{x}}_j$ is not a robust loner.
	\end{proof}
	
	\begin{lemma}\label{lem:lem1bis}
	\ckc{Let $\bar{\mathbf{X}} = \ps(\mathbf{W})\mathbf{H + N}$ be $r$-LQ near-separable (Definition~\ref{ass:rLSsep}), $i \in \lbr r \rbr$ and $\bar{\mathbf{x}}_{k(i)}$ a canonical column associated to $i$.  
	If, for some $k \in \lbr t \rbr$, $f(\bar{\mathbf{x}}_k - \mathbf{w}_i) \leq d - \frac{3}{2}\epsilon L(2 K(\mathbf{X})+\epsilon)$, then $f(\bar{\mathbf{x}}_k - \bar{\mathbf{x}}_{k(i)}) \leq d$.}
	\end{lemma}
	\begin{proof} We have 
		\begin{equation*}
			\begin{split}
				f(\bar{\mathbf{x}}_k - \bar{\mathbf{x}}_{k(i)}) & =f(\bar{\mathbf{x}}_k-\mathbf{w}_i - \mathbf{n}_{k(i)})\\
				&\leq f(\bar{\mathbf{x}}_k - \mathbf{w}_i) + \frac{3}{2}\epsilon L K(\bar{\mathbf{x}}_k - \mathbf{w}_i)\\
				& = d - \frac{3}{2}\epsilon L(2K(\mathbf{X})+\epsilon) + \frac{3}{2}\epsilon L(2K(\mathbf{X})+\epsilon)\\
				& = d
			\end{split}
		\end{equation*}
	\end{proof}

	\begin{lemma}\label{lem:C12bis2}\ckc{Let $\bar{\mathbf{X}} = \ps(\mathbf{W})\mathbf{H + N}$ be $r$-LQ near-separable (Definition~\ref{ass:rLSsep}).
		All the columns $\bar{\mathbf{x}}_j$, $j \in \lbr t \rbr$, with ${\mathbf{x}}_j = \sum_{k=1}^{r} a_{kj}\mathbf{w}_k + \sum_{l=1}^{r}\sum_{q=l+1}^{r} b_{jlq}\mathbf{w}_l \odot \mathbf{w}_q$ and $a_{ji} > 1 - \sqrt{\frac{2d - 3\epsilon L(4K(\mathbf{X})+\epsilon)}{LK(\mathbf{X})^2[1+\max(1,K(\mathbf{X}))]^2}}$ for some $i \in \lbr r \rbr$ satisfy $f(\bar{\mathbf{x}}_j - \mathbf{w}_i) \leq d - \frac{3}{2}\epsilon L (2K(\mathbf{X})+\epsilon)$.}
	\end{lemma}
	
	\begin{proof}
		We want to prove that 
		\begin{equation*}
		\small
			\left\{\bar{\mathbf{x}}_j \left| \ a_{ij} > 1 - \sqrt{\frac{2d - 3\epsilon L(4K(\mathbf{X})+\epsilon)}{LK(\mathbf{X})^2[1+\max(1,K(\mathbf{X}))]^2}}\right.\right\} \subseteq \left\{\bar{\mathbf{x}}_j \left|\ f(\bar{\mathbf{x}}_j - \mathbf{w}_i) \leq d - \frac{3}{2}\epsilon L (2K(\mathbf{X})+\epsilon)\right.\right\} 
		\end{equation*}
		Let us consider a column $\bar{\mathbf{x}}_j \in \left\{\bar{\mathbf{x}}_j \left| \ a_{ij} > 1 - \sqrt{\frac{2d - 3\epsilon L(4K(\mathbf{X})+\epsilon)}{LK(\mathbf{X})^2[1+\max(1,K(\mathbf{X}))]^2}}\right.\right\}$. We have (looking at the noiseless version $\mathbf{x}_j$ of $\bar{\mathbf{x}}_j$) that 
		\begin{equation*}
			\begin{split}
				f(\mathbf{x}_j - \mathbf{w}_j) 
				&\leq \frac{L}{2}\norm{\mathbf{x}_j - \mathbf{w}_i}{2}^2\\
				&= \frac{L}{2}\norm{\mathbf{w}_i - a_{ij}\mathbf{w}_i - \sum_{k\neq i}^r a_{kj}\mathbf{w}_k - \sum_{l = 1}^r\sum_{q = l}^r b_{jlq}\mathbf{w}_l \odot \mathbf{w}_q}{2}^2\\
				&= \frac{L}{2}(1-a_{ij})^2\norm{\mathbf{w}_i - \frac{1}{1-a_{ij}}\left(\sum_{k\neq i}^r a_{kj}\mathbf{w}_k + \sum_{l = 1}^r\sum_{q = l}^r b_{jlq}\mathbf{w}_l \odot \mathbf{w}_q \right)}{2}^2. 
			\end{split}
		\end{equation*}
		Moreover 
		\begin{equation*}
			\begin{split}
				&\norm{\mathbf{w}_i - \frac{1}{1-a_{ij}}\left(\sum_{k\neq i}^r a_{kj}\mathbf{w}_k + \sum_{l = 1}^r\sum_{q = l}^r b_{jlq}\mathbf{w}_l \odot \mathbf{w}_q \right)}{2}\\
				&\leq \norm{\mathbf{w}_i}{2} + \norm{\frac{1}{1-a_{ij}}\left(\sum_{k\neq i}^r a_{kj}\mathbf{w}_k + \sum_{l = 1}^r\sum_{q = l}^r b_{jlq}\mathbf{w}_l \odot \mathbf{w}_q \right)}{2}\\
				&\leq K(\mathbf{X}) + \max(K(\mathbf{X}),K(\mathbf{X})^2),
			\end{split}
		\end{equation*}
		where the second inequality is obtained using $\frac{1}{1-a_{ij}} \left(\sum_{k\neq i}^r a_{kj} + \sum_{l = 1}^r\sum_{q = l}^r b_{jlq} \right) = 1$. Therefore, 
		\begin{equation*}
			f(\mathbf{x}_j - \mathbf{w}_j) \leq \frac{L}{2}(1-a_{ij})^2K(\mathbf{X})^2[1+\max(1,K(\mathbf{X}))]^2.
		\end{equation*}
		To conclude the proof, let us consider the noisy $\bar{\mathbf{x}}_j$, we have 
		\begin{equation*}
			\begin{split}
				f(\bar{\mathbf{x}}_j - \mathbf{w}_i) 
				&= f(\mathbf{x}_j + \mathbf{n}_j - \mathbf{w}_i) \\
				&\leq f(\mathbf{x}_j - \mathbf{w}_i) + \frac{3}{2}L\epsilon K(\mathbf{x}_j - \mathbf{w}_i)\\
				&\leq \frac{L}{2}(1-a_{ij})^2K(\mathbf{X})^2[1+\max(1,K(\mathbf{X}))]^2 + 3L\epsilon K(\mathbf{X})\\
				&\leq d - \frac{3}{2}\epsilon L(2K(\mathbf{X})+\epsilon). 
			\end{split}
		\end{equation*}
	\end{proof}
	
		\begin{lemma}\label{lem:4bis}\ckc{Let $\bar{\mathbf{X}}$ be $r$-LQ near-separable (Definition~\ref{ass:rLSsep}), $\mathcal{J} \subseteq \lbr t \rbr$ and $\bar{\mathbf{x}}_{k(i)}$ a canonical column associated to $i$.
	If 
	\begin{equation*}
	    \text{for some } i \in \lbr r \rbr, f(\mathbf{w}_i - \ps(\mathbf{X}_\mathcal{J})\mathbf{h}) \geq \frac{\mu}{2L}\frac{2d - 3\epsilon L(4K(\mathbf{X})+\epsilon)}{K(\mathbf{X})^2[1+\max(1,K(\mathbf{X}))]^2}\alpha_{\pstr(\mathbf{W})}(\mathbf{W})^2,
	\end{equation*}
	then
	\begin{equation*}
	    f(\bar{\mathbf{x}}_{k(i)} - \ps(\bar{\mathbf{X}}_\mathcal{J})\mathbf{h}) \geq \frac{L}{2}\epsilon^2(3+\epsilon)^2.
	\end{equation*}
	}
	\end{lemma}
	
	\begin{proof} We have 
		\begin{equation*}
			\begin{split}
				f(\bar{\mathbf{x}}_{k(i)} - \ps(\bar{\mathbf{X}}_\mathcal{J})\mathbf{h})
				&\geq f(\bar{\mathbf{x}}_{k(i)} - \ps(\mathbf{X}_\mathcal{J})\mathbf{h}) - L\epsilon \max(1,2K(\mathbf{X})+\epsilon)K(\bar{\mathbf{x}}_{k(i)} - \ps(\mathbf{X}_\mathcal{J})\mathbf{h})\\
				&\geq f(\mathbf{w}_i - \ps(\mathbf{X}_\mathcal{J})\mathbf{h}) - L\epsilon K(\mathbf{w}_i - \ps(\mathbf{X}_\mathcal{J})\mathbf{h}) \\
				& \quad - L\epsilon \max(1,2K(\mathbf{X})+\epsilon)\left(\epsilon + K(\mathbf{X})[1+\max(1,K(\mathbf{X}))]\right)\\
				&\geq f(\mathbf{w}_i - \ps(\mathbf{X}_\mathcal{J})\mathbf{h}) - L\epsilon K(\mathbf{X})[1+\max(1,K(\mathbf{X}))] \\
				& \quad - L\epsilon \max(1,2K(\mathbf{X})+\epsilon)\left(\epsilon + K(\mathbf{X})[1+\max(1,K(\mathbf{X}))]\right)\\
				&\geq \frac{L}{2}\epsilon^2 (2 + K(\mathbf{X})+  \epsilon)^2 . 
			\end{split}
		\end{equation*}
	\end{proof}
	
	\begin{lemma}[Extension of \cite{Arora2016} -- Claim 5.10]\label{clm:510}
		All canonical columns are robust-loners. 
	\end{lemma}
	\begin{proof}
		Let $\bar{\mathbf{x}}_{k(i)}$ be a canonical column associated to $i \in \lbr r \rbr$: we have that $f(\bar{\mathbf{x}}_{k(i)} - \mathbf{w}_i) \leq \frac{L}{2}\epsilon^2$. To check whether $\bar{\mathbf{x}}_{k(i)}$ is a robust-loner, we must leave out of consideration the columns $\bar{\mathbf{x}}_k$ such that $f(\bar{\mathbf{x}}_k - \bar{\mathbf{x}}_{k(i)}) \leq d$. This particularly excludes all the columns satisfying 
		\begin{equation*}
			f(\bar{\mathbf{x}}_k - \mathbf{w}_i) \leq d - \frac{3}{2}\epsilon L (2 K(\mathbf{X})+\epsilon), 
		\end{equation*} 
	see Lemma~\ref{lem:lem1bis}. In particular, only the columns $\bar{\mathbf{x}}_j$, $j \in \mathcal{J}$ with ${\mathbf{x}}_j = \sum_{k=1}^{r} a_{kj}\mathbf{w}_k + \sum_{l=1}^{r}\sum_{q=l+1}^{r} b_{jlq}\mathbf{w}_l \odot \mathbf{w}_q$ and $\mathbf{a}_{ij} \leq 1 - \sqrt{\frac{2d - 3\epsilon L(4K(\mathbf{X})+\epsilon)}{LK(\mathbf{X})^2[1+\max(1,K(\mathbf{X}))]^2}}$ are taken into account (Lemma \ref{lem:C12bis2}).\\
	Since the $\ell_2$ distance of $\mathbf{w}_i$ to the convex hull of $\pstr(\mathbf{W})$ is at least $\alpha_{\pstr(\mathbf{W})}(\mathbf{W})$, the distance between $\mathbf{w}_i$ and the convex hull of the retained $\mathbf{X}_\mathcal{J}$ columns and their quadratic product is at least $\sqrt{\frac{2d - 3\epsilon L(4K(\mathbf{X})+\epsilon)}{LK(\mathbf{X})^2[1+\max(1,K(\mathbf{X}))]^2}}\alpha_{\pstr(\mathbf{W})}(\mathbf{W})$. As for all $\mathbf{h} \in \Delta$
	\begin{equation*}
		f(\mathbf{w}_i - \ps(\mathbf{X}_\mathcal{J})\mathbf{h}) \geq \frac{\mu}{2} \norm{\mathbf{w}_i - \ps(\mathbf{X}_\mathcal{J})}{2}^2,
	\end{equation*}
	we obtain
	\begin{equation*}
		f(\mathbf{w}_i - \ps(\mathbf{X}_\mathcal{J})\mathbf{h}) \geq \frac{\mu}{2L}\frac{2d - 3\epsilon L(4K(\mathbf{X})+\epsilon)}{K(\mathbf{X})^2[1+\max(1,K(\mathbf{X}))]^2}\alpha_{\pstr(\mathbf{W})}(\mathbf{W})^2. 
	\end{equation*}
	Thus, $f(\bar{\mathbf{x}}_{k(i)} - \ps(\bar{\mathbf{X}}_\mathcal{J})\mathbf{h}) \geq \frac{L}{2}\epsilon^2(3+\epsilon)^2$ (see Lemma~\ref{lem:4bis}) and hence  $\bar{\mathbf{x}}_{k(i)}$ is a robust loner. 
	\end{proof}

	\begin{theorem}[Robustness of BF when applied on LQ mixings] \label{thm:318}
		Let $\mathbf{X} = \ps(\mathbf{W}) + \mathbf{N}$ satisfying Definition~\ref{ass:rLSsep} with $\norm{\mathbf{n}_i}{1} \leq \epsilon$ for  $i \in \lbr n\rbr$. 
		Let also $\epsilon$ satisfy 
		\begin{equation*}
			4\sqrt{\frac{2}{\mu}(d + \epsilon L (2K(\mathbf{X})+\epsilon))} < \alpha_{\mathbf{W}}(\mathbf{W}).
		\end{equation*}{}
		Then, BF with $f$ satisfying Assumption~\ref{hyp:f} identifies the columns of $\mathbf{W}$ up to a $\ell_2$ error of 
		\begin{equation*}
			\sqrt{\frac{2}{\mu}\left(d + \epsilon L (2K(\mathbf{X}) + \epsilon)\right)}.
		\end{equation*}
	\end{theorem}
	\begin{proof}
		By Lemma~\ref{clm:510}, 
		all canonical columns are robust loners. Moreover, Lemma~\ref{clm:59} shows that every robust-loner $\bar{\mathbf{x}}_j$ satisfies $f(\bar{\mathbf{x}}_j - \mathbf{w}_i) \leq d + \epsilon L (2K(\mathbf{X}) + \epsilon)$ for some $i \in \lbr r \rbr$.  As such, identifying the robust loners enables to approximately identify the columns of $\mathbf{W}$. 
		Since several robust-loners can correspond to the same source, we need to apply a clustering step to regroup them. This is done easily, as two robust loners $\bar{\mathbf{x}}_j$ and $\bar{\mathbf{x}}_k$ correspond to the same source if and only if they satisfy $\norm{\bar{\mathbf{x}}_j - \bar{\mathbf{x}}_k}{2} \leq 2\sqrt{\frac{2}{\mu}(d + \epsilon L (2K(\mathbf{X})+\epsilon))}$. In fact, 
		\begin{itemize}
			\item If two robust loners $\bar{\mathbf{x}}_j$ and $\bar{\mathbf{x}}_k$ correspond to the same source (in the sense that $f(\bar{\mathbf{x}}_j - \mathbf{w}_i) \leq d + \epsilon L (2K(\mathbf{X}) + \epsilon)$ and $f(\bar{\mathbf{x}}_k - \mathbf{w}_i) \leq d + \epsilon L (2K(\mathbf{X}) + \epsilon)$), they must satisfy 
			\begin{equation*}
				\begin{split}
					\norm{\bar{\mathbf{x}}_j - \bar{\mathbf{x}}_k}{2}
					&= \norm{\bar{\mathbf{x}}_j - \mathbf{w}_i + \mathbf{w}_i - \bar{\mathbf{x}}_k}{2}\\
					&\leq \norm{\bar{\mathbf{x}}_j - \mathbf{w}_i}{2} + \norm{\bar{\mathbf{x}}_k - \mathbf{w}_i}{2}\\
					&\leq \sqrt{\frac{2}{\mu}}\left[\sqrt{f(\bar{\mathbf{x}}_j - \mathbf{w}_i)} + \sqrt{f(\bar{\mathbf{x}}_k - \mathbf{w}_i)}\right]\\
					&\leq 2\sqrt{\frac{2}{\mu}(d + \epsilon L (2K(\mathbf{X})+\epsilon))}. 
				\end{split}
			\end{equation*}
			
			\item If two robust loners satisfy $\norm{\bar{\mathbf{x}}_j - \bar{\mathbf{x}}_k}{2} \leq 2\sqrt{\frac{2}{\mu}(d + \epsilon L (2K(\mathbf{X})+\epsilon))}$, they must correspond to the same source $\mathbf{w}_i$. This follows by contradiction: suppose that $\bar{\mathbf{x}}_j$ corresponds to a source $\mathbf{w}_i$ ($f(\bar{\mathbf{x}}_j - \mathbf{w}_i) \leq d + \epsilon L (2K(\mathbf{X}) + \epsilon)$) and $\bar{\mathbf{x}}_j$ to another source $\mathbf{w}_l$, $l \neq i$ ($f(\bar{\mathbf{x}}_j - \mathbf{w}_l) \leq d + \epsilon L (2K(\mathbf{X}) + \epsilon)$). Then we obtain that:
			\begin{equation*}
				\norm{\bar{\mathbf{x}}_j - \mathbf{w}_i}{2} \leq  \sqrt{\frac{2}{\mu}f(\bar{\mathbf{x}}_j - \mathbf{w}_i)} \leq \sqrt{\frac{2}{\mu}\left(d + \epsilon L (2K(\mathbf{X}) + \epsilon)\right)} , \text{ and } 
			\end{equation*}
			\begin{equation*}
				\norm{\bar{\mathbf{x}}_k - \mathbf{w}_l}{2} \leq  \sqrt{\frac{2}{\mu}f(\bar{\mathbf{x}}_k - \mathbf{w}_l)} \leq \sqrt{\frac{2}{\mu}\left(d + \epsilon L (2K(\mathbf{X}) + \epsilon)\right)},
			\end{equation*}
			from which it can be deduced that 
			\begin{equation*}
				\begin{split}
					\norm{\bar{\mathbf{x}}_j - \bar{\mathbf{x}}_k}{2}
					&= \norm{\bar{\mathbf{x}}_j - \mathbf{w}_i + \mathbf{w}_i - \mathbf{w}_l + \mathbf{w}_l + \bar{\mathbf{x}}_k}{2}\\
					&\geq \norm{\mathbf{w}_i - \mathbf{w}_l}{2} - \norm{\bar{\mathbf{x}}_j - \mathbf{w}_i}{2} - \norm{\bar{\mathbf{x}}_k - \mathbf{w}_l}{2}\\
					&\geq \alpha_{\mathbf{W}}(\mathbf{W}) - 2\sqrt{\frac{2}{\mu}\left(d + \epsilon L (2K(\mathbf{X}) + \epsilon)\right)}\\
					&> 4\sqrt{\frac{2}{\mu}\left(d + \epsilon L (2K(\mathbf{X}) + \epsilon)\right)} - 2\sqrt{\frac{2}{\mu}\left(d + \epsilon L (2K(\mathbf{X}) + \epsilon)\right)} \\
					&= 2\sqrt{\frac{2}{\mu}\left(d + \epsilon L (2K(\mathbf{X}) + \epsilon)\right)},
				\end{split}
			\end{equation*}
			which is a contradiction.
		\end{itemize}
		Therefore, once the robust-loners are found and the clustering described above performed, each source can be identified by picking a point from each cluster. The $\ell_2$-norm error is then at most $\sqrt{\frac{2}{\mu}\left(d + \epsilon L (2K(\mathbf{X}) + \epsilon)\right)}$.
	\end{proof}

	\bibliographystyle{siamplain}
	\bibliography{strings_all_ref,bibliographie}
	
\end{document}